\documentclass[a4paper,UKenglish,cleveref, autoref, thm-restate]{lipics-v2021}
\hideLIPIcs
\usepackage{graphicx} 
\usepackage{amsmath,amssymb}
\usepackage{amsthm}
\usepackage{thmtools} 
\usepackage{thm-restate}

\usepackage{hyperref}  
\usepackage{cite}      

\usepackage{todonotes}
\usepackage{xcolor}

\usepackage[normalem]{ulem}
\usetikzlibrary{decorations.pathreplacing}
\usetikzlibrary{decorations.pathmorphing}
\usetikzlibrary{arrows.meta}

\newcommand{\RR}{\mathbb{R}}

\newcommand{\OO}{O}

\let\oldsim\sim 
\renewcommand{\sim}{{\oldsim}}

\DeclareMathOperator{\opt}{OPT}

\DeclareMathOperator {\poly}{poly}

\newcommand{\NP}{NP}

\newcommand{\tw}{\ensuremath{\textit{w}}}
\newcommand{\dm}{\ensuremath{\Delta}}

\newcommand{\tn}{\ensuremath{\sigma}}

\title{Parameterized Algorithms for the Drone Delivery Problem}

\author{Simon Bartlmae}{University of Copenhagen, Denmark}{simon.bartlmae7@gmail.com}{https://orcid.org/0009-0009-6953-8347}{}
\author{Andreas Hene}{Institute of Computer Science, University of Bonn, Germany}{hene@uni-bonn.de}{https://orcid.org/0009-0002-0336-374X}{}
\author{Joshua Könen}{Institute of Computer Science, University of Bonn, Germany}{koenen@cs.uni-bonn.de}{https://orcid.org/0000-0003-4245-4812}{}
\author{Heiko Röglin}{Institute of Computer Science, University of Bonn, Germany}{roeglin@cs.uni-bonn.de}{https://orcid.org/0009-0006-8438-3986}{}

\authorrunning{S. Bartlmae, A. Hene, J. Könen, H. Röglin}


\pdfoutput=1
\nolinenumbers
\EventLogo{}

\relatedversion{A version of this paper appeared in the proceedings of ISAAC 2025.}


\begin{document}
\maketitle

\begin{abstract}
Timely delivery and optimal routing remain fundamental challenges in the modern logistics industry. Building on prior work that considers single-package delivery across networks using multiple types of collaborative agents with restricted movement areas (e.g., drones or trucks), we examine the complexity of the problem under structural and operational constraints. Our focus is on minimizing total delivery time by coordinating agents that differ in speed and movement range across a graph.
This problem formulation aligns with the recently proposed \textit{Drone Delivery Problem with respect to delivery time} (DDT), introduced by Erlebach et al.~[ISAAC 2022]. 

We first resolve an open question posed by Erlebach et al.~[ISAAC 2022] by showing that even when the delivery network is a path graph, DDT admits no polynomial-time approximation within any polynomially encodable factor $a(n)$, unless P=NP.
Additionally, we identify the \textit{intersection graph} of the agents, where nodes represent agents and edges indicate an overlap of the movement areas of two agents, as an important structural concept. 
For path graphs, we show that DDT becomes tractable when parameterized by the treewidth $w$ of the intersection graph, and we present an exact FPT algorithm with running time $f(w)\cdot\text{poly}(n,k)$, for some computable function $f$.
For general graphs, we give an FPT algorithm with running time $f(\Delta,w)\cdot\text{poly}(n,k)$, where $\Delta$ is the maximum degree of the intersection graph.
In the special case where the intersection graph is a tree, we provide a simple polynomial-time algorithm.
\end{abstract}
\setcounter{page}{0}

\newpage
\section{Introduction}

Drone-based package delivery has the potential to transform the last-mile segment. Recent studies foretell that utilizing unmanned aerial vehicles (i.e., drones) can have a great sustainability impact while meeting increasing demands of customers \cite{GOODCHILD201858,drones9060413}. Additionally, it is observed that in the U.S.\ approximately 10--15\% of package deliveries are delayed due to increasing loads \cite{BULDEORAI2021100753}. 
According to ElSayed et al.~\cite{ELSAYED202437}, recent surveys indicate that a vast majority of people prefer a same-day delivery, which proves challenging for the logistics industry. 
Consequently, large companies focus more and more on developing suitable logistics models incorporating drones. Apart from monetary benefits, the utilization of drones can be very helpful in supplying medicine or food in areas struck by war or natural disasters \cite{bamburry2015drones}.

In this work, we study a cooperative delivery setting recently introduced by Erlebach et al.\ \cite{erlebach22} in which drones can hand over packages and work together to achieve the best possible delivery schedule. Each agent operates within a restricted movement area on the graph and is assumed to have unlimited battery capacity. In general, this model is inspired by real-world constraints such as restricted air spaces or environmental circumstances. 
Restricted movement areas also capture settings with different kinds of agents, not necessarily only drones -- different types of agents might be eligible to traverse certain parts of a graph while others might not. Package handover between drones is already implemented in industry; companies such as IBM have developed and patented parcel transfer methods since 2017 \cite{dronetransfer}. 

Our focus is on \textit{minimizing total delivery time}, a problem referred to as the \textit{Drone Delivery Problem with respect to time} (DDT) \cite{erlebach22, bartlmae25}. 
We study general graphs as well as the special case of path graphs. For path graphs, we first establish a hardness result, followed by the first fixed-parameter tractable algorithms for the path as well as the general settings. 
For some parameter~$k$, an algorithm is fixed-parameter tractable (FPT) if it runs in time~$f(k) \cdot \text{poly}(n)$, where~$f$ is a computable function; such algorithms are efficient for small parameters and useful for fine-grained complexity analysis.
We analyze the complexity of DDT using the \emph{intersection graph} of the agents, where nodes represent agents and edges indicate overlapping movement areas. The \emph{treewidth} and maximum degree $\dm$ of this graph serve as natural parameters for characterizing instance complexity -- especially since DDT is already \NP{}-hard on path graphs. These parameters reflect the extent of agent overlap and potential interference, and are likely to be small in practical settings.

\subparagraph*{Related Work.} For the most part, research regarding collaborative drone delivery focuses on two main objectives: minimizing fuel consumption and minimizing delivery time. 
In this work, we focus exclusively on the latter. 
Research on minimizing total consumption differs fundamentally from our analysis, since in our setting there are no battery constraints; therefore, the results do not translate. For work on minimizing fuel consumption, we refer to \cite{blank25, erlebach22}. Regarding delivery time, Bärtschi et al.\ \cite{bartschi2018} and Carvalho et al.\ \cite{carvalho2021fast} show that the problem of delivering a single package can be solved in polynomial time if the agents are free to move without restrictions -- even with predetermined starting positions. 
However, Carvalho et al.~\cite{carvalho2021fast} also prove that the problem becomes \NP{}-hard if there are at least two packages involved. For the case of restricted movement areas, Erlebach et al.\ \cite{erlebach22} prove that it is \NP{}-hard to approximate DDT within a factor of $O(\min\{n^{1-\varepsilon},k^{1-\varepsilon}\})$, even if agents have equal speeds and fixed starting positions.
Here, $n$ denotes the number of vertices in the graph and $k$ the number of agents.
In the problem variant in which the starting positions are not given but can be chosen by the algorithm, they show that no polynomial-time approximation algorithm with any finite approximation ratio exists, unless P=NP. 
They also show that DDT with fixed starting positions remains \NP{}-hard on path graphs if agents can have arbitrary speeds.

Bartlmae et al.\ \cite{bartlmae25} recently refined these results, showing that the problem on a path is still \NP{}-hard if only two different speeds are allowed. 
This closes the gap regarding speeds as these instances can be solved trivially for only one speed. 
Additionally, they show that for the special case of unit grid graphs and only two different speeds, the problem remains NP-hard for fixed starting positions and is hard to approximate within a factor of $O(n^{1-\varepsilon})$ for selectable starting positions, as well as agents constrained to movement areas forming rectangles. They also provide a simple $n$-approximation algorithm for this setting. Complementing these hardness results, Luo et al.\ \cite{luo25} demonstrated that on a path with fixed starting positions and exactly two speeds the problem can be solved in polynomial time; the complexity for three or more speeds, however, remains open.
\subparagraph*{Our Results.}

In Section \ref{sec:line:hard}, we provide a hardness result for DDT on path graphs with selectable starting positions, closing open questions posed by Erlebach et al.~\cite{erlebach22}, Bartlmae et al.~\cite{bartlmae25} and Luo et al.~\cite{luo25}.

\begin{restatable}{theorem}{thmlinehard}
DDT on a path with selectable starting positions is $a(n)$-APX-hard for any function $a : \mathbb{N} \rightarrow \mathbb{R}_+$ that can be computed in polynomial time. \label{thmlinehard}
\end{restatable}

We then present the first FPT algorithm for this setting in Section \ref{sec:line:fpt}, parameterized by the treewidth of the agents' intersection graph.

\begin{restatable}{theorem}{thmlinefpt}
DDT on a path with selectable starting positions can be solved in time $\OO(2^{\tw{}}\cdot \tw{}^2\cdot k)$, where $\tw{}$ denotes the treewidth of the intersection graph.\label{thmlinefpt}
\end{restatable}

In Section~\ref{sec:graph}, we consider the more complex case of general graphs. 
We design an FPT algorithm whose running time depends on the treewidth and maximum degree of the intersection graph.

\begin{restatable}{theorem}{thmgraphfpt}
DDT with selectable starting positions can be solved in time $f(\dm{},\tw{})\cdot \poly(n,k)$ on general graphs, where $\tw{}$ denotes the treewidth and $\dm{}$ denotes the maximum degree of the intersection graph and $f$ is a function in $\tw{}$ and $\dm{}$.\label{thmgraphfpt}
\end{restatable}

For the special case where the intersection graph is a tree, we show in Section~\ref{subsec:treegraph} that DDT can be solved in polynomial time. 
This result is not implied by the previous theorems; it allows for arbitrary maximum degree in the intersection graph. 

\begin{restatable}{theorem}{thmgraphtree}
DDT on a general graph with selectable starting positions can be solved in time $O(k^2\cdot n^2)$ if the underlying graph of the intersection graph is a tree.\label{thmgraphtree}
\end{restatable}
\section{Preliminaries}\label{sectionprelim}
In this section, we formally define the Drone Delivery Problem with respect to time (DDT).
\subparagraph*{Setting.} We define an instance of DDT as a tuple $I=(G,(s,t),A)$, where $G = (V, E, \ell)$ with $n=\vert V\vert$, is an undirected graph with edge lengths $\ell: E \rightarrow \mathbb{R}_{\geq 0}$. 
The package starting position, as well as its destination, are given by $(s,t)$, where $s\in V$ and $t\in V$. The set $A$ represents all agents, with $k=\vert A \vert$. There are two settings regarding agents which are typically considered for DDT:
\begin{enumerate}
    \item DDT with \textbf{selectable positions} (\textit{DDT-SP}): Each agent $a\in A$ is defined as a tuple $a=(v_a,G_a)$, where $v_a$ is the agent's speed and $G_a=(V_a,E_a)$ is a connected subgraph of $G$ representing the agent's movement area. The initial starting position of $a$ can be chosen freely once before delivery is initiated.
    \item DDT with \textbf{fixed positions} (\textit{DDT-FP}): Each agent $a\in A$ is a tuple $a=(v_a,G_a,p_a)$, 
    with $v_a$ and $G_a$ defined as above, and the additional parameter $p_a\in V_a$ specifies the fixed initial starting position of the agent.
\end{enumerate}
An agent $a$ traverses an edge $\{u,v\}\in E_a$ in time $\frac{\ell(\{u,v\})}{v_a}$, regardless of whether it is carrying the package or not.
Whenever two agents meet at the same vertex, they can hand over the package instantaneously.
We define the travel time of agent $a$ between two vertices $u,v$ as the length of the shortest path in $G_a$ divided by $v_a$ and denote it by $d_a(u,v)$. If either $u$ or $v$ is not contained in $V_a$, we define $d_a(u,v)=\infty$. 
For any pair of vertices and any possible agent the values can be precomputed in polynomial time.

We specify a feasible solution to be a schedule $S$ of consecutive trips by agents delivering the package from $s$ to $t$. More precisely we want the schedule to consist of two components: A complete edge-by-edge trip of every agent and a complete edge-by-edge trip for the package specifying the current package carrier at any point in time.
For the first part we want for every agent $a$ a set of sorted triples $\mathcal{T}_a$. Every triple $(u,v,\tau)$ indicates that agent $a$ moves over edge $\{u,v\}$ starting at time $\tau$. It has to hold that $\{u,v\}\in E_a$ and that for two consecutive triples $(u,v,\tau)$ and $(u',v',\tau')$ that $u'=v$ and $\tau + \frac{\ell(\{u,v\})}{v_a}\leq \tau'$. Moreover, for DDT-SP we need to provide the starting position $p_a$ for every agent $a\in A$ and for DDT-FP we need to have for the first triple that $u=p_a$. 
The trip of the package is defined as a sorted set of tuples $\mathcal{T}$ of the form $(u,v,a,\tau)$, indicating that agent $a$ carries the package over edge $\{u,v\}$ starting at time $\tau$. A schedule is feasible only if, for every tuple $(u, v, a, \tau) \in \mathcal{T}$, it holds that $(u, v, \tau) \in \mathcal{T}_a$. Again for two consecutive tuples $(u,v,a,\tau)$ and $(u',v',a',\tau')$ we also require $u'=v$ and $\tau + \frac{\ell(\{u,v\})}{v_a}\leq \tau'$. Additionally, the first tuple must satisfy $u = s$, and the last tuple must satisfy $v = t$.
The total delivery time is defined as $c(S):=\tau+\frac{\ell(\{u,v\})}{v_a}$, assuming $(u,v,a,\tau)$ is the final tuple in the schedule.

An instance of DDT-SP is illustrated in Figure \ref{fig:example:g}. Throughout this paper we will focus exclusively on DDT-SP.

\subparagraph*{Useful Properties.}

There are several useful properties regarding DDT. Erlebach et al.~\cite{erlebach22} demonstrated that there exists an optimal schedule in which every involved agent picks up and drops off the package exactly once. The key insight is that instead of using an agent multiple times we let that agent carry the package until its final drop-off without increasing the delivery time.

Another important observation relates to trips made by agents:

\begin{observation}
\label{obs:nowaiting}
In DDT-SP there exists an optimal schedule such that all involved agents move if and only if they carry the package.
\end{observation}

Placing agents initially at the first vertex of their respective trip immediately leads to this observation. Unless stated otherwise, this implicit initial positioning is assumed.

\noindent
\subparagraph*{Intersection Graph.}

\noindent
A useful structural concept is the \textit{intersection graph} of the agents. We define the intersection (multi)graph $G_I$ of a DDT instance $(G, (s, t), A)$ as follows:

\begin{itemize}
\item The vertex set $V(G_I)$ is the set of agents $A$.
\item For any two agents $a, a' \in A$ with corresponding vertex sets $V_a$ and $V_{a'}$ in $G$, we add $\lvert V_a \cap V_{a'} \rvert$ parallel edges between $a$ and $a'$ in $G_I$ -- one for each vertex they share in $G$.

Formally, the edge set $E(G_I)$ is defined as the multiset
\[
E(G_I) = \biguplus_{\substack{a,a' \in A\\ a \neq a'}} \big\{\, \underbrace{\{a, a'\}, \dots, \{a, a'\}}_{\lvert V_a \cap V_{a'} \rvert\text{ times}}\, \big\}.
\]
\end{itemize}
Throughout this paper, we treat $E$ as a multiset and omit explicit edge indices where they are not relevant. 
Intuitively, a feasible schedule corresponds to a path in the intersection graph that starts with an agent intersecting $s$ and ends with an agent intersecting $t$. 
If we assume each agent is used at most once, this path will be simple. Figure \ref{fig:example:intersect} illustrates the intersection graph for a given instance.
We denote by $\Delta(G_I)$ the maximum degree of the intersection graph $G_I$. When it is clear from the context, we will write $\Delta$ instead of $\Delta(G_I)$.

\begin{figure}[htb]
  \centering
  \begin{subfigure}[t]{0.5\textwidth}
    \centering
    \resizebox{0.8\textwidth}{!}{
    \begin{tikzpicture}[scale=1]
  \coordinate (A) at (0,0);
  \coordinate (B) at (1,0);
  \coordinate (C) at (2,0);
  \coordinate (D) at (2,0.8);
  \coordinate (E) at (1,0.8);
  \coordinate (F) at (0,0.8);

  \foreach \v in {A,B,C,D,E,F}
    \fill (\v) circle (1pt);

  \draw (A) -- (B) -- (C) -- (D) -- (E) -- (F) -- (A);
  \draw (B) -- (D);
  \draw (B) -- (E);

\filldraw[blue!80!white, line width=0.5pt, rounded corners=3pt, fill opacity=0.1](-0.1, -0.1) -- (1.1, -0.1) -- (1.1, 0.9) -- (-0.1, 0.9) -- cycle;
\filldraw[red!80!black, line width=0.5pt, rounded corners=2pt, fill opacity=0.1](-0.07, 0.73) -- (1.07, 0.73) -- (1.07, 0.87) -- (-0.07, 0.87) -- cycle;

\node[scale=0.4] at (0.5, 0.08) {1};

\node[scale=0.4] at (0.5, 0.65) {10};
\node[scale=0.4] at (1.5, 0.72) {1};
\node[scale=0.4] at (1.55, 0.12) {2};
\node[scale=0.4] at (0.099, 0.4) {10};
\node[scale=0.4] at (2.05, 0.4) {3};
\node[scale=0.4] at (1.44, 0.46) {2};
\node[scale=0.4] at (0.5, 0.08) {1};

\node[scale=0.39] at (0.93, 0.4) {$8$};

\node[scale=0.39] at (-0.12, -0.12) {$s$};
\node[scale=0.39] at (1, -0.2) {$u_1$};
\node[scale=0.39] at (2.06, -0.17) {$u_2$};
\node[scale=0.39] at (2.1, 0.97) {$u_3$};
\node[scale=0.39] at (1, 0.99) {$u_4$};
\node[scale=0.39] at (-0.12, 0.95) {$t$};


\filldraw[green!50!black, line width=0.5pt, rounded corners=3.5pt, fill opacity=0.1](0.87, -0.13) -- (1.13, -0.13) -- (2.13, 0.67) -- (2.13, .93) -- (0.87, .93) -- cycle;

\filldraw[violet, line width=0.5pt, rounded corners=2pt, fill opacity=0.1](0.92, -0.08) -- (2.12, -0.08) -- (2.12, 0.88) -- (1.92, .88) -- (1.9, 0.05) -- (0.9, 0.05) -- cycle;



\end{tikzpicture}}
    \caption{}
    \label{fig:example:g}
  \end{subfigure}\hfill
  \begin{subfigure}[t]{0.5\textwidth}
    \centering
    \resizebox{0.5\textwidth}{!}{\begin{tikzpicture}[scale=1]
  \coordinate (A) at (0.3,0.5);
  \coordinate (B) at (1,0);
  \coordinate (C) at (1,1);
  \coordinate (D) at (1.7,0.5);

\fill[red] (A) circle (1pt);
\fill[blue] (A) circle (1pt);
\fill[green] (A) circle (1pt);
\fill[violet] (A) circle (1pt);

  \draw (C) -- (A);
  \draw (B) -- (D);
  \draw[] (A) to[out=-35.26+10,in=180-35.26-10] (B);
  \draw[] (C) to[out=-35.26+10,in=180-35.26-10] (D);
  \draw[] (A) to[out=-35.26-10,in=180-35.26+10] (B);
  \draw[] (C) to[out=-35.26-10,in=180-35.26+10] (D);
  

    \draw[] (C) to[out=260,in=100] (B);
  \draw[] (C) to[out=280,in=80] (B);
  

\fill[red!80!black] (A) circle (2pt);
\fill[blue] (B) circle (2pt);
\fill[green!50!black] (C) circle (2pt);
\fill[violet] (D) circle (2pt);

\end{tikzpicture}}
    \caption{}
    \label{fig:example:intersect}
  \end{subfigure}
  \caption{Fig.~\ref{fig:example:g} shows a DDT instance with four agents, each restricted to a colored movement area (bounding boxes). The red agent has speed $10$, the blue and green agents have speed $1$, and the purple agent has speed $2$. In the optimal schedule, the blue agent starts at $s$ and delivers the package to $u_1$ in time $1$. The green agent then transports it via $u_3$ to $u_4$ in time $2 + 1 = 3$. Finally, the red agent delivers it to the target $t$ in time $1$, for a total delivery time of $5$.
Fig.~\ref{fig:example:intersect} shows the corresponding intersection graph.}
  \label{fig:example}
\end{figure}
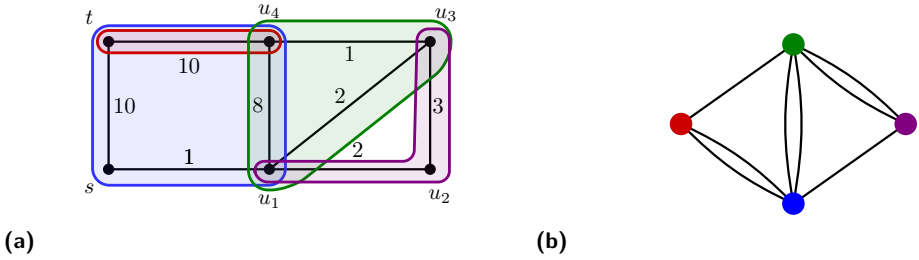

\noindent
For any vertex $u$, we define $B_u$ as the set of agents that cover $u$, i.e., $B_u := \{a \in A \mid u \in V_a\}$.
\subparagraph*{Tree Decomposition.}

For computing the optimal traversal of the agents, we will use the concept of a \textit{tree decomposition}.
A tree decomposition is used as a measure of how similar a given graph is to a tree and was first introduced by Robertson and Seymour~\cite{RS86}.
Computing the optimal treewidth is \NP{}-hard~\cite{ACP87}, but finding the optimal tree decomposition is fixed-parameter tractable if parameterized by treewidth~\cite{B93}.

\begin{definition}[Tree Decomposition]
A tree decomposition of a graph $G=(V,E)$ is a pair $\mathcal{T}=(\mathbb{T},\{X_t\}_{t\in V(\mathbb{T})})$ where $\mathbb{T}$ is a tree with root $r\in V(\mathbb{T})$ and every node $t\in V(\mathbb{T})$ is associated with a set of the vertices $X_t\subseteq V$, called bag.
Additionally, the following three conditions must be fulfilled:
\begin{itemize}
\item $\bigcup_{t\in V(\mathbb{T})}X_t = V(G)$, i.e., every vertex $v\in V(G)$ appears in at least one bag,
\item for every edge $\{u,v\}\in E(G)$ there must exist a bag $t\in V(\mathbb{T})$ such that the vertices $u$ and $v$ are contained in this bag,
\item the set of nodes $T_u = \{t\in V(\mathbb{T}) \mid u \in X_t\}$ forms a connected subtree of $\mathbb{T}$ for every choice $u\in V(G)$.
\end{itemize}
\end{definition}
The \emph{width} of a tree decomposition is the size of its largest bag minus one. The \emph{treewidth} of a graph $G$ is the minimum width over all possible tree decompositions of $G$, denoted by $\tw$. 
To simplify the algorithm design, we assume the decomposition has a specific structure known as a \emph{nice tree decomposition}: a binary tree with root $r$ and leaves $l$ such that $X_r = \emptyset$ and $X_l = \emptyset$ for all $l \in V(\mathbb{T})$, containing only three types of nodes:

\begin{itemize}
\item \textbf{Introduce node}: For two nodes $t_\text{parent},t_\text{child} \in V(\mathbb{T})$ where $t_\text{parent}$ has exactly one child $t_\text{child}$, $t_\text{parent}$ is an \textit{introduce node} iff $X_{t_\text{parent}} = X_{t_\text{child}}\cup \{v\}$ for some $v\in V(G)$.
\item \textbf{Forget node}: For two nodes $t_\text{parent},t_\text{child} \in V(\mathbb{T})$ where $t_\text{parent}$ has exactly one child $t_\text{child}$, $t_\text{parent}$ is a \textit{forget node} iff $X_{t_\text{parent}} \cup \{v\} = X_{t_\text{child}}$ for some $v\in V(G)$.
\item \textbf{Join node}: For three nodes $t_\text{parent},t_\text{child}^1, t_\text{child}^2 \in V(\mathbb{T})$ where $t_\text{parent}$ has exactly two children $t_\text{child}^1$ and $t_\text{child}^2$, $t_\text{parent}$ is a \textit{join node} iff $X_{t_\text{parent}} = X_{t_\text{child}^1} = X_{t_\text{child}^2}$.
\end{itemize}
As in~\cite{CF15} we additionally assume that the edges are also introduced like the vertices in their respective type of node:
\begin{itemize}
    \item \textbf{Introduce edge node}: For every edge $e=\{u,v\}\in E$ there is exactly one introduce edge node $t\in V(\mathbb{T})$ with $u,v\in X_t=X_{t'}$ for its only child bag $t'$ and $X_t = X_{t'}$.
\end{itemize}
We assume that for every edge $\{u,v\}$ the corresponding introduce edge node appears directly before either $u$ or $v$ is forgotten in an ancestor node.

Using all node types, we define $V_t$ and $E_t$ as the set of vertices and edges that are either contained or introduced in bag $t$ or were introduced in some child of $t$.
Additionally, we define $G_t=(V_t,E_t)$ as the subgraph revealed up to node $t$.
As a consequence for the construction of the introduce edge nodes, for every join node $t\in V(\mathbb{T})$ there are no edges with both endpoints in $X_t$, i.e. $E_t\cap \{\{u,v\}\mid u,v\in X_t\}=\emptyset$.

Based on a tree decomposition of width $\tw{}$ for some graph $G$, we can always compute a nice tree decomposition (with or without introduce edge nodes) of width \tw{} with $\OO(|V(G)|\cdot \tw{})$ nodes in polynomial time~\cite{CF15}.
\section{Drone Delivery on Path Graphs}\label{sec:line}

We consider DDT-SP instances where the underlying graph is a path: a sequence of vertices connected in a line, with degree two at each internal vertex and degree one at the two endpoints. For an agent~$a$, we denote its movement area by $[s_a, f_a]$, where $s_a$ and $f_a$ are the leftmost and rightmost vertices that $a$ can access. This setting can also be interpreted as agents moving along a one-dimensional line, with each movement area corresponding to an interval on that line, as illustrated in Figure~\ref{fig:path}.  
This allows us to interpret DDT-SP on a path as an instance of the \emph{subinterval covering} problem introduced by Luo et al.~\cite{luo25}, where the objective is to select, for each agent's interval, a subinterval such that the union covers $[s, t]$, while minimizing the total cost defined by the length of each selected subinterval divided by the agent’s speed.

\begin{figure}[h]
    \centering
    \begin{tikzpicture}[scale=0.6]
\definecolor{dgreen}{RGB}{13, 138, 8}
[0.0, 1.8, 3.6, 5.4, 7.2, 9.0, 10.8, 12.6, 14.4]
\foreach \x in {0.0, 3.6, 5.4, 7.2, 9.0, 10.8, 14.4, 16.2} {
    \draw[gray!50, dashed] (\x, -0.4) -- (\x, 2.1);
}


\draw[black, opacity=0.5, line width=2.7pt]   (0.06, 0.0) -- (5.340000000000001, 0.0);
\draw[black, opacity=0.5, line width=0.06*1cm/1pt]   (0.03, 0.15) -- (0.03, -0.15);
\draw[black, opacity=0.5, line width=0.06*1cm/1pt]   (5.37, 0.15) -- (5.37, -0.15);
\node[] at (-0.5, 0.0) {\color{black}$a_1$};
\draw[black, opacity=0.5, line width=2.7pt]   (0.06, 0.4) -- (16.14, 0.4);
\draw[black, opacity=0.5, line width=0.06*1cm/1pt]   (0.03, 0.55) -- (0.03, 0.25);
\draw[black, opacity=0.5, line width=0.06*1cm/1pt]   (16.169999999999998, 0.55) -- (16.169999999999998, 0.25);
\node[] at (-0.5, 0.4) {\color{black}$a_2$};
\draw[black, opacity=0.5, line width=2.7pt]   (3.66, 0.8) -- (8.94, 0.8);
\draw[black, opacity=0.5, line width=0.06*1cm/1pt]   (3.63, 0.9500000000000001) -- (3.63, 0.65);
\draw[black, opacity=0.5, line width=0.06*1cm/1pt]   (8.97, 0.9500000000000001) -- (8.97, 0.65);
\node[] at (3.1, 0.8) {\color{black}$a_3$};
\draw[black, opacity=0.5, line width=2.7pt]   (7.26, 1.2000000000000002) -- (16.14, 1.2000000000000002);
\draw[black, opacity=0.5, line width=0.06*1cm/1pt]   (7.23, 1.35) -- (7.23, 1.0500000000000003);
\draw[black, opacity=0.5, line width=0.06*1cm/1pt]   (16.169999999999998, 1.35) -- (16.169999999999998, 1.0500000000000003);
\node[] at (6.7, 1.2000000000000002) {\color{black}$a_4$};
\draw[black, opacity=0.5, line width=2.7pt]   (10.860000000000001, 1.6) -- (14.34, 1.6);
\draw[black, opacity=0.5, line width=0.06*1cm/1pt]   (10.83, 1.75) -- (10.83, 1.4500000000000002);
\draw[black, opacity=0.5, line width=0.06*1cm/1pt]   (14.370000000000001, 1.75) -- (14.370000000000001, 1.4500000000000002);
\node[] at (10.3, 1.6) {\color{black}$a_5$};
\draw (0,-1) -- (16.2,-1.3);
\node[circle, draw, fill=white, inner sep=2pt, font=\small] at (0.0, -1.3) {\shortstack{$a_1$ \\ $a_2$}};
\node[circle, draw, fill=white, inner sep=1pt, font=\small] at (3.6, -1.3) {\shortstack{$a_1$ \\ $a_2$ \\ $a_3$}};
\node[circle, draw, fill=white, inner sep=2pt, font=\small] at (5.4, -1.3) {\shortstack{$a_2$ \\ $a_3$}};
\node[circle, draw, fill=white, inner sep=1pt, font=\small] at (7.2, -1.3) {\shortstack{$a_2$ \\ $a_3$ \\ $a_4$}};
\node[circle, draw, fill=white, inner sep=2pt, font=\small] at (9.0, -1.3) {\shortstack{$a_2$ \\ $a_4$}};
\node[circle, draw, fill=white, inner sep=1pt, font=\small] at (10.8, -1.3) {\shortstack{$a_2$ \\ $a_4$ \\ $a_5$}};
\node[circle, draw, fill=white, inner sep=2pt, font=\small] at (14.4, -1.3) {\shortstack{$a_2$ \\ $a_4$}};
\node[circle, draw, fill=white, inner sep=2pt, font=\small] at (16.2, -1.3) {\shortstack{\textcolor{white}{$a$}}};
\draw  ;
\end{tikzpicture}
    \caption{A DDT-SP instance on a line and a corresponding tree decomposition of its intersection graph.
    At every start or end position we either add or remove the corresponding agent.}
    \label{fig:path}
\end{figure}
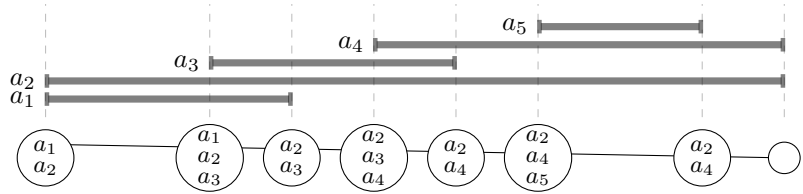

For the setting with fixed initial positions, NP-hardness has been proven~\cite{bartlmae25, erlebach22}, although strong NP-hardness remains unresolved. 
In contrast, for the setting with selectable positions -- which is the primary focus of this paper -- no hardness results were previously known. 
Before presenting parametrized algorithms for DDT on path graphs, we first establish hardness results for this setting, which motivates the relevance and necessity of parameterized approaches.

\subsection{Inapproximability}\label{sec:line:hard}
Instead of merely proving NP-hardness, we establish a much stronger result: no polynomial-time algorithm with any polynomially encodable approximation ratio exists, unless P=NP.
\thmlinehard*
In the following, we first provide a high-level overview before presenting the proof in full technical detail.
To prove the inapproximability, we reduce from the \textsc{PartitionInto}-$k$ problem.
\begin{definition}[\textnormal{\textsc{PartitionInto}-$k$}]
Given a set of natural numbers $p_1, \dots, p_n$ and an integer $k$, the task is to decide whether the index set $\{1, \dots, n\}$ can be partitioned into $k$ disjoint subsets $I_1 \dot\cup \dots \dot\cup I_k$ such that $\sum_{i \in I_1} p_i = \dots = \sum_{i \in I_k} p_i$.
\end{definition}
Note that in this subsection $k$ refers to the number of subsets (and not the agent set). The NP-hardness of \textsc{PartitionInto}-$k$ follows immediately from the hardness of \textsc{Partition} with $k=2$; however for the proof of Theorem~\ref{thmlinehard} we require the stronger result that \textsc{PartitionInto}-$k$ is also NP-hard in the strong sense, which we will prove first. It is particularly important that $k$ is part of the input; for a fixed $k$, the problem admits an FPTAS, as shown by Sahni \cite{sahni76}. 
\begin{restatable}{theorem}{propkpart}
    \textnormal{\textsc{PartitionInto}}-$k$ is NP-hard in the strong sense (even with distinct integers).
\end{restatable}
\begin{proof}
     For the reduction, we use the problem $3$-\textsc{Partition}, whose NP-hardness was first shown by Garey and Johnson~\cite{garey75}. Hullet et al.~\cite{hullet08} also showed that $3$-\textsc{Partition} remains strongly NP-hard even when all numbers are distinct.

In the $3$-\textsc{Partition} problem, the task is to decide, given a set of integers $p_1, \dots, p_{3m}$, whether it can be partitioned into $m$ triples, each summing to the same value.

Let $I = \{p_1, \dots, p_{3m}\}$ be an instance of $3$-\textsc{Partition}, and let $P = \sum_{i=1}^{3m} p_i$. We define $p_i' := p_i + P$ for all $i \in \{1, \dots, 3m\}$. Let $I'$ be the corresponding \textsc{PartitionInto}-$k$ instance consisting of the elements $p_1', \dots, p_{3m}'$ and set $k := m$.

We now show that $I$ is a “yes”-instance of $3$-\textsc{Partition} if and only if $I'$ is a “yes”-instance of \textsc{PartitionInto}-$k$.

For the forward direction, consider a valid partition $I_1, \dots, I_m$ of the $3$-\textsc{Partition} instance $I$ such that $\sum_{i \in I_1} p_i = \dots = \sum_{i \in I_m} p_i$. Since each set $I_j$ contains exactly 3 elements, we have $\sum_{i \in I_j} p_i' = \sum_{i \in I_j} p_i + 3P$. Hence, $\sum_{i \in I_1} p_i' = \dots = \sum_{i \in I_m} p_i'$, and thus $I_1, \dots, I_m$ is a valid partition for $I'$.
    
For the reverse direction, consider a valid partition $I_1, \dots, I_m$ of the \textsc{PartitionInto}-$k$ instance $I'$, such that $\sum_{i \in I_1} p_i' = \dots = \sum_{i \in I_m} p_i' = 3P + \frac{P}{m}$. Since $0 < \frac{P}{m} < P$, each set $I_j$ must contain exactly 3 elements. Using the identity $\sum_{i \in I_j} p_i' = \sum_{i \in I_j} p_i + 3P$, it follows that $\sum_{i \in I_1} p_i = \dots = \sum_{i \in I_m} p_i$, and thus $I_1, \dots, I_m$ is a valid partition for $I$.
\end{proof}

To show that no approximation algorithm $\mathcal{A}$ with a polynomially encodable approximation ratio $a(n)$ can exist (unless P=NP), we demonstrate that such an algorithm could be used to solve \textsc{PartitionInto}-$k$ in polynomial time. We do so by constructing a DDT instance $I'$ from a given \textsc{PartitionInto}-$k$ instance $I$ such that:

$I$ is a “yes”-instance of \textsc{PartitionInto}-$k$ if and only if $c_{\mathcal{A}}(I') \leq d \cdot a(n)$, where $c_{\mathcal{A}}(I')$ is the delivery time of the schedule produced by $\mathcal{A}$ for the instance $I'$ and the value $d$ is determined by the construction.

Let $I = (p_1, \dots, p_n)$ be an arbitrary \textsc{PartitionInto}-$k$ instance with distinct integers. We assume that the numbers are sorted in decreasing order, i.e., $p_1 > \dots > p_n$. We denote the total sum of all elements by $P$.

For each object $p_i$ we create $k$ drones – one for each partition set. 
We refer to these drones as \emph{element agents} and denote them by $e_{i,j}$, where $e_{i,j}$ is associated with object $p_i$ and partition set $j \in \{1, \dots, k\}$. The speed of a drone $e_{i,j}$ is denoted by $v_j$, meaning that all drones associated with the same partition set (i.e., with the same index $j$) share the same speed. 
We place several gadgets along the line and aim to ensure that each element agent can only be placed and transport the package at specific designated positions and not at arbitrary locations along the path. 
To achieve this, we employ two key mechanisms: first, we restrict the movement areas of the agents such that they are clearly unable to assist outside their interval. Second, we choose the ratios between the speeds $v_1, \dots, v_k$ to be sufficiently large, allowing us to define long intervals that certain agents cannot traverse in their entirety without exceeding the time bound $d \cdot a(n)$.
 
For each of the $k$ partition sets, we create a \emph{partition gadget}, that enforces that the sum of the weights $p_i$ associated with the element agents $e_{i,j}$ placed at that gadget is at least $\frac{P}{k}$. Placing an element agent $e_{i,j}$ at a partition gadget can thus be interpreted as assigning $p_i$ to the corresponding partition set.

Additionally, for each object $p_i$ we construct a gadget $O_i$ – called an \emph{object gadget} – that ensures each object is assigned to at most one partition set. This is done by requiring that exactly $k-1$ of the agents $e_{i,1}, \dots, e_{i,k}$ must be present at the object gadget, resulting in at most one of them being placed at a partition gadget. As a result, we can ensure that each partition gadget ends up with a total weight of exactly $\frac{P}{k}$, if such a partition exists. If not, the construction forces a very large delivery time, causing the bound $d \cdot a(n)$ to be exceeded.

An example layout of the gadgets, ordered as $P_1, \dots, P_k, O_n, \dots, O_1$, is shown in Figure~\ref{fig:line:sp:gen}. Consecutive gadgets are placed sufficiently far apart so that only specially designated fast transition agents can traverse the gaps, ensuring that each element agent can contribute to at most one gadget.  
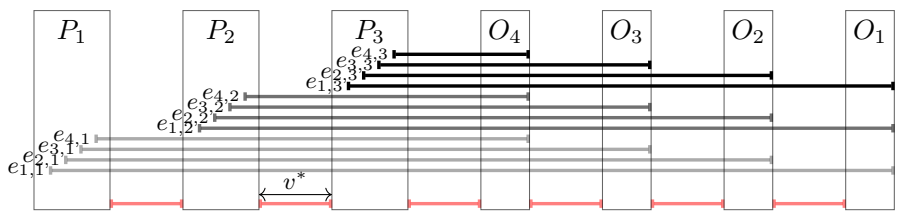
\begin{figure}[h]
        \centering
        \begin{tikzpicture}[scale=0.8]
\definecolor{dgreen}{RGB}{13, 138, 8}
0.25 14.15 0.55 black
\draw[black, opacity=0.33, line width=0.045*1cm/1pt]   (0.295, 0.55) -- (14.105, 0.55);
\draw[black, opacity=0.33, line width=0.045*1cm/1pt]   (0.2725, 0.6200000000000001) -- (0.2725, 0.48000000000000004);
\draw[black, opacity=0.33, line width=0.045*1cm/1pt]   (14.1275, 0.6200000000000001) -- (14.1275, 0.48000000000000004);
\node at (-0.12, 0.55)  {\small $e_{1, 1}$};
0.5 12.15 0.7250000000000001 black
\draw[black, opacity=0.33, line width=0.045*1cm/1pt]   (0.545, 0.7250000000000001) -- (12.105, 0.7250000000000001);
\draw[black, opacity=0.33, line width=0.045*1cm/1pt]   (0.5225, 0.7950000000000002) -- (0.5225, 0.655);
\draw[black, opacity=0.33, line width=0.045*1cm/1pt]   (12.1275, 0.7950000000000002) -- (12.1275, 0.655);
\node at (0.13, 0.7250000000000001)  {\small $e_{2, 1}$};
0.75 10.15 0.9000000000000001 black
\draw[black, opacity=0.33, line width=0.045*1cm/1pt]   (0.795, 0.9000000000000001) -- (10.105, 0.9000000000000001);
\draw[black, opacity=0.33, line width=0.045*1cm/1pt]   (0.7725, 0.9700000000000002) -- (0.7725, 0.8300000000000001);
\draw[black, opacity=0.33, line width=0.045*1cm/1pt]   (10.1275, 0.9700000000000002) -- (10.1275, 0.8300000000000001);
\node at (0.38, 0.9000000000000001)  {\small $e_{3, 1}$};
1.0 8.15 1.0750000000000002 black
\draw[black, opacity=0.33, line width=0.045*1cm/1pt]   (1.045, 1.0750000000000002) -- (8.105, 1.0750000000000002);
\draw[black, opacity=0.33, line width=0.045*1cm/1pt]   (1.0225, 1.1450000000000002) -- (1.0225, 1.0050000000000001);
\draw[black, opacity=0.33, line width=0.045*1cm/1pt]   (8.1275, 1.1450000000000002) -- (8.1275, 1.0050000000000001);
\node at (0.63, 1.0750000000000002)  {\small $e_{4, 1}$};
2.7 14.15 1.2500000000000002 black
\draw[black, opacity=0.56, line width=0.045*1cm/1pt]   (2.745, 1.2500000000000002) -- (14.105, 1.2500000000000002);
\draw[black, opacity=0.56, line width=0.045*1cm/1pt]   (2.7225, 1.3200000000000003) -- (2.7225, 1.1800000000000002);
\draw[black, opacity=0.56, line width=0.045*1cm/1pt]   (14.1275, 1.3200000000000003) -- (14.1275, 1.1800000000000002);
\node at (2.33, 1.2500000000000002)  {\small $e_{1, 2}$};
2.95 12.15 1.4250000000000003 black
\draw[black, opacity=0.56, line width=0.045*1cm/1pt]   (2.995, 1.4250000000000003) -- (12.105, 1.4250000000000003);
\draw[black, opacity=0.56, line width=0.045*1cm/1pt]   (2.9725, 1.4950000000000003) -- (2.9725, 1.3550000000000002);
\draw[black, opacity=0.56, line width=0.045*1cm/1pt]   (12.1275, 1.4950000000000003) -- (12.1275, 1.3550000000000002);
\node at (2.58, 1.4250000000000003)  {\small $e_{2, 2}$};
3.2 10.15 1.6000000000000003 black
\draw[black, opacity=0.56, line width=0.045*1cm/1pt]   (3.245, 1.6000000000000003) -- (10.105, 1.6000000000000003);
\draw[black, opacity=0.56, line width=0.045*1cm/1pt]   (3.2225, 1.6700000000000004) -- (3.2225, 1.5300000000000002);
\draw[black, opacity=0.56, line width=0.045*1cm/1pt]   (10.1275, 1.6700000000000004) -- (10.1275, 1.5300000000000002);
\node at (2.83, 1.6000000000000003)  {\small $e_{3, 2}$};
3.45 8.15 1.7750000000000004 black
\draw[black, opacity=0.56, line width=0.045*1cm/1pt]   (3.495, 1.7750000000000004) -- (8.105, 1.7750000000000004);
\draw[black, opacity=0.56, line width=0.045*1cm/1pt]   (3.4725, 1.8450000000000004) -- (3.4725, 1.7050000000000003);
\draw[black, opacity=0.56, line width=0.045*1cm/1pt]   (8.1275, 1.8450000000000004) -- (8.1275, 1.7050000000000003);
\node at (3.08, 1.7750000000000004)  {\small $e_{4, 2}$};
5.15 14.15 1.9500000000000004 black
\draw[black, opacity=1, line width=0.045*1cm/1pt]   (5.195, 1.9500000000000004) -- (14.105, 1.9500000000000004);
\draw[black, opacity=1, line width=0.045*1cm/1pt]   (5.1725, 2.0200000000000005) -- (5.1725, 1.8800000000000003);
\draw[black, opacity=1, line width=0.045*1cm/1pt]   (14.1275, 2.0200000000000005) -- (14.1275, 1.8800000000000003);
\node at (4.78, 1.9500000000000004)  {\small $e_{1, 3}$};
5.4 12.15 2.1250000000000004 black
\draw[black, opacity=1, line width=0.045*1cm/1pt]   (5.445, 2.1250000000000004) -- (12.105, 2.1250000000000004);
\draw[black, opacity=1, line width=0.045*1cm/1pt]   (5.4225, 2.1950000000000003) -- (5.4225, 2.0550000000000006);
\draw[black, opacity=1, line width=0.045*1cm/1pt]   (12.1275, 2.1950000000000003) -- (12.1275, 2.0550000000000006);
\node at (5.03, 2.1250000000000004)  {\small $e_{2, 3}$};
5.65 10.15 2.3000000000000003 black
\draw[black, opacity=1, line width=0.045*1cm/1pt]   (5.695, 2.3000000000000003) -- (10.105, 2.3000000000000003);
\draw[black, opacity=1, line width=0.045*1cm/1pt]   (5.6725, 2.37) -- (5.6725, 2.2300000000000004);
\draw[black, opacity=1, line width=0.045*1cm/1pt]   (10.1275, 2.37) -- (10.1275, 2.2300000000000004);
\node at (5.28, 2.3000000000000003)  {\small $e_{3, 3}$};
5.9 8.15 2.475 black
\draw[black, opacity=1, line width=0.045*1cm/1pt]   (5.945, 2.475) -- (8.105, 2.475);
\draw[black, opacity=1, line width=0.045*1cm/1pt]   (5.9225, 2.545) -- (5.9225, 2.4050000000000002);
\draw[black, opacity=1, line width=0.045*1cm/1pt]   (8.1275, 2.545) -- (8.1275, 2.4050000000000002);
\node at (5.53, 2.475)  {\small $e_{4, 3}$};
1.25 2.45 0 red
\draw[red, opacity=0.5, line width=0.045*1cm/1pt]   (1.295, 0) -- (2.4050000000000002, 0);
\draw[red, opacity=0.5, line width=0.045*1cm/1pt]   (1.2725, 0.07) -- (1.2725, -0.07);
\draw[red, opacity=0.5, line width=0.045*1cm/1pt]   (2.4275, 0.07) -- (2.4275, -0.07);
\node at (0.88, 0)  {\small $ $};
3.7 4.9 0 red
\draw[red, opacity=0.5, line width=0.045*1cm/1pt]   (3.745, 0) -- (4.855, 0);
\draw[red, opacity=0.5, line width=0.045*1cm/1pt]   (3.7225, 0.07) -- (3.7225, -0.07);
\draw[red, opacity=0.5, line width=0.045*1cm/1pt]   (4.8775, 0.07) -- (4.8775, -0.07);
\node at (3.33, 0)  {\small $ $};
6.15 7.3500000000000005 0 red
\draw[red, opacity=0.5, line width=0.045*1cm/1pt]   (6.195, 0) -- (7.305000000000001, 0);
\draw[red, opacity=0.5, line width=0.045*1cm/1pt]   (6.1725, 0.07) -- (6.1725, -0.07);
\draw[red, opacity=0.5, line width=0.045*1cm/1pt]   (7.327500000000001, 0.07) -- (7.327500000000001, -0.07);
\node at (5.78, 0)  {\small $ $};
8.15 9.35 0 red
\draw[red, opacity=0.5, line width=0.045*1cm/1pt]   (8.195, 0) -- (9.305, 0);
\draw[red, opacity=0.5, line width=0.045*1cm/1pt]   (8.172500000000001, 0.07) -- (8.172500000000001, -0.07);
\draw[red, opacity=0.5, line width=0.045*1cm/1pt]   (9.327499999999999, 0.07) -- (9.327499999999999, -0.07);
\node at (7.78, 0)  {\small $ $};
10.15 11.35 0 red
\draw[red, opacity=0.5, line width=0.045*1cm/1pt]   (10.195, 0) -- (11.305, 0);
\draw[red, opacity=0.5, line width=0.045*1cm/1pt]   (10.172500000000001, 0.07) -- (10.172500000000001, -0.07);
\draw[red, opacity=0.5, line width=0.045*1cm/1pt]   (11.327499999999999, 0.07) -- (11.327499999999999, -0.07);
\node at (9.780000000000001, 0)  {\small $ $};
12.15 13.35 0 red
\draw[red, opacity=0.5, line width=0.045*1cm/1pt]   (12.195, 0) -- (13.305, 0);
\draw[red, opacity=0.5, line width=0.045*1cm/1pt]   (12.172500000000001, 0.07) -- (12.172500000000001, -0.07);
\draw[red, opacity=0.5, line width=0.045*1cm/1pt]   (13.327499999999999, 0.07) -- (13.327499999999999, -0.07);
\node at (11.780000000000001, 0)  {\small $ $};
\draw[<->] (3.7, 0.15) -- (4.9, 0.15);
\node[] at (4.300000000000001, 0.38) {\small $v^*$};
\draw[opacity=0.6] (0, -0.1) rectangle (1.25, 3.2);
\node at (0.625, 2.85)  {\large $P_1$};
\draw[opacity=0.6] (2.45, -0.1) rectangle (3.7, 3.2);
\node at (3.075, 2.85)  {\large $P_2$};
\draw[opacity=0.6] (4.9, -0.1) rectangle (6.15, 3.2);
\node at (5.525, 2.85)  {\large $P_3$};
\draw[opacity=0.6] (7.3500000000000005, -0.1) rectangle (8.15, 3.2);
\node at (7.75, 2.85)  {\large $O_4$};
\draw[opacity=0.6] (9.35, -0.1) rectangle (10.15, 3.2);
\node at (9.75, 2.85)  {\large $O_3$};
\draw[opacity=0.6] (11.35, -0.1) rectangle (12.15, 3.2);
\node at (11.75, 2.85)  {\large $O_2$};
\draw[opacity=0.6] (13.35, -0.1) rectangle (14.15, 3.2);
\node at (13.75, 2.85)  {\large $O_1$};
\end{tikzpicture}
        \caption{Overview of all gadgets in an example with four objects and three partition sets. The element agents are shown in gray, with darker shades indicating higher speeds. Transition agents (marked in red) with speed $v^* \gg v_1, \dots, v_k$ are placed between gadgets; if these are bypassed and an element agent is used at multiple gadgets, the time bound $d \cdot a(n)$ is exceeded. }
        \label{fig:line:sp:gen}
    \end{figure}

We now define the partition gadgets and object gadgets, whose arrangement -- in the order $P_1, \dots, P_k, O_n, \dots, O_1$ -- is illustrated in \cref{fig:line:sp:gen}. Since each element agent is intended to contribute to exactly one gadget, a consecutive gadgets are placed at very large distance of $v^*$ (where $v^* > v_k \cdot d\cdot a(n) $). For each such interval, we place a dedicated agent with speed $v^*$, referred to as a \emph{G-transition agent}. This setup prevents element agents from contributing to multiple gadgets, as doing so would require them to traverse a distance of $v^*$, which would take at least $\frac{v^*}{v_k} > d \cdot a(n)$ time. We refer to such a violation as a \emph{G-error}.

\subparagraph*{Partition gadgets} We now describe a partition gadget $P_j$ for $j \in \{1, \dots, k\}$:
The core idea behind a partition gadget is as follows: for each drone $e_{i,j}$, there is a special construction within the gadget $P_j$ called the \emph{critical position}, denoted $C_{i,j}$. If a drone $e_{i,j}$ is not present at its corresponding critical position, it must instead be covered by other drones, referred to as \emph{helper agents}. These helper agents have a strictly lower speed $v_j' < v_j$, which allows us to design the construction such that if $e_{i,j}$ is not present at its critical position, exactly $p_i$ helper agents are required to cover $C_{i, j}$.

Each partition gadget $P_j$ contains exactly $P - \frac{P}{k}$ helper agents (all with speed $v_j' < v_j$), who can move freely within the boundaries of $P_j$. As a result, the total sum of the $p_i$ values associated with the element agents ($e_{i,j}$) located at their critical positions ($C_{i,j}$) must be at least $\frac{P}{k}$; otherwise, there are not enough helper agents available to cover the remaining critical positions.

In other words, covering the critical position $C_{i,j}$ with the element agent $e_{i,j}$ can be seen as saving $p_i$ helper agents. Since each partition gadget provides only $P - \frac{P}{k}$ helper agents, at least $\frac{P}{k}$ many must be saved by having element agents present at their respective critical positions. This implies that the sum of $p_i$-values of the element agents $e_{i, j}$ located at $P_j$ must be at least $\frac{P}{k}$. We define the speed hierarchy as $v^* > v_k > v_k' > \dots > v_1 > v_1'$. 

\subparagraph*{Critical positions}
A critical position $C_{i,j}$ consists of $p_i - 1$ \emph{partition agents} with speed $v^* > v_j$, placed with gaps between them that must be bridged. There are two options: either $e_{i,j}$ bridges all gaps alone, covering the entire critical position, or $p_i$ helper agents do so in combination with the partition agents, with each helper agent covering exactly one gap. These gaps are spaced at distance $v_j$, so that a helper agent with speed $v_j'$ would require time $\frac{v_j}{v_j'} > a(n) \cdot d$ to move between gaps, meaning it can only cover one. A helper agent traversing more than one gap would result in a \emph{C-error}. A depiction of $C_{i,j}$ is provided in \cref{fig:line:sp:crit}.

\begin{figure}[h]
    \centering
    \begin{tikzpicture}[scale=1.0]
\definecolor{dgreen}{RGB}{13, 138, 8}
\fill[blue, opacity=0.5] (6.0, 1.65) circle (1.5pt);
\fill[blue, opacity=0.5] (6.0, 1.8499999999999999) circle (1.5pt);
\fill[blue, opacity=0.5] (6.0, 1.45) circle (1.5pt);
\draw[black, line width=0.02*1cm/1pt]   (0.22, -0.3) -- (0.22, 0);
\node[] at (0.4, -0.5) {\color{black}$c_{i, j, 1}$};
\draw[black, line width=0.02*1cm/1pt]   (1.02, -0.3) -- (1.02, 0);
\node[] at (1.3, -0.5) {\color{black}$c_{i, j, 1}'$};
\draw[black, line width=0.02*1cm/1pt]   (2.88, -0.3) -- (2.88, 0);
\node[] at (2.8, -0.5) {\color{black}$c_{i, j, 2}$};
\draw[black, line width=0.02*1cm/1pt]   (3.72, -0.3) -- (3.72, 0);
\node[] at (3.8000000000000003, -0.5) {\color{black}$c_{i, j, 2}'$};
\draw[black, line width=0.02*1cm/1pt]   (5.58, -0.3) -- (5.58, 0);
\node[] at (5.5, -0.5) {\color{black}$c_{i, j, 3}$};
\draw[black, line width=0.02*1cm/1pt]   (6.419999999999999, -0.3) -- (6.419999999999999, 0);
\node[] at (6.499999999999999, -0.5) {\color{black}$c_{i, j, 3}'$};
\draw[black, line width=0.02*1cm/1pt]   (8.28, -0.3) -- (8.28, 0);
\node[] at (8.2, -0.5) {\color{black}$c_{i, j, 4}$};
\draw[black, line width=0.02*1cm/1pt]   (9.12, -0.3) -- (9.12, 0);
\node[] at (9.2, -0.5) {\color{black}$c_{i, j, 4}'$};
\draw[black, line width=0.02*1cm/1pt]   (10.98, -0.3) -- (10.98, 0);
\node[] at (10.7, -0.5) {\color{black}$c_{i, j, 5}$};
\draw[black, line width=0.02*1cm/1pt]   (11.780000000000001, -0.3) -- (11.780000000000001, 0);
\node[] at (11.600000000000001, -0.5) {\color{black}$c_{i, j, 5}'$};
\draw[black, opacity=0.5, line width=2.7pt]   (0.25, 3.0) -- (12.0, 3.0);
\draw[black, opacity=0.5, line width=0.05*1cm/1pt]   (0.225, 3.1) -- (0.225, 2.9);
\fill[black, opacity=0.5] (12.15, 3.0) circle (1.5pt);
\fill[black, opacity=0.5] (12.3, 3.0) circle (1.5pt);
\fill[black, opacity=0.5] (12.45, 3.0) circle (1.5pt);
\draw[blue, opacity=0.5, line width=2.7pt]   (0.0, 2.5) -- (12.0, 2.5);
\fill[blue, opacity=0.5] (-0.15, 2.5) circle (1.5pt);
\fill[blue, opacity=0.5] (-0.3, 2.5) circle (1.5pt);
\fill[blue, opacity=0.5] (-0.44999999999999996, 2.5) circle (1.5pt);
\fill[blue, opacity=0.5] (12.15, 2.5) circle (1.5pt);
\fill[blue, opacity=0.5] (12.3, 2.5) circle (1.5pt);
\fill[blue, opacity=0.5] (12.45, 2.5) circle (1.5pt);
\draw[blue, opacity=0.5, line width=2.7pt]   (0.0, 2.2) -- (12.0, 2.2);
\fill[blue, opacity=0.5] (-0.15, 2.2) circle (1.5pt);
\fill[blue, opacity=0.5] (-0.3, 2.2) circle (1.5pt);
\fill[blue, opacity=0.5] (-0.44999999999999996, 2.2) circle (1.5pt);
\fill[blue, opacity=0.5] (12.15, 2.2) circle (1.5pt);
\fill[blue, opacity=0.5] (12.3, 2.2) circle (1.5pt);
\fill[blue, opacity=0.5] (12.45, 2.2) circle (1.5pt);
\draw[blue, opacity=0.5, line width=2.7pt]   (0.0, 1.1) -- (12.0, 1.1);
\fill[blue, opacity=0.5] (-0.15, 1.1) circle (1.5pt);
\fill[blue, opacity=0.5] (-0.3, 1.1) circle (1.5pt);
\fill[blue, opacity=0.5] (-0.44999999999999996, 1.1) circle (1.5pt);
\fill[blue, opacity=0.5] (12.15, 1.1) circle (1.5pt);
\fill[blue, opacity=0.5] (12.3, 1.1) circle (1.5pt);
\fill[blue, opacity=0.5] (12.45, 1.1) circle (1.5pt);
\draw[blue, opacity=0.5, line width=2.7pt]   (0.0, 0.8) -- (12.0, 0.8);
\fill[blue, opacity=0.5] (-0.15, 0.8) circle (1.5pt);
\fill[blue, opacity=0.5] (-0.3, 0.8) circle (1.5pt);
\fill[blue, opacity=0.5] (-0.44999999999999996, 0.8) circle (1.5pt);
\fill[blue, opacity=0.5] (12.15, 0.8) circle (1.5pt);
\fill[blue, opacity=0.5] (12.3, 0.8) circle (1.5pt);
\fill[blue, opacity=0.5] (12.45, 0.8) circle (1.5pt);
\draw[red, opacity=0.5, line width=2.7pt]   (1.05, 0) -- (2.85, 0);
\draw[red, opacity=0.5, line width=0.05*1cm/1pt]   (1.025, 0.1) -- (1.025, -0.1);
\draw[red, opacity=0.5, line width=0.05*1cm/1pt]   (2.875, 0.1) -- (2.875, -0.1);
\draw[red, opacity=0.5, line width=2.7pt]   (3.75, 0) -- (5.55, 0);
\draw[red, opacity=0.5, line width=0.05*1cm/1pt]   (3.725, 0.1) -- (3.725, -0.1);
\draw[red, opacity=0.5, line width=0.05*1cm/1pt]   (5.574999999999999, 0.1) -- (5.574999999999999, -0.1);
\draw[red, opacity=0.5, line width=2.7pt]   (6.449999999999999, 0) -- (8.249999999999998, 0);
\draw[red, opacity=0.5, line width=0.05*1cm/1pt]   (6.425, 0.1) -- (6.425, -0.1);
\draw[red, opacity=0.5, line width=0.05*1cm/1pt]   (8.274999999999999, 0.1) -- (8.274999999999999, -0.1);
\draw[red, opacity=0.5, line width=2.7pt]   (9.15, 0) -- (10.95, 0);
\draw[red, opacity=0.5, line width=0.05*1cm/1pt]   (9.125, 0.1) -- (9.125, -0.1);
\draw[red, opacity=0.5, line width=0.05*1cm/1pt]   (10.975, 0.1) -- (10.975, -0.1);
\draw[black, opacity=0.6, line width=0.03*1cm/1pt]   (0.2, -1) -- (0.2, 3.9);
\draw[black, opacity=0.6, line width=0.03*1cm/1pt]   (0.2, -1) -- (11.8, -1);
\draw[black, opacity=0.6, line width=0.03*1cm/1pt]   (0.2, 3.9) -- (11.8, 3.9);
\draw[black, opacity=0.6, line width=0.03*1cm/1pt]   (11.8, -1) -- (11.8, 3.9);
\node[] at (-0.39999999999999997, 3) {$e_{i, j}$};
\node[] at (6.0, 3.4) {\Large $C_{i, j}$};
\draw[<->] (3.7, 0.15) -- (5.6, 0.15) node[midway, above] {$v_j$};
\draw[<->] (5.6, 0.15) -- (6.3999999999999995, 0.15) node[midway, above] {$v_j'$};
\end{tikzpicture}
    \caption{Illustration of the critical position for an element agent $e_{i,j}$, shown here for $p_i = 5$. Helper agents are marked in blue, partition agents in red, and $e_{i,j}$ in black. In the construction, there are $p_i = 5$ gaps between the partition agents. Either $e_{i,j}$ alone transports the package completely from $c_{i,j,1}$ to $c_{i,j,5}'$, or five helper agents are present -- each positioned at a gap -- who assist the partition agents with transport. If neither $e_{i,j}$ nor at least five helper agents are present, a helper agent must move between gaps, which we refer to as a C-error.}
    \label{fig:line:sp:crit}
\end{figure}
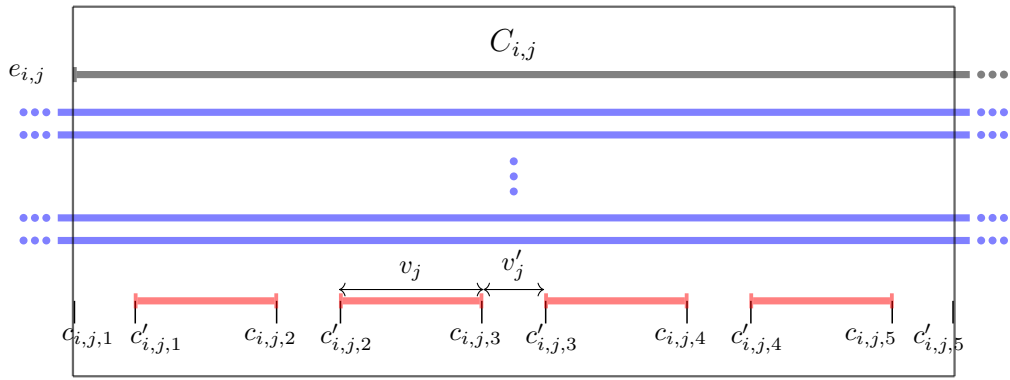
A partition gadget $P_j$ consists solely of the critical positions $C_{1,j}, \dots, C_{n,j}$, which are spaced apart by a large distance of $v^*$. To bridge these gaps, we create one agent with speed $v^*$ between each pair -- referred to as \emph{P-transition agents}. This ensures that element agents (with speed $v_j$) cannot be used at multiple critical positions, since doing so would require a delivery time of $\frac{v^*}{v_j} > d \cdot a(n)$, which we refer to as a \emph{P-error}. We now make use of the pyramidal structure: since the objects satisfy $p_1 > p_2 > \dots > p_n$, we may assume that a drone $e_{i,j}$ can only be located at its own critical position.

The critical position of a drone $e_{i',j}$ with $i' < i$ lies outside the movement interval of $e_{i,j}$. If $e_{i,j}$ is located at the critical position of some $e_{i'',j}$ with $i'' > i$, we can apply an exchange argument: assume that $e_{i,j}$ is the element agent with the largest index $i$ that is placed at a critical position $C_{i'',j}$ for some $i'' > i$. Since $C_{i,j}$ must then be covered by helper agents, we can instead place $e_{i,j}$ at $C_{i,j}$ and use the helper agents to cover $C_{i'',j}$.
As a result, both critical positions are still covered, and additional helper agents are freed, since $p_i > p_{i''}$. Moreover, the critical positions $C_{i', j'}$ with $j' < j$ lie outside the movement area of $e_{i,j}$, so $e_{i,j}$ cannot be present there either. Additionally, we choose the speeds such that, due to $v_k > v_k' > \dots > v_1 > v_1'$, an element agent $e_{i,j}$ cannot assist at a critical position $C_{i', j''}$ with $j'' > j$, since it would be too slow: $e_{i,j}$ has speed $v_j$, and to bridge a gap in $C_{i', j''}$ of length $v_{j''}'$, it would require time $\frac{v_{j''}'}{v_j} > d \cdot a(n)$. We refer to this as an \emph{E-error}.
For example, in \cref{fig:line:sp:gen}, an E-error would occur if the drone $e_{2,2}$ were used at gadget $P_3$, since $e_{2,2}$ is too slow. Moreover, it cannot reach $P_1$, so it can only be used at $P_2$. The structure of the partition gadget is illustrated in \cref{fig:line:sp:part}.

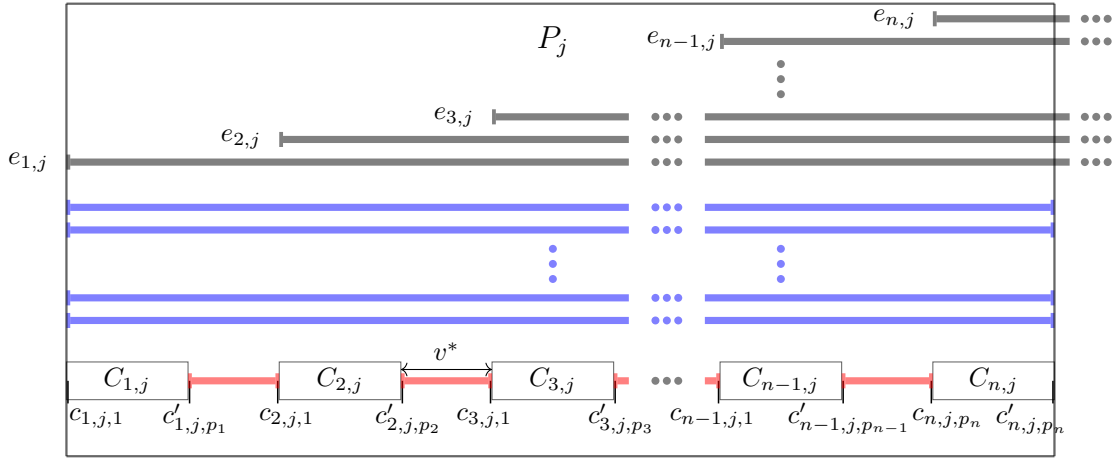
\begin{figure}[h]
    \centering
    \begin{tikzpicture}[scale=1.0]
\definecolor{dgreen}{RGB}{13, 138, 8}
\draw[opacity=0.6] (0.2, -0.25) rectangle (1.8, 0.25);
\node[] at (1.0, 0.0) {\color{black}$C_{1, j}$};
\draw[opacity=0.6] (3.0, -0.25) rectangle (4.6, 0.25);
\node[] at (3.8, 0.0) {\color{black}$C_{2, j}$};
\draw[opacity=0.6] (5.8, -0.25) rectangle (7.4, 0.25);
\node[] at (6.6, 0.0) {\color{black}$C_{3, j}$};
\draw[opacity=0.6] (8.8, -0.25) rectangle (10.4, 0.25);
\node[] at (9.600000000000001, 0.0) {\color{black}$C_{n-1, j}$};
\draw[opacity=0.6] (11.6, -0.25) rectangle (13.2, 0.25);
\node[] at (12.399999999999999, 0.0) {\color{black}$C_{n, j}$};
\fill[blue, opacity=0.5] (6.6, 1.55) circle (1.5pt);
\fill[blue, opacity=0.5] (6.6, 1.75) circle (1.5pt);
\fill[blue, opacity=0.5] (6.6, 1.35) circle (1.5pt);
\fill[blue, opacity=0.5] (9.600000000000001, 1.55) circle (1.5pt);
\fill[blue, opacity=0.5] (9.600000000000001, 1.75) circle (1.5pt);
\fill[blue, opacity=0.5] (9.600000000000001, 1.35) circle (1.5pt);
\fill[black, opacity=0.5] (9.600000000000001, 4.0) circle (1.5pt);
\fill[black, opacity=0.5] (9.600000000000001, 4.2) circle (1.5pt);
\fill[black, opacity=0.5] (9.600000000000001, 3.8) circle (1.5pt);
\draw[black, line width=0.02*1cm/1pt]   (0.22, -0.3) -- (0.22, 0);
\node[] at (0.6000000000000001, -0.5) {\color{black}$c_{1, j, 1}$};
\draw[black, line width=0.02*1cm/1pt]   (1.82, -0.3) -- (1.82, 0);
\node[] at (1.9000000000000001, -0.5) {\color{black}$c_{1, j, p_1}'$};
\draw[black, line width=0.02*1cm/1pt]   (2.98, -0.3) -- (2.98, 0);
\node[] at (3.1, -0.5) {\color{black}$c_{2, j, 1}$};
\draw[black, line width=0.02*1cm/1pt]   (4.619999999999999, -0.3) -- (4.619999999999999, 0);
\node[] at (4.699999999999999, -0.5) {\color{black}$c_{2, j, p_2}'$};
\draw[black, line width=0.02*1cm/1pt]   (5.78, -0.3) -- (5.78, 0);
\node[] at (5.7, -0.5) {\color{black}$c_{3, j, 1}$};
\draw[black, line width=0.02*1cm/1pt]   (7.42, -0.3) -- (7.42, 0);
\node[] at (7.5, -0.5) {\color{black}$c_{3, j, p_3}'$};
\draw[black, line width=0.02*1cm/1pt]   (8.780000000000001, -0.3) -- (8.780000000000001, 0);
\node[] at (8.700000000000001, -0.5) {\color{black}$c_{n-1, j, 1}$};
\draw[black, line width=0.02*1cm/1pt]   (10.42, -0.3) -- (10.42, 0);
\node[] at (10.5, -0.5) {\color{black}$c_{n-1, j, p_{n-1}}'$};
\draw[black, line width=0.02*1cm/1pt]   (11.58, -0.3) -- (11.58, 0);
\node[] at (11.799999999999999, -0.5) {\color{black}$c_{n, j, p_n}$};
\draw[black, line width=0.02*1cm/1pt]   (13.18, -0.3) -- (13.18, 0);
\node[] at (12.899999999999999, -0.5) {\color{black}$c_{n, j, p_n}'$};
\draw[black, opacity=0.5, line width=2.7pt]   (8.850000000000001, 4.5) -- (13.399999999999999, 4.5);
\draw[black, opacity=0.5, line width=0.05*1cm/1pt]   (8.825000000000001, 4.6) -- (8.825000000000001, 4.4);
\draw[black, opacity=0.5, line width=2.7pt]   (8.600000000000001, 3.5) -- (13.399999999999999, 3.5);
\draw[black, opacity=0.5, line width=2.7pt]   (8.600000000000001, 3.2) -- (13.399999999999999, 3.2);
\draw[black, opacity=0.5, line width=2.7pt]   (8.600000000000001, 2.9) -- (13.399999999999999, 2.9);
\draw[black, opacity=0.5, line width=2.7pt]   (5.85, 3.5) -- (7.6000000000000005, 3.5);
\draw[black, opacity=0.5, line width=0.05*1cm/1pt]   (5.825, 3.6) -- (5.825, 3.4);
\draw[black, opacity=0.5, line width=2.7pt]   (3.05, 3.2) -- (7.6000000000000005, 3.2);
\draw[black, opacity=0.5, line width=0.05*1cm/1pt]   (3.025, 3.3000000000000003) -- (3.025, 3.1);
\draw[black, opacity=0.5, line width=2.7pt]   (0.25, 2.9) -- (7.6000000000000005, 2.9);
\draw[black, opacity=0.5, line width=0.05*1cm/1pt]   (0.225, 3.0) -- (0.225, 2.8);
\draw[black, opacity=0.5, line width=2.7pt]   (11.65, 4.8) -- (13.399999999999999, 4.8);
\draw[black, opacity=0.5, line width=0.05*1cm/1pt]   (11.625, 4.8999999999999995) -- (11.625, 4.7);
\draw[blue, opacity=0.5, line width=2.7pt]   (0.25, 2.3) -- (7.6000000000000005, 2.3);
\draw[blue, opacity=0.5, line width=0.05*1cm/1pt]   (0.225, 2.4) -- (0.225, 2.1999999999999997);
\draw[blue, opacity=0.5, line width=2.7pt]   (0.25, 2.0) -- (7.6000000000000005, 2.0);
\draw[blue, opacity=0.5, line width=0.05*1cm/1pt]   (0.225, 2.1) -- (0.225, 1.9);
\draw[blue, opacity=0.5, line width=2.7pt]   (0.25, 1.1) -- (7.6000000000000005, 1.1);
\draw[blue, opacity=0.5, line width=0.05*1cm/1pt]   (0.225, 1.2000000000000002) -- (0.225, 1.0);
\draw[blue, opacity=0.5, line width=2.7pt]   (0.25, 0.8) -- (7.6000000000000005, 0.8);
\draw[blue, opacity=0.5, line width=0.05*1cm/1pt]   (0.225, 0.9) -- (0.225, 0.7000000000000001);
\draw[blue, opacity=0.5, line width=2.7pt]   (8.600000000000001, 2.3) -- (13.149999999999999, 2.3);
\draw[blue, opacity=0.5, line width=0.05*1cm/1pt]   (13.174999999999999, 2.4) -- (13.174999999999999, 2.1999999999999997);
\draw[blue, opacity=0.5, line width=2.7pt]   (8.600000000000001, 2.0) -- (13.149999999999999, 2.0);
\draw[blue, opacity=0.5, line width=0.05*1cm/1pt]   (13.174999999999999, 2.1) -- (13.174999999999999, 1.9);
\draw[blue, opacity=0.5, line width=2.7pt]   (8.600000000000001, 1.1) -- (13.149999999999999, 1.1);
\draw[blue, opacity=0.5, line width=0.05*1cm/1pt]   (13.174999999999999, 1.2000000000000002) -- (13.174999999999999, 1.0);
\draw[blue, opacity=0.5, line width=2.7pt]   (8.600000000000001, 0.8) -- (13.149999999999999, 0.8);
\draw[blue, opacity=0.5, line width=0.05*1cm/1pt]   (13.174999999999999, 0.9) -- (13.174999999999999, 0.7000000000000001);
\draw[red, opacity=0.5, line width=2.7pt]   (1.85, 0) -- (2.95, 0);
\draw[red, opacity=0.5, line width=0.05*1cm/1pt]   (1.825, 0.1) -- (1.825, -0.1);
\draw[red, opacity=0.5, line width=0.05*1cm/1pt]   (2.975, 0.1) -- (2.975, -0.1);
\draw[red, opacity=0.5, line width=2.7pt]   (4.6499999999999995, 0) -- (5.75, 0);
\draw[red, opacity=0.5, line width=0.05*1cm/1pt]   (4.625, 0.1) -- (4.625, -0.1);
\draw[red, opacity=0.5, line width=0.05*1cm/1pt]   (5.7749999999999995, 0.1) -- (5.7749999999999995, -0.1);
\draw[red, opacity=0.5, line width=2.7pt]   (7.45, 0) -- (7.6000000000000005, 0);
\draw[red, opacity=0.5, line width=0.05*1cm/1pt]   (7.425000000000001, 0.1) -- (7.425000000000001, -0.1);
\draw[red, opacity=0.5, line width=2.7pt]   (8.600000000000001, 0) -- (8.75, 0);
\draw[red, opacity=0.5, line width=0.05*1cm/1pt]   (8.775, 0.1) -- (8.775, -0.1);
\draw[red, opacity=0.5, line width=2.7pt]   (10.450000000000001, 0) -- (11.549999999999999, 0);
\draw[red, opacity=0.5, line width=0.05*1cm/1pt]   (10.425, 0.1) -- (10.425, -0.1);
\draw[red, opacity=0.5, line width=0.05*1cm/1pt]   (11.575, 0.1) -- (11.575, -0.1);
\node[] at (-0.3, 2.9) {$e_{1, j}$};
\node[] at (2.5, 3.2) {$e_{2, j}$};
\node[] at (5.3, 3.5) {$e_{3, j}$};
\node[] at (8.3, 4.5) {$e_{n-1, j}$};
\node[] at (11.1, 4.8) {$e_{n, j}$};
\fill[black, opacity=0.5] (7.950000000000001, 0) circle (1.5pt);
\fill[black, opacity=0.5] (8.100000000000001, 0) circle (1.5pt);
\fill[black, opacity=0.5] (8.250000000000002, 0) circle (1.5pt);
\fill[blue, opacity=0.5] (7.950000000000001, 0.8) circle (1.5pt);
\fill[blue, opacity=0.5] (8.100000000000001, 0.8) circle (1.5pt);
\fill[blue, opacity=0.5] (8.250000000000002, 0.8) circle (1.5pt);
\fill[blue, opacity=0.5] (7.950000000000001, 1.1) circle (1.5pt);
\fill[blue, opacity=0.5] (8.100000000000001, 1.1) circle (1.5pt);
\fill[blue, opacity=0.5] (8.250000000000002, 1.1) circle (1.5pt);
\fill[blue, opacity=0.5] (7.950000000000001, 2.0) circle (1.5pt);
\fill[blue, opacity=0.5] (8.100000000000001, 2.0) circle (1.5pt);
\fill[blue, opacity=0.5] (8.250000000000002, 2.0) circle (1.5pt);
\fill[blue, opacity=0.5] (7.950000000000001, 2.3) circle (1.5pt);
\fill[blue, opacity=0.5] (8.100000000000001, 2.3) circle (1.5pt);
\fill[blue, opacity=0.5] (8.250000000000002, 2.3) circle (1.5pt);
\fill[black, opacity=0.5] (7.950000000000001, 2.9) circle (1.5pt);
\fill[black, opacity=0.5] (8.100000000000001, 2.9) circle (1.5pt);
\fill[black, opacity=0.5] (8.250000000000002, 2.9) circle (1.5pt);
\fill[black, opacity=0.5] (7.950000000000001, 3.2) circle (1.5pt);
\fill[black, opacity=0.5] (8.100000000000001, 3.2) circle (1.5pt);
\fill[black, opacity=0.5] (8.250000000000002, 3.2) circle (1.5pt);
\fill[black, opacity=0.5] (7.950000000000001, 3.5) circle (1.5pt);
\fill[black, opacity=0.5] (8.100000000000001, 3.5) circle (1.5pt);
\fill[black, opacity=0.5] (8.250000000000002, 3.5) circle (1.5pt);
\fill[black, opacity=0.5] (13.599999999999998, 2.9) circle (1.5pt);
\fill[black, opacity=0.5] (13.749999999999998, 2.9) circle (1.5pt);
\fill[black, opacity=0.5] (13.899999999999999, 2.9) circle (1.5pt);
\fill[black, opacity=0.5] (13.599999999999998, 3.2) circle (1.5pt);
\fill[black, opacity=0.5] (13.749999999999998, 3.2) circle (1.5pt);
\fill[black, opacity=0.5] (13.899999999999999, 3.2) circle (1.5pt);
\fill[black, opacity=0.5] (13.599999999999998, 3.5) circle (1.5pt);
\fill[black, opacity=0.5] (13.749999999999998, 3.5) circle (1.5pt);
\fill[black, opacity=0.5] (13.899999999999999, 3.5) circle (1.5pt);
\fill[black, opacity=0.5] (13.599999999999998, 4.5) circle (1.5pt);
\fill[black, opacity=0.5] (13.749999999999998, 4.5) circle (1.5pt);
\fill[black, opacity=0.5] (13.899999999999999, 4.5) circle (1.5pt);
\fill[black, opacity=0.5] (13.599999999999998, 4.8) circle (1.5pt);
\fill[black, opacity=0.5] (13.749999999999998, 4.8) circle (1.5pt);
\fill[black, opacity=0.5] (13.899999999999999, 4.8) circle (1.5pt);
\draw[black, opacity=0.6, line width=0.03*1cm/1pt]   (0.2, -1) -- (0.2, 5.0);
\draw[black, opacity=0.6, line width=0.03*1cm/1pt]   (0.2, -1) -- (13.2, -1);
\draw[black, opacity=0.6, line width=0.03*1cm/1pt]   (0.2, 5.0) -- (13.2, 5.0);
\draw[black, opacity=0.6, line width=0.03*1cm/1pt]   (13.2, -1) -- (13.2, 5.0);
\node[] at (6.6, 4.5) {\Large $P_{j}$};
\draw[<->] (4.6, 0.15) -- (5.8, 0.15) node[midway, above] {$v^*$};
\end{tikzpicture}
    \caption{Illustration of a general partition gadget $P_j$, consisting of the critical positions $C_{1,j}, \dots, C_{n,j}$, with P-transition agents (shown in red) placed between them, each having speed $v^*$. The element agents, shown in black, have speed $v_j$, and for each agent $e_{i,j}$, the left endpoint of its movement interval is located at its critical position $C_{i,j}$. The $P - \frac{P}{k}$ helper agents are shown in blue and have speed $v_j'$. If an element agent moves from one critical position to another, this is referred to as a \emph{P-error}. We assume that each element agent $e_{i,j}$ can only be located at its own critical position $C_{i,j}$.}
    \label{fig:line:sp:part}
\end{figure}
Since each object $p_i$ may be assigned to at most one partition set, at most one of the drones $e_{i,1}, \dots, e_{i,k}$ may be located at its corresponding partition gadget, i.e., at its critical position. To enforce this, we construct an object gadget $O_i$ for each object $p_i$. While the left endpoints of the intervals for $e_{i,1}, e_{i,2}, \dots, e_{i,k}$ lie at their respective partition gadgets $P_1, P_2, \dots, P_k$, their right endpoints are all located at the object gadget $O_i$.

The gadget $O_i$ consists of $k-2$ agents with speed $v^*$, referred to as \emph{O-transition agents}, which, however, are not placed directly adjacent to one another. As a result, $O_i$ contains $k-1$ gaps of length 1, such that one element agent -- $k-1$ in total -- must be present at each gap. The gaps are separated by distance $v^*$, so no element agent can move between them; doing so would require at least $\frac{v^*}{v_j}$ time and thus exceed the time bound $a(n) \cdot d$. We refer to this as an \emph{O-error}.

\begin{figure}[ht]
    \centering
    \begin{tikzpicture}[scale=1.0]
\definecolor{dgreen}{RGB}{13, 138, 8}
\fill[black, opacity=0.5] (3.0000000000000004, 1.55) circle (1.5pt);
\fill[black, opacity=0.5] (3.0000000000000004, 1.75) circle (1.5pt);
\fill[black, opacity=0.5] (3.0000000000000004, 1.35) circle (1.5pt);
\fill[black, opacity=0.5] (8.600000000000001, 1.55) circle (1.5pt);
\fill[black, opacity=0.5] (8.600000000000001, 1.75) circle (1.5pt);
\fill[black, opacity=0.5] (8.600000000000001, 1.35) circle (1.5pt);
\draw[black, line width=0.02*1cm/1pt]   (0.22, -0.3) -- (0.22, 0);
\node[] at (0.30000000000000004, -0.5) {\color{black}$o_{i, 1}$};
\draw[black, line width=0.02*1cm/1pt]   (0.6200000000000001, -0.3) -- (0.6200000000000001, 0);
\node[] at (0.9000000000000001, -0.5) {\color{black}$o_{i, 1}'$};
\draw[black, line width=0.02*1cm/1pt]   (2.7800000000000002, -0.3) -- (2.7800000000000002, 0);
\node[] at (2.7, -0.5) {\color{black}$o_{i, 2}$};
\draw[black, line width=0.02*1cm/1pt]   (3.22, -0.3) -- (3.22, 0);
\node[] at (3.3000000000000003, -0.5) {\color{black}$o_{i, 2}'$};
\draw[black, line width=0.02*1cm/1pt]   (5.380000000000001, -0.3) -- (5.380000000000001, 0);
\node[] at (5.300000000000001, -0.5) {\color{black}$o_{i, 3}$};
\draw[black, line width=0.02*1cm/1pt]   (5.82, -0.3) -- (5.82, 0);
\node[] at (5.9, -0.5) {\color{black}$o_{i, 3}'$};
\draw[black, line width=0.02*1cm/1pt]   (7.1800000000000015, -0.3) -- (7.1800000000000015, 0);
\node[] at (7.000000000000001, -0.5) {\color{black}$o_{i, k-2}$};
\draw[black, line width=0.02*1cm/1pt]   (7.620000000000001, -0.3) -- (7.620000000000001, 0);
\node[] at (8.05, -0.5) {\color{black}$o_{i, k-2}'$};
\draw[black, line width=0.02*1cm/1pt]   (9.780000000000001, -0.3) -- (9.780000000000001, 0);
\node[] at (9.4, -0.5) {\color{black}$o_{i, k-1}$};
\draw[black, line width=0.02*1cm/1pt]   (10.180000000000001, -0.3) -- (10.180000000000001, 0);
\node[] at (10.4, -0.5) {\color{black}$o_{i, k-1}'$};
\draw[black, opacity=0.9, line width=2.7pt]   (0.0, 2.6) -- (6.000000000000001, 2.6);
\fill[black, opacity=0.9] (-0.15, 2.6) circle (1.5pt);
\fill[black, opacity=0.9] (-0.3, 2.6) circle (1.5pt);
\fill[black, opacity=0.9] (-0.44999999999999996, 2.6) circle (1.5pt);
\draw[black, opacity=0.65, line width=2.7pt]   (0.0, 2.3) -- (6.000000000000001, 2.3);
\fill[black, opacity=0.65] (-0.15, 2.3) circle (1.5pt);
\fill[black, opacity=0.65] (-0.3, 2.3) circle (1.5pt);
\fill[black, opacity=0.65] (-0.44999999999999996, 2.3) circle (1.5pt);
\draw[black, opacity=0.5, line width=2.7pt]   (0.0, 2.0) -- (6.000000000000001, 2.0);
\fill[black, opacity=0.5] (-0.15, 2.0) circle (1.5pt);
\fill[black, opacity=0.5] (-0.3, 2.0) circle (1.5pt);
\fill[black, opacity=0.5] (-0.44999999999999996, 2.0) circle (1.5pt);
\draw[black, opacity=0.25, line width=2.7pt]   (0.0, 1.1) -- (6.000000000000001, 1.1);
\fill[black, opacity=0.25] (-0.15, 1.1) circle (1.5pt);
\fill[black, opacity=0.25] (-0.3, 1.1) circle (1.5pt);
\fill[black, opacity=0.25] (-0.44999999999999996, 1.1) circle (1.5pt);
\draw[black, opacity=0.15, line width=2.7pt]   (0.0, 0.8) -- (6.000000000000001, 0.8);
\fill[black, opacity=0.15] (-0.15, 0.8) circle (1.5pt);
\fill[black, opacity=0.15] (-0.3, 0.8) circle (1.5pt);
\fill[black, opacity=0.15] (-0.44999999999999996, 0.8) circle (1.5pt);
\draw[black, opacity=0.9, line width=2.7pt]   (7.000000000000001, 2.6) -- (10.15, 2.6);
\draw[black, opacity=0.9, line width=0.05*1cm/1pt]   (10.175, 2.7) -- (10.175, 2.5);
\draw[black, opacity=0.65, line width=2.7pt]   (7.000000000000001, 2.3) -- (10.15, 2.3);
\draw[black, opacity=0.65, line width=0.05*1cm/1pt]   (10.175, 2.4) -- (10.175, 2.1999999999999997);
\draw[black, opacity=0.5, line width=2.7pt]   (7.000000000000001, 2.0) -- (10.15, 2.0);
\draw[black, opacity=0.5, line width=0.05*1cm/1pt]   (10.175, 2.1) -- (10.175, 1.9);
\draw[black, opacity=0.25, line width=2.7pt]   (7.000000000000001, 1.1) -- (10.15, 1.1);
\draw[black, opacity=0.25, line width=0.05*1cm/1pt]   (10.175, 1.2000000000000002) -- (10.175, 1.0);
\draw[black, opacity=0.15, line width=2.7pt]   (7.000000000000001, 0.8) -- (10.15, 0.8);
\draw[black, opacity=0.15, line width=0.05*1cm/1pt]   (10.175, 0.9) -- (10.175, 0.7000000000000001);
\draw[red, opacity=0.5, line width=2.7pt]   (0.6500000000000001, 0) -- (2.7500000000000004, 0);
\draw[red, opacity=0.5, line width=0.05*1cm/1pt]   (0.6250000000000001, 0.1) -- (0.6250000000000001, -0.1);
\draw[red, opacity=0.5, line width=0.05*1cm/1pt]   (2.7750000000000004, 0.1) -- (2.7750000000000004, -0.1);
\draw[red, opacity=0.5, line width=2.7pt]   (3.25, 0) -- (5.3500000000000005, 0);
\draw[red, opacity=0.5, line width=0.05*1cm/1pt]   (3.225, 0.1) -- (3.225, -0.1);
\draw[red, opacity=0.5, line width=0.05*1cm/1pt]   (5.375, 0.1) -- (5.375, -0.1);
\draw[red, opacity=0.5, line width=2.7pt]   (5.8500000000000005, 0) -- (6.000000000000001, 0);
\draw[red, opacity=0.5, line width=0.05*1cm/1pt]   (5.825000000000001, 0.1) -- (5.825000000000001, -0.1);
\draw[red, opacity=0.5, line width=2.7pt]   (7.000000000000001, 0) -- (7.150000000000001, 0);
\draw[red, opacity=0.5, line width=0.05*1cm/1pt]   (7.175000000000001, 0.1) -- (7.175000000000001, -0.1);
\draw[red, opacity=0.5, line width=2.7pt]   (7.650000000000001, 0) -- (9.75, 0);
\draw[red, opacity=0.5, line width=0.05*1cm/1pt]   (7.625000000000002, 0.1) -- (7.625000000000002, -0.1);
\draw[red, opacity=0.5, line width=0.05*1cm/1pt]   (9.775, 0.1) -- (9.775, -0.1);
\fill[red, opacity=0.5] (6.500000000000001, 0) circle (1.5pt);
\fill[red, opacity=0.5] (6.3500000000000005, 0) circle (1.5pt);
\fill[red, opacity=0.5] (6.650000000000001, 0) circle (1.5pt);
\fill[black, opacity=0.15] (6.500000000000001, 0.8) circle (1.5pt);
\fill[black, opacity=0.15] (6.3500000000000005, 0.8) circle (1.5pt);
\fill[black, opacity=0.15] (6.650000000000001, 0.8) circle (1.5pt);
\fill[black, opacity=0.5] (6.500000000000001, 1.1) circle (1.5pt);
\fill[black, opacity=0.5] (6.3500000000000005, 1.1) circle (1.5pt);
\fill[black, opacity=0.5] (6.650000000000001, 1.1) circle (1.5pt);
\fill[black, opacity=0.5] (6.500000000000001, 2.0) circle (1.5pt);
\fill[black, opacity=0.5] (6.3500000000000005, 2.0) circle (1.5pt);
\fill[black, opacity=0.5] (6.650000000000001, 2.0) circle (1.5pt);
\fill[black, opacity=0.65] (6.500000000000001, 2.3) circle (1.5pt);
\fill[black, opacity=0.65] (6.3500000000000005, 2.3) circle (1.5pt);
\fill[black, opacity=0.65] (6.650000000000001, 2.3) circle (1.5pt);
\fill[black, opacity=0.85] (6.500000000000001, 2.6) circle (1.5pt);
\fill[black, opacity=0.85] (6.3500000000000005, 2.6) circle (1.5pt);
\fill[black, opacity=0.85] (6.650000000000001, 2.6) circle (1.5pt);
\node[] at (10.8, 0.8) {$e_{i, 1}$};
\node[] at (10.8, 1.1) {$e_{i, 2}$};
\node[] at (10.8, 2.0) {$e_{i, k-2}$};
\node[] at (10.8, 2.3) {$e_{i, k-1}$};
\node[] at (10.8, 2.6) {$e_{i, k}$};
\draw[black, opacity=0.6, line width=0.03*1cm/1pt]   (0.2, -1) -- (0.2, 4.0);
\draw[black, opacity=0.6, line width=0.03*1cm/1pt]   (0.2, -1) -- (10.200000000000001, -1);
\draw[black, opacity=0.6, line width=0.03*1cm/1pt]   (0.2, 4.0) -- (10.200000000000001, 4.0);
\draw[black, opacity=0.6, line width=0.03*1cm/1pt]   (10.200000000000001, -1) -- (10.200000000000001, 4.0);
\node[] at (5.6000000000000005, 3.6) {\Large $O_{i}$};
\draw[<->] (3.2, 0.15) -- (5.4, 0.15) node[midway, above] {$v^*$};
\draw[<->] (5.4, 0.15) -- (5.800000000000001, 0.15) node[midway, above] {$1$};
\end{tikzpicture}
    \caption{Illustration of an object gadget $O_i$, consisting of $k-2$ O-transition agents shown in red, each with speed $v^*$. There are small gaps between them that must be covered by element agents, shown in black, where darker agents indicate higher speed. In total, there are $k-1$ gaps; if fewer than $k-1$ element agents are present, one must move from one gap to another -- this is referred to as an O-error.}
    \label{fig:line:sp:obj}
\end{figure}
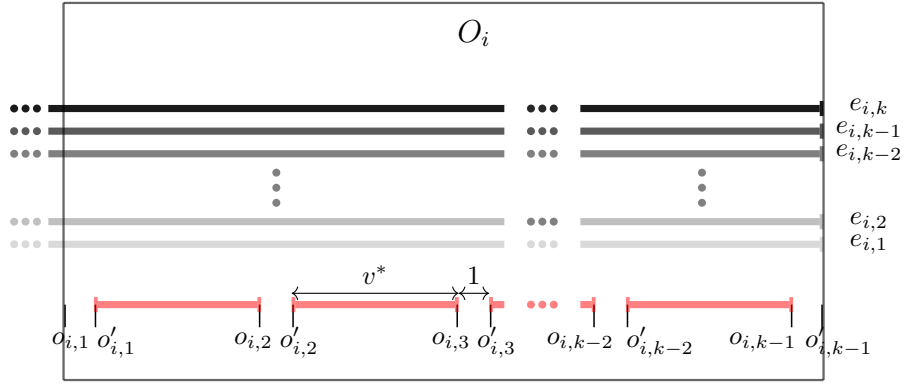

\begin{figure}[ht]
    \centering
    \begin{tikzpicture}[scale=1.0]
\definecolor{dgreen}{RGB}{13, 138, 8}
\draw[red, opacity=0.4, line width=2.7pt]   (0.0, 0) -- (0.15000000000000002, 0);
\draw[red, opacity=0.4, line width=0.05*1cm/1pt]   (0.17500000000000002, 0.1) -- (0.17500000000000002, -0.1);
\draw[red, opacity=0.4, line width=2.7pt]   (1.75, 0) -- (1.9, 0);
\draw[red, opacity=0.4, line width=0.05*1cm/1pt]   (1.7249999999999999, 0.1) -- (1.7249999999999999, -0.1);
\draw[red, opacity=0.4, line width=2.7pt]   (2.9, 0) -- (3.0500000000000003, 0);
\draw[red, opacity=0.4, line width=0.05*1cm/1pt]   (3.075, 0.1) -- (3.075, -0.1);
\draw[red, opacity=0.4, line width=2.7pt]   (4.6499999999999995, 0) -- (4.8, 0);
\draw[red, opacity=0.4, line width=0.05*1cm/1pt]   (4.625, 0.1) -- (4.625, -0.1);
\draw[red, opacity=0.4, line width=2.7pt]   (5.8, 0) -- (5.95, 0);
\draw[red, opacity=0.4, line width=0.05*1cm/1pt]   (5.975, 0.1) -- (5.975, -0.1);
\draw[red, opacity=0.4, line width=2.7pt]   (7.55, 0) -- (7.7, 0);
\draw[red, opacity=0.4, line width=0.05*1cm/1pt]   (7.525, 0.1) -- (7.525, -0.1);
\draw[red, opacity=0.4, line width=2.7pt]   (8.7, 0) -- (8.899999999999999, 0);
\draw[red, opacity=0.4, line width=2.7pt]   (10.45, 0) -- (10.599999999999998, 0);
\draw[red, opacity=0.4, line width=0.05*1cm/1pt]   (10.424999999999999, 0.1) -- (10.424999999999999, -0.1);
\draw[black, opacity=0.6, line width=2.7pt]   (3.35, 1.5) -- (10.349999999999998, 1.5);
\draw[black, opacity=0.6, line width=0.05*1cm/1pt]   (3.325, 1.6) -- (3.325, 1.4);
\draw[black, opacity=0.6, line width=0.05*1cm/1pt]   (10.374999999999998, 1.6) -- (10.374999999999998, 1.4);
\draw[black, opacity=0.4, line width=2.7pt]   (0.35000000000000003, 1.0) -- (10.349999999999998, 1.0);
\draw[black, opacity=0.4, line width=0.05*1cm/1pt]   (0.32500000000000007, 1.1) -- (0.32500000000000007, 0.9);
\draw[black, opacity=0.4, line width=0.05*1cm/1pt]   (10.374999999999998, 1.1) -- (10.374999999999998, 0.9);
\draw[black, opacity=0.4, line width=2.7pt]   (0.55, 0.5) -- (7.45, 0.5);
\draw[black, opacity=0.4, line width=0.05*1cm/1pt]   (0.525, 0.6) -- (0.525, 0.4);
\draw[black, opacity=0.4, line width=0.05*1cm/1pt]   (7.475, 0.6) -- (7.475, 0.4);
\draw[opacity=0.6] (0.2, -0.2) rectangle (1.7, 3.1);
\node at (0.95, 2.5)  {\Large $P_{j'}$};
\draw[opacity=0.6] (3.1, -0.2) rectangle (4.6, 3.1);
\node at (3.8499999999999996, 2.5)  {\Large $P_j$};
\draw[opacity=0.6] (6.0, -0.2) rectangle (7.5, 3.1);
\node at (6.75, 2.5)  {\Large $O_{i}$};
\draw[opacity=0.6] (8.899999999999999, -0.2) rectangle (10.399999999999999, 3.1);
\node at (9.649999999999999, 2.5)  {\Large $O_{i'}$};
\fill[black, opacity=0.5] (2.4, 0) circle (1.5pt);
\fill[black, opacity=0.5] (2.1999999999999997, 0) circle (1.5pt);
\fill[black, opacity=0.5] (2.6, 0) circle (1.5pt);
\fill[black, opacity=0.5] (5.3, 0) circle (1.5pt);
\fill[black, opacity=0.5] (5.1, 0) circle (1.5pt);
\fill[black, opacity=0.5] (5.5, 0) circle (1.5pt);
\fill[black, opacity=0.5] (8.2, 0) circle (1.5pt);
\fill[black, opacity=0.5] (7.999999999999999, 0) circle (1.5pt);
\fill[black, opacity=0.5] (8.399999999999999, 0) circle (1.5pt);
\fill[red, opacity=0.5] (-0.3, 0) circle (1.5pt);
\fill[red, opacity=0.5] (-0.44999999999999996, 0) circle (1.5pt);
\fill[red, opacity=0.5] (-0.15, 0) circle (1.5pt);
\fill[red, opacity=0.5] (10.899999999999999, 0) circle (1.5pt);
\fill[red, opacity=0.5] (10.749999999999998, 0) circle (1.5pt);
\fill[red, opacity=0.5] (11.049999999999999, 0) circle (1.5pt);
\node at (-0.09999999999999998, 0.5)  {$e_{i, j'}$};
\node at (-0.09999999999999998, 1.0)  {$e_{i', j'}$};
\node at (2.6, 1.5)  {$e_{i', j}$};
\draw[->, thick, bend left=0, >=Triangle, line width=0.8pt, shorten >=0.4mm, shorten <=0.7mm] (0.95, 0.5) -- (0.95, 0.25) --(6.75, 0.25) -- (6.75, 0.5);
\draw[line width=1pt] (0.85, 0.4) -- (1.05, 0.6);
\draw[line width=1pt] (0.85, 0.6) -- (1.05, 0.4);
\draw[->, thick, bend left=0, >=Triangle, line width=0.8pt, shorten >=0.4mm, shorten <=0.7mm] (9.649999999999999, 1) -- (9.649999999999999, 0.75) --(0.95, 0.75) -- (0.95, 1);
\draw[line width=1pt] (9.549999999999999, 0.9) -- (9.749999999999998, 1.1);
\draw[line width=1pt] (9.549999999999999, 1.1) -- (9.749999999999998, 0.9);
\draw[->, thick, bend left=0, >=Triangle, line width=0.8pt, shorten >=0.4mm, shorten <=0.7mm] (6.75, 1.5) -- (6.75, 1.25) --(9.649999999999999, 1.25) -- (9.649999999999999, 1.5);
\draw[line width=1pt] (6.65, 1.4) -- (6.85, 1.6);
\draw[line width=1pt] (6.65, 1.6) -- (6.85, 1.4);
\end{tikzpicture}
    \caption{Visualization of the ring exchange argument. Crosses indicate the original positions of the agents, and arrows show their new positions as part of the cyclic exchange. Since $e_{i', j'}$ and $e_{i, j'}$ have the same speed and $i' < i$, the movement interval of $e_{i, j'}$ is entirely contained within that of $e_{i', j'}$, allowing $e_{i', j'}$ to take the place of $e_{i, j'}$. In turn, $e_{i', j}$ can move to the position of $e_{i', j'}$, and $e_{i, j'}$ to the original position of $e_{i', j}$, i.e., to $O_i$.
}
    \label{fig:line:sp:ring}
\end{figure}
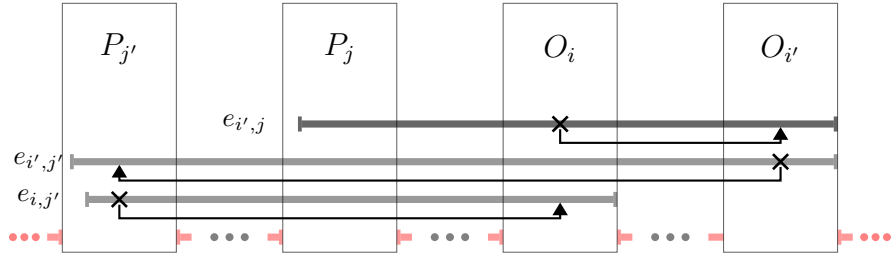
The object gadgets are arranged in mirrored order to the critical positions within a partition gadget, i.e., from left to right we place $O_n$, then $O_{n-1}$, down to $O_1$. We then use an exchange argument to show that, without loss of generality, no element agent can be assigned to the wrong object gadget. Suppose at $O_i$ there is an agent $e_{i',j}$ with $i \neq i'$. Since for $i < i'$ the gadget $O_i$ lies outside the movement range of $e_{i',j}$, it must be that $i > i'$. If this occurs for multiple gadgets, consider only the one with the smallest index (the rightmost). There must be at least two of the agents $e_{i,1},\dots,e_{i,k}$ not at $O_i$, while exactly $k-1$ of $e_{i',1},\dots,e_{i',k}$ are at $O_{i'}$. Hence there is some $j'$ such that $e_{i,j'}$ is not at $O_i$ and $e_{i',j'}$ is at $O_{i'}$. Since $e_{i,j'}$ and $e_{i',j'}$ share the same speed and the movement range of $e_{i,j'}$ is contained in that of $e_{i',j'}$, we perform the cyclic exchange: place $e_{i',j}$ at the original position of $e_{i',j'}$, $e_{i',j'}$ at that of $e_{i,j'}$, and $e_{i,j'}$ at that of $e_{i,j}$. Iterating this swap removes any element agent $e_{i',j}$ with $i \neq i'$ from all object gadgets.

Since this exchange is somewhat intricate, it is illustrated in \cref{fig:line:sp:ring}. With this all in hand we can now formulate the resulting DDT-SP instance and proof correctness. As it is rather technical, we defer it to \cref{arxivappendixA}.

\subsection{A Parameterized Algorithm}\label{sec:line:fpt}

We now design an FPT algorithm for DDT-SP on path graphs. As a parameter for our approach, we use the maximum overlap -- the \emph{thickness} $\tn := \max_{u \in V(G)} |B_u|$, i.e., the largest number of agents covering any single vertex.

Since each agent corresponds to an interval on the line, the intersection graph of agents is an interval graph. In such graphs, agents covering a common point form a clique, and each maximal clique corresponds to some point on the line \cite{lekkeikerker62}. As interval graphs are chordal, their treewidth equals the size of the largest clique minus one \cite{rose76}, which in our case implies $\tw{}=\tn - 1$. Thus, the treewidth of the intersection graph and the thickness parameter are equivalent.
This structural property allows for a dynamic programming algorithm over a tree decomposition of the intersection graph. We define the \emph{event points} $e_1 < \dots < e_{2k}$ as the start and end positions of all agent intervals, i.e., the multiset $\{s_a, f_a \mid a \in A\}$. We sort these events from left to right along the path ($e_1 < \dots < e_{2k}$), resolving ties by placing start points before end points; the order among start points or among end points is chosen arbitrarily. In an optimal schedule, handovers occur only at these event points -- specifically at $s_{a'}$ or $f_a$ for a handover between $a$ and $a'$ -- since otherwise extending the faster agent's segment would yield a strictly better or equally good schedule.

We interpret each event $e_i$ as a node in a tree decomposition that forms a path, see \cref{fig:path}. We can consider the set of agents covering $e_i$, $Z_{e_i}$ as the bag at position $i$. If $e_i$ is the start vertex of some agent $a$, we say that $a$ is \emph{introduced} to the bag, and if $e_i$ is the end point of some agent $a$, it is \emph{forgotten}.
We define a dynamic program over the path decomposition induced by this event sequence. At each $e_i$, we maintain:
\begin{itemize}
\item  the \emph{bag} $Z_{e_i}$ of agents covering $e_i$,
\item   a subset $S \subseteq Z_{e_i}$ of agents already used, and
\item   a current carrier $a \in Z_{e_i}$ responsible for transporting the package from $e_i$ to $e_{i+1}$.
\end{itemize}
Since at most $\tn$ agents cover any point, we have $|Z_{e_i}| \leq \tn$, and the number of subsets $S \subseteq Z_{e_i}$ (representing already-used agents) is bounded by $2^\tn$.
The dynamic programming table stores the minimum delivery time to reach $e_i$ under the assumption that agent $a$ will deliver it to $e_{i+1}$ and the current set of already used agents is $S$, denoted by $D[e_i,S,a]$.
We assume no agent starts at $t$ or ends at $s$, as those movement intervals can be cut off without affecting the solution. Furthermore, at most one handover occurs per event point, and it happens only if $e_i = s_{a'}$ or $e_i = f_a$, as argued before.
While we only store delivery times in the DP table, the full delivery sequence can be reconstructed via backtracking.

In the following we always assume for some DP entry $D[e_i, S, a]$ that $a\in S$, otherwise the entry is defined as $\infty$.
Assume we have already computed all possible entries $D[e_{i-1},S\subseteq Z_{e_{i-1}}, a\in Z_{e_{i-1}}]$ correctly, and now wish to compute $D[e_i,S\subseteq Z_{e_i}, a\in Z_{e_i}]$. 

\subparagraph*{Introduce Event.}For an introduce event $e_i$, it is the start point of an agent $a'$, so $a'$ is now newly contained in $Z_{e_i}$. 
We differentiate between $a'$ being included into $S$, which corresponds to the package being handed over to $a'$ and the case where $a'$ is omitted (for now). 
Note that when $a'$ enters $S$ it will also be the dedicated package carrier.

\begin{restatable}{lemma}{lemmaIntroLine}
    Let $e_i$ be an introduce event that introduces some agent $a'$.
    Assume that for all valid combinations the solutions $D[e_{i-1},S'\subseteq Z_{e_{i-1}}, a\in S']$ have been computed correctly, then we correctly compute $D[e_i,S\subseteq Z_{e_i}, a\in S]$ in time $\OO(\tn)$.
\end{restatable}

\begin{proof}    
    First, assume that the newly introduced agent $a'$ is not contained in the set of already used agents $S$, i.e., $a'\notin S$.
    Since $Z_{e_{i}}=Z_{e_{i-1}}\cup \{a'\}$ and $a'$ is assumed to not deliver the package to $e_{i+1}$, $a$ must correspond to an agent from $Z_{e_{i-1}}$.    
    Since we assume that a handover only happens from some agent $a_1$ to $a_2$, if either we reached the end point of $a_1$ or the start point of $a_2$ and neither is true in this case, the package carrier $a$, that delivers the package from $e_i$ to $e_{i+1}$ is the same as the package carrier that delivered it from $e_{i-1}$ to $e_i$, so we must have $a_{i-1}^{\opt} = a$.     
    Therefore, agent $a$ incurred time $d_a(e_{i-1},e_i)$ from event point $e_{i-1}$ to $e_i$.
    At last, since $a$ was delivering the package from $e_{i-1}$ to $e_i$, the solution $\opt_{-i}=((e_1,a_1^{\opt}),\dots,(e_{i-1},a_{i-1}^{\opt}=a))$ with $Z_{e_{i-1}}\cap \opt_{-i}=S$ must be optimal for the instance $(e_{i-1},S,a)$, which, by our previous assumption, is correctly computed for $D[e_{i-1},S,a]$.    
    Therefore, the delivery time at $e_i$ in this case is $D[e_{i-1},S,a]+ d_a(e_{i-1},e_i)$.

    Now assume that the newly introduced agent is the next package carrier, i.e., $a_i^{\opt}=a'$, which implies $a'\in S$ and $a=a'$.
    Let $a''=a_{i-1}^{\opt}$ be the agent which passed the package over to $a'$ at event point $e_i$.
    Since, by the same reason as before, $a''$ could not have handed over the package between $e_{i-1}$ and $e_i$, it was the designated package delivery agent at $e_{i-1}$.
    Therefore, it incurred cost $d_{a''}(e_{i-1},e_i)$ to deliver the package from $e_{i-1}$ to $e_{i}$.
    Additionally, we have $Z_{e_{i-1}}\cap S_{-i}^{\opt} = S\setminus\{a'\}$, as no additional transfer to other agents happened in between $e_{i-1}$ and $e_i$.
    This means that $\opt_{-i}=((e_1,a_1^{\opt}),\dots,(e_{i-1},a_{i-1}^{\opt}=a''))$
    is also an optimal solution for the instance $(e_{i-1},S\setminus \{a'\}, a'')$, which is already correctly computed as $D[e_{i-1},S\setminus \{a'\}, a'']$.
    Therefore, the delivery time in the optimal solution must correspond in this case to $D[e_{i-1},S\setminus \{a'\},a]+ d_a(e_{i-1},e_i)$.

    Since we do not know the previous package carrier $a''$, we try out all possible options of previous package carriers that are contained in $S\setminus \{a'\}$.
    Each such solution corresponds to a valid delivery tour that transports the package from $s$ to $e_i$ such that the set $S$ of agents was already used and $a$ is still the currently assumed package carrier.
    This give us the following final case distinction for setting the values for $D[t,S,a]$:

    \begin{equation*}
    D[e_i,S\subseteq Z_{e_i},a]=\begin{cases}
                D[e_{i-1},S,a] + d_a(e_{i-1},e_i) & a'\notin S \\
                \min\limits_{{a''\in S\setminus\{a'\}}}\big\lbrace D[e_{i-1},S\setminus\{a'\},a''] + d_{a''}(e_{i-1},e_i)\big\rbrace & a' \in S~ (\text{and }a=a')\\
                \infty & \text{otherwise}
                \end{cases}
\end{equation*}    
The runtime is upper bounded by the second case where we iterate over all possible previous agent carriers $a''$ from $S\setminus \{a'\}$. Since $\vert S \vert \leq \tn$, the runtime follows.
\end{proof}

\subparagraph*{Forget Event.}For a forget event $e_i$, an end point for some agent $a'$ was reached, thus $a'$ is removed from $Z_{e_i}$.
We consider all possibilities for how $a'$ could have been involved in an optimal delivery: either $a'$ delivered the package to $e_i$ and now hands it over to a different agent $a$, or $a'$ was involved in the optimal delivery but did not carry the package to $e_i$, or $a'$ never held the package at any point.
Out of all of these options we again take the solution variant with the smallest cost.

\begin{restatable}{lemma}{lemmaForgetLine}
    Let $e_i$ be a forget event that forgets some agent $a'$.
    Assume that for all valid combinations the solutions $D[e_{i-1},S'\subseteq Z_{e_{i-1}}, a\in S']$ have been computed correctly, then we correctly compute $D[e_i,S\subseteq Z_{e_i}, a\in S]$ in time $\OO(1)$.
\end{restatable}

\begin{proof}
    Let $a''$ be the package carrier that held the package between event point $e_{i-1}$ and $e_i$.
    First assume that $a'$ was never the carrier of the package in the optimal solution, i.e., $a'\notin S^{\opt}$.
    By our observation about package exchanges this means that no handover at event point $e_i$ takes place, so we have $a''=a$ and the currently assumed package carrier $a$ also held the package at event point $e_{i-1}$, i.e., $a_{i-1}^{\opt}=a_i^{\opt}=a$.
    Between $e_{i-1}$ and $e_i$, this agent $a$ incurred cost $d_a(e_{i-1},e_i)$. 
    Since we have $\opt_{-i}=((e_1,a_1^{\opt}),\dots,(e_{i-1},a_{i-1}^{\opt})=a)$ and $Z_{e_{i-1}}\cap \opt_{-i}=S$, the solution $\opt_{-i}$ must also be optimal for the instance $(e_{i-1},S,a)$ with delivery time $D[e_{i-1},S,a]$.
    Therefore, the final delivery time for the schedule $\opt$ must correspond to $D[e_{i-1},S,a] + d_a(e_{i-1},e_i)$.    

    Now assume that $a'$ was used at some point as a package carrier, but was not the recent package carrier $a''$.
    By the same reason as in the case before we have $a''=a$.
    Again, we incur cost $d_a(e_{i-1},e_i)$ between event points $e_{i-1}$ and $e_i$.
    Also, since $a'$ was used at some point, we have $Z_{e_{i-1}}\cap \opt_{-i}=S\cup \{a'\}$, and since $a''$ was the recent carrier, the solution $\opt_{-i}$ must be optimal for the instance $(e_{i-1}, S\cup \{a'\}, a)$ with delivery time $D[e_{i-1}, S\cup \{a'\}, a] + d_a(e_{i-1},e_i)$
 
    At last, assume that $a'$ was the very recent package carrier, i.e., $a''=a'=a_{i-1}^{\opt}$.
    In this case the optimal solution incurred travel time $d_{a'}(e_{i-1},e_i)$ between event points $e_{i-1}$ and $e_i$.
    Additionally, agent $a$ could not have been used before, i.e., $a\notin Z_{e_{i-1}}\cap S_{-i}^{\opt}$.
    This means that the solution $\opt_{-i}=((e_1,a_1^{\opt}),\dots,(e_{i-1},a_{i-1}^{\opt})=a')$
    with $Z_{e_{i-1}}\cap S^{\opt}=(S\cup \{a'\})\setminus \{a\}$ must be optimal for the instance $(e_{i-1}, S\cup \{a'\}, a')$.
    Therefore, the optimal solution incurred travel time $D[e_{i-1}, S\cup \{a'\}, a'] + d_{a'}(e_{i-1},e_i)$.

    Since only one of these three cases can happen in the optimal solution, we simply choose the variant with the smallest incurred cost, i.e., we set    
    \begin{align*}   
    D[e_i,S,a]=& \min \{ \min \{D[e_{i-1}, S, a], D[e_{i-1}, S\cup\{a'\}, a]\} + d_a(e_{i-1},e_i), \\
    &D[e_{i-1}, (S \cup \{a'\})\setminus \{a\}, a'] + d_{a'}(e_{i-1},e_i) \}.
    \end{align*}   
    We can see that each of these variants corresponds to a feasible drone schedule such that we arrive at event point $e_i$ and continue with agent $a$.
    Since each of the variants can be computed in $\OO(1)$, the runtime follows.
\end{proof}

\thmlinefpt*

\begin{proof}
    The correctness follows from the correctness of the different types of events and that the optimal solution must be stored in some entry $D[e_j,S\subseteq Z_{e_j}, a]$ for some final end point event $e_j=e_{2k}$.    
    Since for every type of event we have runtime at most $\OO(\tn)$ and the set of instances at an event point $e_i$ is upper bounded by $2^{\tn}\cdot \tn$ and $2k$ event points exist, we get a final runtime of $\OO(2^{\tn}\cdot \tn^2\cdot k) =\OO(2^{\tw{}}\cdot \tw{}^2 \cdot k)$. We note that technically the entire path needs to be read at least once which yields an additional summand of $n$. However, we omit it as it is unrealistic that it is the dominating term.
\end{proof}
\section{Fixed-Parameter Tractability on General Graphs}\label{sec:graph}
Before designing an FPT algorithm on general graphs, we argue that, w.l.o.g., the intersection graph can be assumed to be simple -- i.e., agents intersect at most once. Any instance can be transformed accordingly, increasing degree and treewidth by only a constant.

\subsection{Unique Intersections between Agents}

Given any DDT instance $(G, (s, t), A)$, we now create an instance $(G', (s', t'), A')$ that has equivalent schedules, but its intersection graph $G_I'$ is simple and therefore any pair of agents overlap in at most one vertex.
Each agent $a \in A$ originally moves on a (possibly overlapping) subgraph $G_a$ of $G$.
To guarantee that no pair of agents share multiple vertices, we create a \emph{disjoint copy} of every $G_a$ and relabel each vertex $u$ as $(u,a)$.
The $k$ original agents are then confined to their respective copies $G_a'$.  
Whenever two subgraphs $G_a$ and $G_b$ share a vertex $u$ in $G$, the transformed graph $G'$ now contains two distinct copies, i.e., $(u,a)$ and $(u,b)$. 
We connect these copies of the same vertex by edges (of length 0) and introduce a dedicated \emph{auxiliary agent} $h_{u}$ with speed $1$ and movement area restricted to all those copies.
Thus, the package can be transferred instantly between any two copies of the same original vertex, while original agents share no common vertex.
The resulting instance preserves all feasible schedules with the same delivery time as the original one but it also satisfies the new constraint that any pair of agents shares at most one vertex, as any pair of original agents $a', b'$ does not share a common vertex in $G'$, neither does any pair of helping agents $h_{u}, h_{u'}$, and agents $a'$ and $h_{u}$ may only intersect in $(u, a)$.
\begin{restatable}{lemma}{lemmaTransfo}\label{thm:unuqueinter}
    For any DDT instance $(G, (s, t), A)$, there exists an equivalent instance $(G', (s', t'), A')$, where $\tw{}(G_I') \leq \tw{}(G_I)+1$, $\Delta(G_I') \leq \Delta(G_I)+1$ and $G_I'$ is simple.
\end{restatable}

\begin{proof}
    Formally, we define
\[V(G') = \{(u, a) \mid u \in V(G_a), a \in A \}\] and \[E(G') = \{\{(u, a), (w,a)\} \mid \{u, w\} \in E(G_a), a \in A \} \cup \{\{(u, a), (u, b)\} \mid u \in V(G_a) \cap V(G_b) \}. \]
For edges of the form $\{(u, a), (w, a)\}$ we define the length as $\ell(u, w)$, corresponding to the original edges in $G_a$. For edges of the form $\{(u, a), (u, b)\}$ -- representing transfers between agents at a shared vertex -- we define the length as $0$ to make instantaneous transfers possible.
For each agent $a \in A$ we define an agent $a'$ with $v_{a'} = v_{a}$ and 
\[G'_{a'} = (\{(u, a) \mid u \in V(G_a)\}, \{\{(u, a), (w, a)\} \mid \{u, w\} \in E(G_a) \}).\]
For each $u \in V(G)$ with at least two agents in its subgraph (i.e. $|B_u|\geq 2$), we introduce an helping agents $h_{u}$. We define the speed $v_{h_u} = 1$, the vertex set $V(G'_{h_{u}}) = \{(u, a)\mid u \in V(G_a), a \in A \}$ and the edge set $E(G'_{h_{u}}) = \{\{u, w\} \mid u, w \in V(G'_{h_{u}}) \}$, therefore we have a complete subgraph. Thus, $A'$ consists of all transformed agents $a'$ corresponding to agents $a \in A$, together with all auxiliary agents. For $s'$, we select an arbitrary node $(u, a) \in V(G')$ such that $u = s$; analogously for $t'$.

We now argue that $tw(G') \leq tw(G)+1$. Given a tree decomposition of $G$, we construct a tree decomposition of $G'$ by incorporating the new helping agents. Each helping agent $h_u$ is connected to every agent $a'$ such that $u \in V(G_a)$. Since all such agents share a node in $G$, they must form a clique in $G_I$. Consequently, there exists a bag in the decomposition of $G$ that contains all of them. We can attach a new bag to the decomposition that includes this clique along with the new agent $h_u$. This ensures that both endpoints of all edges incident to $h_u$ are contained within the bag, and the remaining conditions of a tree decomposition are trivially satisfied. The size of this new bag is at most one greater than that of the original bag containing the clique, and hence $tw(G') \leq tw(G) + 1$ follows.

Furthermore, we show that $\Delta(G_I') \leq \Delta(G_I)+1$. We first establish that $|\delta_{G_I'}(h_u)| \leq \Delta(G_I)+1$ for any auxiliary agent $h_u$. By construction, we have
$
|\delta_{G_I'}(h_u)| = |\{a \in A \mid u \in V(G_a)\}| = |B_u|,
$
and we observe that the subgraph of $G_I$ induced by $B_u$ forms a clique (of size $|B_u|$), as all these agents share the vertex $u$. We chose some agent $b \in B_u$ and can observe that
$
|\delta_{G_I}(b)| \geq |B_u|-1.
$
Therefore, $|\delta_{G_I'}(h_u)| \leq |B_u| \leq |\delta_{G_I}(b)|+1 \leq \Delta(G_I)+1$.

Next, we prove that $|\delta_{G_I'}(a')| \leq \Delta(G_I)+1$ for each agent $a'$. We have
\begin{align*}
|\delta_{G_I'}(a')| =& |\{h_u \mid u \in V(G_a) \cap V(G_b), b \in A,\ a \neq b\}| \\=&|\{u \in V(G_a) \cap V(G_b) \mid b \in A,\ a \neq b\}|.
\end{align*}
Each such $u$ appears in some $G_b$ for $b \neq a$, so the number of such nodes $u$ is bounded by the total number of edges in $G_I$ as there may be intersections with several other agents at the same vertex, so
\begin{align*}    
&|\{u \in V(G_a) \cap V(G_b) \mid b \in A,\ b \neq a\}| \\
\leq &|\{(u, a) \mid b \neq a,\ u \in V(G_a) \cap V(G_b), b \in A\}| \\
=& |\delta_{G_I}(a)|.
\end{align*}
Thus,
$
|\delta_{G_I'}(a')| \leq |\delta_{G_I}(a)|
\leq \Delta(G_I)$
and in total the claim $\Delta(G_I') \leq \Delta(G_I)+1$ follows.

As argued the intersections are unique, so the intersection graph is simple. We now argue that the transformed instance on $G'$ is equivalent to the original instance on $G$ with respect to schedule feasibility and duration.

Given a schedule in $G$, we construct a schedule in $G'$ of the same duration. Each agent $a \in A$ is replaced by $a' \in A'$, moving in the isomorphic subgraph $G'_a$, so the movement can be reused directly.
Whenever a package is transferred between agents $a$ and $b$ at vertex $u \in V(G)$, we can use the helping agent $h_u$ in $G'$. Since $G'_{h_u}$ is a clique over all $(u, a)$, $h_u$ can transfer the package from $(u, a)$ to $(u, b)$ in 0 time, so the schedule remains valid and of the same duration.

Given a feasible schedule in $G'$, we map each agent $a'$ back to its original $a$ in $G$ using the isomorphism between $G'_a$ and $G_a$.

Each transfer via $h_u$ -- from $(u, a)$ to $h_u$ to $(u, b)$ -- corresponds to a direct transfer between $a$ and $b$ at $u$ in $G$. Since the agents must both be present at $u$ in $G'$ for this to occur, the direct transfer in $G$ is valid.

Therefore, every schedule in $G$ maps to one in $G'$ of the same duration, and vice versa. Hence, the two instances are equivalent in feasibility and duration.
\end{proof}

\subsection{Tree Decomposition Algorithm for Unique Intersections}

Let $(G, (s,t), A)$ be a drone delivery instance in which each pair of agents $a,a'\in A$ intersects in at most one vertex.
For an agent $u\in A$, we define $N_u:=\{v\in V(G_I) \mid \{u,v\}\in E(G_I)\}$.
Let $D:A\times A\rightarrow V(G) \cup \{-\}$ be a function that returns for two agents $a_i,a_j$ their unique intersection point $V_i\cap V_j$ if it exists and $-$ otherwise.
Let $G_I$ be the intersection graph, and assume we have a nice tree decomposition $\mathcal{T}=(\mathbb{T},\{X_t\}_{t\in V(\mathbb{T})})$ of $G_I$ of width \tw{} and maximum degree $\dm=\max_{u\in V(G_I)}\vert N_u \vert$.  

We now design a dynamic program over $\mathcal{T}$ that is parameterized by $\tw{}$ and $\dm{}$ and computes the optimal agent tour. 
The algorithm is inspired by standard treewidth-based DP techniques for TSP\footnote{See e.g. \url{http://www.cs.bme.hu/~dmarx/papers/marx-warsaw-fpt2}, page 17.}.
For each node $t\in V(\mathbb{T})$, we process the subgraph $G_I[t]=(V_t,E_t)$.

A delivery tour can be modeled as an ordered sequence $ (a_1,a_2,\dots, a_b)$ of agents.
Each contiguous subsequence $A_{i,j}=(a_i,\dots,a_j)$ induces a valid subpath in $G_t$,
if the corresponding edges $\{a_l,a_{l+1}\}\in E(G_t)$ for all $l\in [i, j-1]$, and $a_i,a_j\in X_t$.
Intermediate agents $a_l$ with $i<l<j$ incur cost $d_{a_l}(D(a_{l-1}, a_l), D(a_l,a_{l+1}))$.

However, to determine the cost of agents $a_i,a_j$ we require knowledge about the corresponding unknown preceding and succeeding agents to determine entry and exit points. 
Thus, the DP state additionally stores for the boundary agents in $X_t$ their assumed predecessor and successor.
This allows consistent cost computation when merging subpaths at join nodes.

We slightly modify the input instance to streamline boundary cases: we add two artificial agents $a_s$ and $a_t$, each covering only vertex $s$ and $t$ respectively, and having zero speed.
The unique intersection property extends naturally to these agents.
We now additionally require any valid delivery tour to start with $a_s$ and end with $a_t$.
Both agents are required to be always contained in any bag, increasing the treewidth by exactly 2.
Note that the actual delivery can still start at $s$, since $a_s$ can immediately hand over the package at no cost.

To simplify boundary handling, we also introduce auxiliary agents $a_s'$ and $a_t'$, each identical to $a_s$ and $a_t$, but always assumed to lie outside the tree decomposition (i.e., they are never contained in any bag).
A delivery tour will now implicitly require the tour to start with $a_s'$ and end with $a_t'$. 
Together, this guarantees that for each boundary agent,
which always includes $a_s$ and $a_t$, its predecessor and successor are explicitly known, enabling us to compute exact cost for partial solutions.

Each DP state $D[t,B_t^0,B_t^1,B_t^2,M,\hat{A}, \hat{Z}]$ represents the optimal forest of paths in subgraph $G_I[t]$ that connects agents according to the following:
\begin{itemize}
    \item $B_t^i\subseteq X_t$ for $i\in \{0,1,2\}$ contains all agents with degree exactly $i$
    \item $M\subseteq B_t^1\times B_t^1$ encodes directed pairwise path endpoints as a matching 
    \item $\hat{A}$ and $\hat{Z}$ assign known predecessor/successor agents to agents in $B_t^1$:
    \begin{itemize}
        \item $\hat{A}$ maps to a neighbor agent s.t. the corresponding edge is not contained in $E_t$
        \item $\hat{Z}$ maps to the neighbor s.t. the corresponding edge is contained in $E_t$
    \end{itemize}
\end{itemize}
Each path $P_j$ in a feasible forest $F$ must satisfy the following conditions:
\begin{itemize}
    \item The endpoints of $P_j$ lie in $B_t^1$ and every agent $a\in B_t^i$ has degree $i$ in $F$
    \item For every pair $(a,a')\in M$, there must be a path $P_j$ in $F$ with both of these as endpoints
\end{itemize}
The cost of a forest $F$ is then evaluated using the previously mentioned reasoning of computing the cost of every individual agent measured by its two neighboring agents in the corresponding sequence.

For the root node $r$ the DP entry $D[r,\emptyset,\{a_s,a_t\},\emptyset, \{(a_s,a_t),(a_t,a_t')\},\hat{Z}]$ captures the optimal path from $a_s$ to $a_t$, if $\hat{Z}$ correctly encodes the adjacent agents to $a_s$ and $a_t$ in an optimal sequence.
By iterating over all valid candidates for the neighbors $u\in N_{a_s}\setminus \{a_s'\}$ and $v\in N_{a_t}\setminus \{a_t'\}$ we determine the optimal solution.

To compute the delivery cost of a sequence of agents, we use the predecessor/successor assignment $\hat{A}$.
For an agent $a$ of degree 1 we write $\hat{A}(a)$ to denote the predecessor (if $a$ is the start of a path) or successor (if $a$ is the end), and its other adjacent agent by $\hat{Z}(a)$ and the partner of $a$ in the matching $M$ by $M(a)$.
We will refer to $\hat{A}(a)$ as the predecessor/successor partner (or \textit{psp} for short) and to $\hat{Z}(a)$ as the next/previous interior partner (or \textit{npip}).

For a feasible forest $S=\{S_1,\dots,S_b\}\subseteq E(G_I[t])$, each component $S_i=(a_1,\dots,a_\ell)$ corresponds to a path of agents starting and ending at agents in $B_t^1$.
Using $\hat{A}$, we extend each such path to the sequence $\hat{S}_i=(\hat{a}_0,a_1,\dots,a_\ell,\hat{a}_{\ell+1})$, where $\hat{a}_0=\hat{A}(a_1)$ and $\hat{a}_{\ell+1}=\hat{A}(a_\ell)$.
These agents define the handover points at the start and end of the path that involves agents $a_1$ and $a_\ell$, which is necessary to compute the full cost contribution of each agent.

The cost of the extended sequence $\hat{S}_i$ is given by

$c(\hat{S}_i) = \sum_{j=1}^{l_i} d_{\hat{a}_j}(D(\hat{a}_{j-1},\hat{a}_j), D(\hat{a}_j,\hat{a}_{j+1})).$
Setting the predecessor of $a_s$ as $a_s'$ and the successor of $a_t$ as $a_t'$ insures sequences starting at $a_s$ or ending at $a_t$ are correctly computed as well. In this case we will incur an additional cost of 0, since these agents must immediately handover the package at their respective vertex.

The total cost of the solution $S$ is the sum over all component sequences, i.e.,
$
    c(S):=\sum_{i=1}^b c(\hat{S}_i)
$.
In the following we write $\hat{A}(S_i)$ for the extended sequence $\hat{S}_i$.
Finally, for any solution $S\subseteq E(G_I)$ and node $t\in V(\mathbb{T})$ we define $S[t]:=S\cap E_t$ and say that $S$ \textit{induces} the instance $(t,B_t^0,B_t^1,B_t^2,M,\hat{A},\hat{Z})$ if 
\begin{itemize}
    \item each $a\in B_t^i$ has degree $i$ in $S_t$,
    \item each pair $(u,v)\in M$ is connected by a path in $S[t]$ with $\hat{A}(u),\hat{A}(v)$ as their psp agents, and
    \item for each npip agent $a$ on a path, the adjacent agent corresponds to $\hat{Z}(a)$.
\end{itemize}
For an example instance and corresponding agent degrees, see \cref{fig:example_instance_1}.
In the following arguments we will for simplicity assume that any DP entry contains the actual forest. 
If each DP entry only contains its cost one can, by simple recursive pointer-logic, compute the optimal solution in the end recursively.
Therefore, we will argue correctness of each entry by considering the actual forest, but derive the computational runtime as if only the cost is saved in every DP entry.

\begin{figure}
    \centering
    \includegraphics[width=0.65\linewidth]{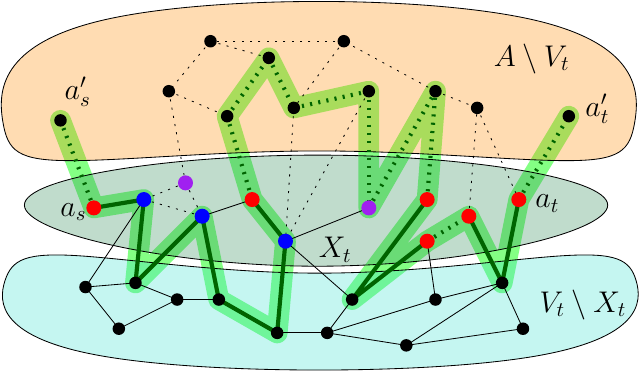}
    \caption{Example instance $G_I$ with corresponding sets of intersections between agents. Agents $a_s,a_t$ are always contained in $X_t$ and $a_s',a_t'$ are always contained in $A\setminus V_t$. Possible solution ordering of agents shown in green. Red agents have degree $1$ in $V_t$, blue agents have degree $2$ and violet agents degree $0$. Dashed lines correspond to edges not introduced so far.}
    \label{fig:example_instance_1}
\end{figure}

The following lemma shows that indeed it is enough to prove correctness for every subinstance in the DP to derive the optimal solution. The proof is deferred to the appendix.
\begin{restatable}{lemma}{lemmaOpt}
    For every node $t\in V(\mathbb{T})$ and optimal tour $\opt\subseteq E(G_I)$ that starts with agent $a_s$ and ends with $a_t$, the partial solution $\opt[t]$ is optimal for the induced instance.\label{lem:opttw}
\end{restatable}
\noindent
We now describe how to handle the different node types in the dynamic program.
We will assume that the instance to solve corresponds to a valid specification, otherwise we assign it cost $\infty$.
The proof of correctness for the different node types is deferred to the appendix.

\smallskip
\noindent
Leaf node: At a leaf node $t$, only $a_s,a_t$ are present and no edges exist. The only valid forest is the empty one with cost 0, so we set $D[t, B_t^0,B_t^1,B_t^2,M,\hat{A},\hat{Z}]=\{ \emptyset \}$.

\smallskip
\noindent
Introduce node: Let $t$ be an introduce node with child $t'$ and $X_t = X_{t'} \cup \{v\}$.  
Since $v$ has no incident edges at this point, it must have degree 0 in any valid solution, i.e., $v\in B_t^0$. 
Therefore $v$ does not influence any solution derived for $t'$ and we set $D[t,B_t^0,B_t^1,B_t^2,M,\hat{A},\hat{Z}] = D[t',B_t^0\setminus \{v\}, B_t^1, B_t^2, M, \hat{A},\hat{Z}]$.

\begin{restatable}{lemma}{lemmaIntro}
    Let node $t$ introduce an agent $v$, and let $t'$ be its child. Assume that for all valid combinations the sets $D[t',{D_t^0}',{D_t^1}',{D_t^2}',M', \hat{A}',\hat{Z}']$ have been computed, then we correctly compute $D[t,D_t^0,D_t^1,D_t^2,M,\hat{A},\hat{Z}]$ in time $\OO(1)$.
\end{restatable}

\smallskip
\noindent
Forget node: 
Let $t$ be a forget node with child $t'$ and $X_t = X_{t'} \setminus \{v\}$.
Since $v\notin X_t$ and $v\neq a_s',a_t'$, it must have degree either $0$ or $2$ in any feasible solution, depending on whether $v$ was used in the respective optimal solution. 
Therefore, we check both corresponding entries at the child node $t'$ and choose the one with minimum cost.

\begin{restatable}{lemma}{lemmaForget}
    Let node $t \in V(\mathbb{T})$ remove some vertex $v$, and let $t'$ be its child. 
    Assume that for all valid combinations, the sets $D[t', {D_t^0}', {D_t^1}', {D_t^2}', M', \hat{A}',\hat{Z}']$ have been computed, then we correctly compute $D[t,B_t^0,B_t^1,B_t^2,M,\hat{A},\hat{Z}]$ in time $\OO(1)$.
\end{restatable}

\smallskip
\noindent
Introduce edge node:
Let $t$ be an introduce-edge node with child $t'$ and let $e=\{u,v\}$ be the newly introduced edge. 
If $e$ is not used in the optimal solution, we reuse the solution from $t'$.
Otherwise, we consider the instance where the degrees of $u$ and $v$ are decreased by 1 and consider all valid choices for their next potential npip agent via $\hat{Z}$.
Additionally, we update $M$ accordingly to reflect all possible new pairings that are possible for solutions containing $e$.

\begin{restatable}{lemma}{lemmaIntroE}
Let node $t\in V(\mathbb{T})$ introduce an edge $e=\{u,v\}$, and let $t'$ be its child.
Assume that for all valid combinations, the sets $D[t',{B_t^0}', {B_t^1}', {B_t^2}', M', \hat{A}',\hat{Z}']$ have been computed, then we correctly compute $D[t,B_t^0,B_t^1,B_t^2,M, \hat{A}, \hat{Z}]$ in time 
$\OO((\tw{}+1)\cdot \dm^2)$.
\end{restatable}

\smallskip
\noindent
Join node:
Let $t$ be a join node with children $t_1$ and $t_2$. 
To compute a an optimal solution, we combine compatible solutions from both children while ensuring degrees, matching, psp agents and npip agents correctly align, and take the solution with minimal cost.

\begin{restatable}{lemma}{lemmaJoin}\label{lem:twjoin}
    If node $t \in V(\mathbb{T})$ is a join node with two children $t_1,t_2 \in V(\mathbb{T})$ and we correctly computed all solutions for all child instances for nodes $t_1$ and $t_2$, then we correctly compute $D[t,B_t^0,B_t^1,B_t^2,M,\hat{A},\hat{Z}]$ in time $\OO((\tw{}+2)^{\OO(\dm{}+\tw{})})$.
\end{restatable}

Using these lemmata, we establish the final theorem.
\thmgraphfpt*
\begin{proof}
    First, assume that the intersection graph given by the DDT-SP instance is simple.
    The correctness follows immediately from the correctness of the different node types. 
    The entry $D[r,\{\},\{a_s,a_t\}, \{\}, \{(a_s,a_t)\},\{(a_s,a_s'),(a_t,a_t')\},\{(a_s,u),(a_t,v)\}]$ for some $u\in N_{a_s}$ and $v\in N_{a_t}$ must contain an optimal solution, if $u$ and $v$ correspond to some npip agents for $a_s$ and $a_t$ in an optimal solution.
    
    Let $\tw{}'$ be the treewidth of our original intersection graph instance and $\tw{}$ be the one of our modified instance.
    Since we include agent $a_s,a_t$ in every bag, we have $\tw{}=\tw{}'+2$
    The runtime is dominated by the computation of join nodes, being $\OO((\tw{}+2)^{\OO(\dm{}+\tw{})})$ for each entry.
    Each node has at most $\OO((3^{\tw{}+1})\cdot (\tw{}+1)^{\tw{}+1}\cdot (\tw{}+1)^{2\dm{}}$ many entries, so we have the same time complexity to compute all entries for join nodes.
    Since the number of nodes is bounded by $\OO(k\cdot \tw{})$, the total runtime is $\OO(k\cdot \tw{} \cdot (\tw{}+2)^{\OO(\dm{}+\tw{})})$.
    Preprocessing all distances $d_a(u,v)$ for all agents $a\in A$ and vertices $u,v\in V$ can be done in $\OO(k\cdot n^3)$, yielding an overall runtime of $\OO(n^3\cdot k + k\cdot \tw{} \cdot (\tw{}+2)^{\OO(\dm{}+\tw{})}) = f(\dm{},\tw{}')\cdot \poly(n,k)$.

    At last, the result also extends to the case that the intersection graph is non-simple, as by \cref{thm:unuqueinter} it can be transformed into such an instance while the number of agents increases to at most $n+k$ and the number of vertices increases to at most $n\cdot k$.
\end{proof}

\subsection{A Polynomial-Time Algorithm for Intersection Trees}\label{subsec:treegraph}

We now briefly describe an exact algorithm for computing optimal schedules when the intersection graph of the underlying graph forms a tree. 
While this may seem similar to the setting previously considered, it does not follow directly, as the treewidth may be constant, but the vertex degrees can be unbounded.
For this, we build a layered directed graph $G'$ whose \emph{layers} represent successive handovers:
\begin{itemize}
\item Layer $0$: single vertex $v_s$ representing the source $s$.
\item Layer $i\,(1\le i\le b)$: one vertex for every intersection point of the movement areas of $a_i$ and $a_{i+1}$ (for $i=b$ the region of $a_{b+1}$ is the sink region containing $t$).
\item Layer $b+1$: single vertex $v_t$ for the destination $t$.
\end{itemize}
Edges connect every vertex in layer $i-1$ to every vertex in layer $i$; the weight equals the travel time of agent $a_i$ between these two points.  
Thus each $s$–$t$ path in $G'$ selects exactly one handover point per layer and has length equal to the total delivery time for that sequence of drones.  
Conversely, any schedule that uses these agents in this particular order defines such a path.  
Hence, the shortest path in $G'$ yields the optimal schedule for the given order, which can easily be computed. 

\begin{restatable}{lemma}{lemmaDijkstra}\label{lemma:dijkstra}
Given any ordered sequence of agents $(a_1,\dots, a_b)$, we can compute the optimal route to deliver the package from $s$ to $t$ using the agents in this order in $\OO(b^2\cdot n ^2)$.
\end{restatable}
\noindent

\begin{proof}
    We construct a directed weighted delivery graph $G'=(V',E')$, $w:E'\rightarrow \RR_{\geq 0}$. The vertex set $V'$ consists of $b+1$ layers. 
    In layer $0$ we have a vertex $v_{a_1}$, in layer $b+1$ we have a vertex $v_t$ and in layer $i\in [b]$ we have all intersection points between agent $a_i$ and $a_{i+1}$, i.e., $V_{a_i} \cap V_{a_{i+1}}$.
    In the following, we refer to these vertices as $V_{i}=\{v_1^i,\dots,v_{m_i}^i\}$. 
    We connect $v_{a_1}$ with every vertex $v_j^1$ for $j\in [m_1]$ and set its weight as $d_{a_1}(s,v_j^1)$.
    This connection represents the case that agent $a_1$ transports the package from $s$ to agent $a_2$ via some intersection point $v_j^1$.
    Similarly, we connect a vertex $v_{j}^i$ , $j\in [m_i]$ in layer $i$ with a vertex $v_{j'}^{i+1}$, $j'\in [m_{i+1}]$, if there exists a path in $G_{a_i}$ that connects $v_{j}^i$ with $v_{j'}^{i+1}$.
    For each such edge we again set its weight to the time it takes to get from the first intersection point to the next using agent $a_i$, i.e., we set the weight of the edge as $d_{a_i}(v_{j}^i, v_{j'}^{i+1})$.
    At last, we set the connections from the second last layer $b$ to $v_t$ in the same manner, i.e., we set its weight to $d_{a_b}(v_j^b,t)$, where $j\in [m_b]$.

    We want to argue, that the shortest path in this graph has to correspond to an optimal delivery tour of the agents in the specified order. 
    Let $P_{\opt}\subseteq E$ be the optimal tour and
    $(p_1^{\opt}=s,p_2^{\opt},\dots,p_{b}^{\opt},p_{b+1}^{\opt}=t)$
    be the set of vertices at which an exchange of the package happens, including $s$ and $t$.
    We can see that for each section $p_i^{\opt},p_{i+1}^{\opt}$, 
    the delivery happens using agent $a_i$ to agent $a_{i+1}$.
    Since the delivery points are specified by $p_i^{\opt},p_{i+1}^{\opt}$, we can upper bound the optimal travel time of $a_i$ for this section by $d_{a_i}(p_i^{\opt},p_{i+1}^{\opt})$.
    At last, for the final section $p_b^{\opt},p_{b+1}^{\opt}=t$, the delivery time of agent $a_b$ is upper bounded by $d_{a_b}(p_b^{\opt}, p_{b+1}^{\opt}=t)$.
    
    We can see that the optimal path $(p_1^{\opt}=s,p_2^{\opt},\dots,p_{b}^{\opt},p_{b+1}^{\opt}=t)$ is always a feasible shortest $s-t$ path solution in graph $G'$, and since its cost is at most the cost of the optimal solution we know that the cost of the shortest path tour will be at most the shortest delivery time in the original instance when using the agents in the specified order.
    On the other hand, every feasible path in our graph corresponds to a unique tour using the same set of agents in the same order and exchanging the package at vertices $p_1=s,p_2,\dots,p_b,p_{b+1}=t$.
    If one such tour would have delivery time strictly less than $\opt$, this would contradict its optimality assumption, since the vertices on the shortest path also define a feasible drone delivery routing that delivers the package from $s$ to $t$ using the same set of agents in the given order.
    Therefore it follows, that the shortest path gives us the optimal delivery tour. 

    Our graph $G'$ consists of at most $(n - 1) \cdot b + 2$ vertices and at most $2\cdot n + (b - 2) \cdot n ^2$ edges.
    Using Dijkstra's algorithm with lists, which has runtime $\OO(n ^2)$, we get a final runtime of $\OO(b^2 \cdot n ^2)$.
\end{proof}

One could use this property to calculate the optimal schedule for general graphs in time $\OO(k\cdot k!\cdot k^2\cdot n ^2)$.
It also allows us to solve DDT-SP efficiently, if the underlying simple graph of the intersection graph is a tree, in time $O(k^2\cdot n^2)$. We conclude by proving \cref{thmgraphtree}.

\thmgraphtree*

\begin{proof}
Any delivery schedule corresponds to a path in the intersection graph.
If this graph is a tree, there always exists a unique simple path between any two agents in $G_I$. 
Given the first and last agents used in an optimal schedule, the order of all participating drones can be determined by this path, as each agent appears at most once. Using this order and \cref{lemma:dijkstra}, we can efficiently compute the delivery schedule. 
To identify the first and last agents, we consider all valid combinations of agents that can pick up the package at $s$ and drop it off at $t$. Since each vertex in the graph belongs to at most two agents -- otherwise, a clique of size at least three would be present in $G_I$, contradicting the tree property -- there are at most four such combinations. Simply enumerating all four candidate pairs and evaluating each leads to an overall runtime of $\OO(k^2 \cdot n^2)$.
\end{proof} 
\section{Concluding Remarks and Future Work}
We first proved that the Drone Delivery Problem with selectable starting positions (DDT-SP) is $a(n)$--APX-hard even on path graphs, thereby resolving an open question posed by Erlebach et al.~and Bartlmae et al.~\cite{erlebach22,bartlmae25}. On the algorithmic side, we showed that DDT-SP on a path is fixed-parameter tractable when parameterized solely by the treewidth~$w$ of the agents' intersection graph. Employing more sophisticated techniques and parameterizing by both $w$ and the maximum degree~$\Delta$ of the intersection graph, we further obtained an FPT algorithm for general graphs. When the intersection graph is a tree, a layered-graph formulation yields a polynomial-time exact algorithm.

We believe that the introduction of the intersection graph and its properties is therefore of great value when analyzing the complexity of DDT.

The question of whether DDT on a path is strongly NP-hard remains open for both DDT-FP and DDT-SP. For DDT-SP on a path the complexity for any fixed number of speeds greater than two also remains unresolved. Future work could tighten our parameter dependencies, or establish stronger lower bounds such as $\mathrm{W}$-hardness. Another promising direction is to analyze DDT under smoothed and semi-random input models, where small random perturbations of worst-case instances provide a more realistic measure of expected computational complexity and can guide the development of practical algorithms.

\bibliography{refs}

\newpage
\appendix
\section{Construction and Proof of Theorem 1}
\label{arxivappendixA}
\begin{proof}
We now formally construct the instance $I'$, starting with the agents and then recursively defining their movement intervals:

\begin{itemize}
    \item \textbf{Element agents:}
    \begin{itemize}
        \item $\forall i \in \{1, \dots, n\}, \forall j \in \{1, \dots, k\}: \quad e_{i,j} = (v_j, [c_{i,j,1}, o'_{i,k-1}])$
    \end{itemize}

    \item \textbf{Partition gadgets} $P_j$ for $j \in \{1, \dots, k\}$:
    \begin{itemize}
        \item $P - \frac{P}{k}$ helper agents, each with speed $v = 1$ and movement interval $[c_{1,j,1}, c'_{n,j,p_n}]$
        \item For each $i \in \{1, \dots, n\}$, the critical position $C_{i,j}$:
        \begin{itemize}
            \item $p_i - 1$ partition agents: for each $l \in \{1, \dots, p_i - 1\}: \quad (v^*, [c'_{i,j,l}, c_{i,j,l+1}])$
        \end{itemize}
        \item $n - 1$ P-transition agents: for each $i \in \{1, \dots, n-1\}: \quad (v^*, [c'_{i,j,p_i}, c_{i+1,j,1}])$
    \end{itemize}

    \item \textbf{Object gadgets $O_i$} for $i \in \{1, \dots, n\}$:
    \begin{itemize}
        \item $k - 2$ O-transition agents: for each $j \in \{1, \dots, k - 2\}: \quad (v^*, [o'_{i,j}, o_{i,j+1}])$
    \end{itemize}

    \item \textbf{G-transition agents:}
    \begin{itemize}
        \item Between partition gadgets: for each $j \in \{1, \dots, k - 1\}: \quad (v^*, [c'_{n,j,p_n}, c_{1,j+1,p_1}])$
        \item Between $P_k$ and $O_n$: \quad $(v^*, [c'_{n,k,p_n}, o_{n,1}])$
        \item Between object gadgets: for each $i \in \{n - 1, \dots, 1\}: \quad (v^*, [o'_{i,k-1}, o_{i-1,1}])$
    \end{itemize}
\end{itemize}
We now recursively define the interval boundaries:
\begin{align*}
    \forall i \in [n], \forall j \in [k], \forall l \in [p_i]: \quad
    c_{i, j, l} &= \begin{cases} 
      0 & \text{if } i = 1, j = 1, l = 1 \\
      c_{i-1, j, p_{i-1}}' + v^* & \text{if } i \neq 1, l = 1 \\
      c_{n, j-1, p_n}' + v^* & \text{if } i = 1, l \neq 1 \\
      c_{i, j, l-1}' + v_j & \text{otherwise}
   \end{cases}\\
   \text{and} \quad c_{i, j, l}' &= c_{i, j, l} + v_j'
\end{align*}
\begin{align*}
    \forall i \in \{1, \dots, n\}, \forall j \in \{1, \dots, k\}: \quad
    o_{i, j} &= \begin{cases} 
      c_{n, k, p_n} + v^* & \text{if } i = n, j = 1 \\
      o_{i+1, k-1}' + v^* & \text{if } i \neq n, j = 1 \\
      o_{i, j-1}' + v^* & \text{otherwise}
   \end{cases}\\
   \text{and} \quad o_{i, j}' &= o_{i, j} + 1
\end{align*}

Finally, we define the speeds as follows:
\begin{align*}
    \forall j \in [k]: \quad
    v_j' &= \begin{cases} 
      P^3 \cdot a(n)+1 & \text{if } j = 1 \\
      v_{j-1} \cdot a(n) \cdot P^3 + 1 & \text{otherwise}
   \end{cases}\\
   v_j &= v_j' \cdot a(n) \cdot P^3 + 1 \\
   v^* &= v_k \cdot P^3 \cdot a(n) + 1
\end{align*}

We set $d := P^3$, and now aim to show that $I$ is a “yes”-instance of \textsc{PartitionInto}-$k$ if and only if $c_\mathcal{A}(I') \leq P^3 \cdot a(n)$.
We first show that any schedule $S$ for $I'$ with $c(S) \leq P^3 \cdot a(n)$ cannot cause a C-error, P-error, E-error, G-error, or O-error.

\begin{enumerate}
    \item Suppose $S$ causes a C-error, i.e., a helper agent with speed $v_j'$ traverses a segment of length $v_j$, requiring time $\frac{v_j}{v_j'} > P^3 \cdot a(n)$. Thus, a helper agent can cover at most one gap in a critical position.

    \item Suppose $S$ causes a P-error, i.e., an element agent $e_{i,j}$ with speed $v_j$ traverses a segment of length $v^*$, requiring time $\frac{v^*}{v_j} \geq \frac{v^*}{v_k} > P^3 \cdot a(n)$. Therefore, an element agent can only cover one critical position.

    \item Suppose $S$ causes an E-error, i.e., an element agent $e_{i,j}$ traverses a segment of length at least $v_{j'}'$ for some $j' > j$, requiring time $\frac{v_{j'}'}{v_j} > P^3 \cdot a(n)$. Thus, $e_{i,j}$ cannot assist at any partition gadget $P_{j'}$ with $j' > j$.

    \item Suppose $S$ causes a G-error, i.e., an element agent $e_{i,j}$ with speed $v_j$ traverses a segment of length $v^*$, requiring time $\frac{v^*}{v_j} \geq \frac{v^*}{v_k} > P^3 \cdot a(n)$. Hence, each element agent can be used in at most one gadget.

    \item Suppose $S$ causes an O-error, i.e., an element agent $e_{i,j}$ with speed $v_j$ traverses a segment of length $v^*$, requiring time $\frac{v^*}{v_j} \geq \frac{v^*}{v_k} > P^3 \cdot a(n)$. Thus, in any object gadget $O_i$, each element agent can cover only one gap, and $k - 1$ element agents must be present.
\end{enumerate}
We can now show that $I$ is a “yes”-instance of \textsc{PartitionInto}-$k$ if and only if $c_\mathcal{A}(I') \leq P^3 \cdot a(n)$ holds.

\begin{itemize}
    \item[$\implies$:] If $I$ is a “yes”-instance of \textsc{PartitionInto}-$k$, let $I_1, \dots, I_k$ be a valid partition. We now assign the element agents as follows: $e_{i,j}$ is assigned to the critical position $C_{i,j}$ if $p_i \in I_j$, and to the object gadget $O_i$ otherwise. Since each object $p_i$ appears in exactly one partition set, exactly $k-1$ of the agents $e_{i,1}, \dots, e_{i,k}$ are located at $O_i$, and one at a partition gadget. For each partition gadget $P_j$, the sum of the $p_i$ values of the element agents $e_{i,j}$ at $C_{i,j}$ equals exactly $\frac{P}{k}$. Hence, the $P - \frac{P}{k}$ helper agents can be distributed so that each critical position $C_{i,j}$ is covered by $p_i$ helper agents if $e_{i,j}$ is not present, ensuring that all critical positions are fully covered by either a corresponding element agent or by sufficient helper agents.

A critical section $C_{i,j}$ consists of $p_i$ gaps and the $p_i - 1$ partition agents placed between them. 
An element agent requires at most 
$\frac{p_i \cdot v_j' + (p_i - 1) \cdot v_j}{v_j} < 2p_i$ 
time to traverse it. If the section is instead traversed by helper agents together with the partition agents, 
the required time is 
$\frac{(p_i - 1) \cdot v_j}{v^*} + \frac{p_i \cdot v_j'}{v_j'} < 2p_i$. 
Each of the $n - 1$ P-transition agents requires time $\frac{v^*}{v^*} = 1$, 
so the entire partition gadget $P_j$ can be traversed in time 
$n - 1 + \sum_{i=1}^n 2p_i = n - 1 + 2P$.

We now turn to the object gadgets:
Each object gadget $O_i$ consists of $k - 1$ gaps and $k - 2$ O-transition agents between them. 
Since exactly $k - 1$ element agents are present at $O_i$, each can be assigned to one of the gaps. 
As the gaps are of length $1$, any element agent $e_{i,j}$ can traverse one in time 
$\frac{1}{v_j} < 1$. Thus, the total traversal time for an object gadget $O_i$ is 
$k - 1 + (k - 2) \cdot \frac{v^*}{v^*} < 2k$.

Moreover, there are a total of $n + k - 1$ G-transition agents between the gadgets, each with speed $v^*$ and covering a distance of $v^*$, requiring a total of $n + k - 1$ units of time.

Since we have now covered all gadgets and the regions in between, we obtain a schedule for the entire interval from $s$ to $t$ that takes 
$(n - 1 + 2P) \cdot k + 2n \cdot k + n + k - 1 \leq 3P \cdot P + 2P^2 + 2P \leq P^3 = d$ 
time (for sufficiently large instances). 
Thus, we have $OPT_{I'} \leq P^3$, and therefore $c_\mathcal{A}(I') \leq a(n) \cdot P^3$ since $\mathcal{A}$ is an $a(n)$-approximation algorithm. 
This completes the forward direction of the proof.

\item[$\impliedby$:] Consider a schedule $S$ returned by algorithm $\mathcal{A}$ with $c(S) \leq P^3$. 
As we have previously shown, this implies that $S$ does not incur any C-errors, O-errors, P-errors, E-errors, or G-errors.

Since there are no P- or G-errors, each element agent can be assigned to at most one gadget (and thus to exactly one critical position). 
Because there are no C-errors, each critical position $C_{i,j}$ must be covered either by the corresponding element agent $e_{i,j}$ or by exactly $p_i$ helper agents. 
Moreover, the absence of E-errors implies that an element agent $e_{i,j}$ cannot contribute to a partition gadget $P_{j'}$ with $j' \neq j$.

We have also previously established that, using a (ring) exchange argument, we can transform the schedule such that each element agent $e_{i,j}$ is placed either at $C_{i,j}$ or at the corresponding object gadget $O_i$. 
Finally, since no O-errors occur, each object gadget $O_i$ contains at least $k - 1$ element agents.

If an agent $e_{i,j}$ is located at its corresponding critical position, we assign the object $p_i$ to the partition set $I_j$. 
Since only one of the agents $e_{i,1}, \dots, e_{i,k}$ can be at a critical position -- because the remaining $k-1$ agents must be assigned to the object gadget $O_i$ -- each object is assigned to exactly one partition set.

It remains to show that $\sum_{i \in I_1} p_i = \dots = \sum_{i \in I_k} p_i = \frac{P}{k}$. 
Assume for contradiction that for some partition set $I_j$, we have $\sum_{i \in I_j} p_i < \frac{P}{k}$. 
Then, all agents $e_{i,j}$ with $i \notin I_j$ are not present at their critical positions, meaning that a total of $\sum_{i \notin I_j} p_i$ helper agents are required to cover those critical positions.

However, since $\sum_{i \notin I_j} p_i > P - \frac{P}{k}$, the available $P - \frac{P}{k}$ helper agents are insufficient, and at least one critical position cannot be fully covered -- leading to a C-error. 
This contradicts our earlier result that no C-errors occur, completing the proof of the reverse direction.
\end{itemize}
It is crucial for this proof that \textsc{PartitionInto}-$k$ is strongly NP-hard, as the reduction is only polynomial under this assumption. 
The number of agents in our construction depends directly on the sizes of the objects in the instance $I$, 
so a reduction from \textsc{Partition} would not be feasible -- since in that problem, the object sizes may be superpolynomial.
\end{proof}

\begin{proof}
    We modify the reduction from the proof of Theorem~\ref{thmlinehard} only slightly to show that \textsc{PartitionInto}-$k$ $\leq_p$ DDT on a line with fixed positions using a technique similar to one that Bartlmae et al.~\cite{bartlmae25} also employed.

Let $I$ be an instance of \textsc{PartitionInto}-$k$, and let $I'$ be the corresponding DDT-SP instance constructed in Theorem~\ref{thmlinehard}, where we set $a(n) = 1$. To transform $I'$ into a DDT-FP instance, we assign arbitrary fixed initial positions to all drones and insert a very long interval at the beginning of the instance, along with a single agent. This introduces an initial delay that is long enough to allow every drone to reach any point in its assigned interval. This yields a new instance $I''$ of DDT-FP. Let $d'$ denote the duration of this delay. Once the initial delay agent has delivered the package, the remainder of $I''$ is equivalent to $I'$, as every drone could have reached any position within its range by that time. Hence, $OPT_{I''} = OPT_{I'} + d'$.

Since we already know that $I$ is a "yes"-instance of \textsc{PartitionInto}-$k$ if and only if $OPT_{I'} \leq d$, it follows that $I$ is a "yes"-instance if and only if $OPT_{I''} \leq d + d'$. As the construction involves no super-polynomial distances or speeds, and \textsc{PartitionInto}-$k$ is strongly NP-hard, the proof is complete.
\end{proof}
\section{Fixed-parameter tractability on general graphs}

\subsection{Tree decomposition algorithm for unique intersections}
For the following proofs we will assume that the agent $a_s',a_t'$ as well as the connections $\{a_s,a_a'\},\{a_t,a_t'\}$ are contained in $G_I$ and are always used in the optimal solution, but will never be introduced in the algorithm.

\lemmaOpt*
\begin{proof}
    Let $\opt = (a_1^{\opt}=a_s',a_2^{\opt}=a_s,\dots,a_{b-1}^{\opt}=a_t, a_b^{\opt}=a_t')$ be the globally optimal sequence of agents to deliver the package.
    This corresponds to a path $F^{\opt} = \{e_1^{\opt}=\{a_1^{\opt},a_2^{\opt}\}, e_2^{\opt}=\{a_2^{\opt}, a_3^{\opt}\},\dots,e_{b-1}^{\opt}=\{a_{b-1}^{\opt}, a_b^{\opt}\}\}\subseteq E(G_I)$.
    Let $(t,B_t^0,B_t^1,B_t^2,M, \hat{A},\hat{Z})$ be the instance induced by $\opt$ and node $t$.
    We define $F_1^{\opt}:=F^{\opt}\cap E_t$ the used edges in $E_t$ and $F_2^{\opt}:=F^{\opt}\setminus F_1^{\opt}$ as the remaining edges.
    We can partition $F^{\opt}$ into disjoint sequences of edges, which alternate with edges from $F_1^{\opt}$ and $F_2^{\opt}$ such that every edge from $F^{\opt}$ is used. %
    
    If $F_1^{\opt}=\emptyset$, then global optimal tour does not contain any edges from $G_I[t]$, so $B_t^0=X_t$ and the only feasible solution is the empty set.
    Thus the dynamic programming table will only contain an empty solution in this case which is correct. 
    
    Therefore, assume $\opt$ contains at least one edge from $E_t$. 
    We partition $F^{\opt}$ into alternating sequences of paths:
    $F^{\opt} = \{F_0^2,F_1^1,F_1^2,F_2^1,F_2^2,\dots, F_{z}^1,F_{z}^2\}$ for $z\geq 1$, where $F_i^1\subseteq F_1^{\opt}$ and $F_i^2\subseteq F_2^{\opt}$ holds for every possible $i$.
    We know that each set is nonempty, since the edges $\{a_s',a_s\}$ as well as $\{a_t,a_t'\}$ are never contained in $G_I[t]$ but in $F^{\opt}$, so $\{a_s',a_s\}\in F_0^2$ and $\{a_t,a_t'\}\in F_z^2$.
    Additionally, only given $F_i^j$, the direction can also always be reconstructed from the matching $M$.
    
    If the edges are ordered by the sequence of how they appear in $\opt$, the first agent in the first edge of every such subsequence derived from $F_i^j$ for $i\in \{1,\dots,z\}$ must correspond to some agent in $B_t^1$, same for the last agent in the last edge of every subsequence resulting from $F_i^j$ for $i\in \{0,\dots,z-1\}$.
    For the very first and very last we do not have this condition since these correspond to $a_s',a_t'$ which always belong to $A\setminus V_t$.

    Let $F_i^j=\{\{a_1^{i,j},a_2^{i,j}\},\{a_2^{i,j},a_3^{i,j}\},\dots \{a_{b_{i,j}-1},a_{b_{i,j}}^{i,j}\}\}\}$ be the corresponding ordered set of edges for every subpath.
    By definition we know that for every $F_i^j$ the agents $a_1^{i,j}$ as well as $a_{b_{i,j}}^{i,j}$ belong to $B_t^1$, except for the very first agent in $F_0^2$ and the very last in $F_z^2$.
    This is because the global optimal sequence alternates between edges contained in $E_t$ and $E(G_I)\setminus E_t$.
    Therefore, the endpoints of each intermediate sequences will be incident to one edge from $E_t$ and one from $E(G_I)\setminus E_t$, thus having degree one.

    \begin{figure}
        \centering
        \includegraphics[width=0.8\linewidth]{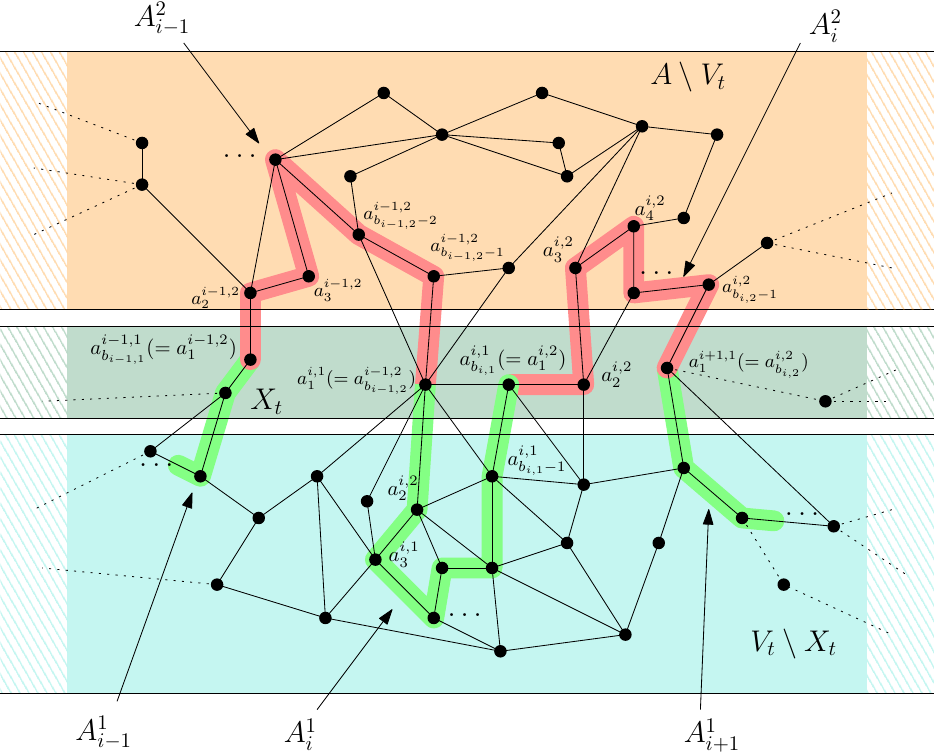}
        \caption{Exemplary partial tour sequences as well as corresponding labeling of agents.}
        \label{fig:partial_paths_sequence}
    \end{figure}

    We represent each set $F_i^1\subseteq E_t$ as a sequence of agents $A_i^1=(a_1^{i,1},\dots,a_{b_{i,1}}^{i,1})$, and each set $F_i^2\subseteq E(G_I)\setminus E_t$ as a sequence of agents $A_i^2=(a_2^{i,2}, \dots, a_{b_{i,2}-1}^{i,2})$ The sequences $A_i^2$ start with the second agent and end with the second last agent to make the set of overall sequences disjoint from each other.
    This allows us to express the full optimal agent sequence as a concatenation of these sequences, i.e., $\opt = (A_0^2,A_1^1,A_1^2,A_2^1,\dots,A_z^1,A_z^2)$. See \cref{fig:partial_paths_sequence} as an example how an intermediate sequence for some instance $G_I[t]$ using this labeling might look like.
    For the initial and final subsequences $A_0^2$ and $A_k^2$ we write $(a_1^{0,2}=a_s',\dots, a_{b_{0,2}-1}^{0,2})$ and $(a_2^{k,2},\dots, a_{b_{k,2}}^{k,2}=a_t')$, since $a_s'$ and $a_t'$ always have degree 1 and the respective edges are not contained in $E_t$.  
    Note that $A_i^2=()$ is possible, if the corresponding forest $F_i^j$ contains only a single edge.

    Since $\opt$ induced the instance $(t,B_t^0,B_t^1,B_t^2,M,\hat{A},\hat{Z})$, we can deduce that each agent $u\in B_t^0$ is either not contained in $\opt$ or is contained in some sequence $A_i^2$ because the two edges with both of its neighboring agents in the sequence are not contained in $E_t$.
    Each agent in $B_t^1$ can not correspond to the starting agent $a_s'$ or the final agent $a_t'$, so it is an intermediate agent with some successor agent and a previous agent in the optimal sequence such that both of these agents are not in the same sequence.
    For every agent in $B_t^2$, both its connections to its preceding and subsequent agent are contained in $E_t$.
    Furthermore, for every $A_i^1$ we have $(a_1^{i,j},a_{b_{i,j}}^{i,j})\in M$ since by definition both endpoints are contained in $B_t^1$ and correspond to the endpoints of some path of edges contained in $E_t$.
    At last, for every pair $A_i^1,A_i^2$ we have $(a_{b_{i,1}}^{i,1}, a_2^{i,2})\in \hat{A}$ since these give the connections from an agent in $B_t^1$ to another agent by an edge not contained in $E_t$.
    Similarly, for every pair $A_i^2,A_{i+1}^1$ we have $(a_1^{i+1,1}, a_{b_{i,2}-1}^{i,2})\in \hat{A}$.
    Additionally, we have for every sequence $A_i^1$ that $\hat{Z}(a_1^{i,1})=a_2^{i,1}$ and $\hat{Z}(a_{b_{i,1}}^{i,1})=a_{b_{i,1}-1}^{i,1}$. 

    Therefore, we can write the cost of $\opt$ as the sum of the individual delivery costs of any agent incurred in their respective sequence over all sequences of agents $(A_0^2,A_1^1,A_1^2,A_2^1,A_2^2,$
    $\dots, A_z^1,A_z^2)$.
    By extending each sequence $A_i^1$ with its respective predecessor and successor agents given by $\hat{A}$, we obtain the extended sequence $\hat{A}(A_i^1)=(\hat{a}_0^{i,1},\hat{a}_1^{i,1}=a_1^{i,1},\dots,\hat{a}_{b_{i,j}}^{i,1}=a_{b_{i,1}}^{i,1}, \hat{a}_{b_{i,1}+1}^{i,1})$, and its cost is defined as $c(\hat{A}(A_i^1))$. 
    Similarly, each sequence $A_i^2$ can be extended using the corresponding psp agents in $\hat{A}$, unless $A_i^2=()$, which occurs if the original forest $F_i^2$ consists of a single edge, in which case the corresponding cost will automatically be accounted for by the sequences $A_i^1$ and $A_{i+1}^1$.

    Hence, the total cost of $\opt$ can be expressed as the sum of costs of all extended sequences, i.e., 
    \begin{align*}
        c(\opt) &= \sum_{i=2}^{\vert \opt\vert-1=b-1}d_{a_i^{\opt}}(D(a_{i-1}^{\opt},a_i^{\opt}),D(a_i^{\opt},a_{i+1}^{\opt})) \\
        &= \sum_{t=2}^{\vert \hat{A}(A_0^2) \vert-1}d_{\hat{a}_t^{i,j}}(D(\hat{a}_{t-1}^{i,j}, \hat{a}_t^{i,j}),D(\hat{a}_t^{i,j},\hat{a}_{t+1}^{i,j}))\\
        &+ \sum_{i=1}^z\sum_{j=1}^2\sum_{t=2}^{\vert \hat{A}(A_i^j) \vert-1}d_{\hat{a}_t^{i,j}}(D(\hat{a}_{t-1}^{i,j}, \hat{a}_t^{i,j}),D(\hat{a}_t^{i,j},\hat{a}_{t+1}^{i,j}))\\
        &= c(\hat{A}(A_0^2)) + \left( \sum_{i=1}^k c(\hat{A}(A_i^1)) + c(\hat{A}(A_i^2)) \right).
    \end{align*}

    Let $F^{D} = \{^{D}F_1^1,^{D}F_2^1,\dots,^{D}F_k^1\}$ denote the set of partial forests output by the dynamic program, and let $^DA^1=\{^{D}A_1^{1},\dots,^{D}A_z^1\}$ be the corresponding ordered sequences of agents from these forest paths, where $^{D}A_i^1=\{^{D}a_1^{i,1},\dots,^{D}a_{l_{i,j}}^{i,j}\}$. 
    By construction, the matching $M$ defines the required endpoints of each path as well as the direction of the sequence, so we can assume that these sequences are aligned with each $A_i^1$ such that their respective starting agent and final agent is identical, i.e., we have $^{D}a_1^{i,1}=a_1^{i,1}$ and $^Da_{l_{i,j}}^{i,j}=a_{b_{i,j}}^{i,j}$ for every $i\in [z]$.   
    
    Assume that the dynamic program returns a solution such that the cost of the extension sequences $\hat{A}(^DA^1)=\{\hat{A}(^DA_1^1),\dots,\hat{A}(^DA_k^1)\}$ is strictly smaller than that of $c(\hat{A}(A^1))$.
    We define the new agent sequence $^DA:=(A_0^2,^DA_1^1,A_1^2,^DA_2^1,\dots, ^DA_k^1,A_k^2)$ resulting from the set of edges $S^D := F_2^{\opt} \cup F^D$.

    We assume $z\geq 1$, otherwise the empty solution will always be optimal and correct.
    We can show by induction that for every $i\in \{1,...,z\}$ there exists a feasible tour given by the order of the agents to deliver the package from $a_s'$ to the intersection vertex $D(a_{b_{i,1}}^{i,1},a_2^{i,2})$.
    For $i=1$, the delivery starts with sequence $A_0^2$ and transports the package to the intersection point $D(a_{b_{0,2}-1}^{0,2}, a_1^{1,1})$. 
    From here the original sequence $A_1^1$ would deliver the package from this point via agents $a_1^{1,1},\dots,a_{b_{1,1}}^{1,1}$ vertex $D(a_{b_{0,2}-1}^{0,2}, a_1^{1,1})$ to agent $a_{b_{1,1}}^{1,1}$ in vertex $D(a_{b_{1,1}}^{1,1}, a_1^{1,2})$.
    Instead of using the delivery route given by $A_1^1$ we use $^DA_1^1$ given by the forest $^DF_1^1$ from the dynamic program.
    Since $^DF_1^1\subseteq E_t$, the route given by $^DA_1^1$ must be a feasible traversal of the agents, i.e., does not use any agents from $A\setminus V_t$ which might already be in use, and every agent is only used at most once in $V_t$.
    Because $\{a_1^{1,1},a_{b_{1,1}}^{1,1}\}\in M$, the sequence $^DA_1^1$ must correspond to a delivery tour for transporting the package from the starting agent $a_1^{1,1}$ to the final agent $a_{b_{1,1}}^{1,1}$. 
    Since $\{a_{b_{1,1}}^{1,1},a_1^{1,2}\}\in \hat{A}$, the cost of this traversal corresponds to the cost of delivering the package to the unique intersection point $D(a_{b_{1,1}}^{1,1},a_1^{1,2})$.
    
    Now the induction step follows in an identical fashion as the previous base case, as one can first expand the already assumed traversal tour by the sequences given by $A_i^2$ and then use  $^DA_{i+1}^1$ by the same reasons as before.
    
    At last, we can see that the cost of this sequence has to be strictly smaller than $c(\opt)$, since we have
    \begin{align*}
        c(\opt) &= c(\hat{A}(A_0^2)) + \left( \sum_{i=1}^k c(\hat{A}(A_i^1)) + c(\hat{A}(A_i^2)) \right)\\
        &< c(\hat{A}(A_0^2)) + \left( \sum_{i=1}^k c(\hat{A}(^DA_i^1)) + c(\hat{A}(A_i^2)) \right) = c({^DA}).
    \end{align*}    
\end{proof}
\lemmaIntro*
\begin{proof}
    Let $F\subseteq E_t$ be any feasible forest for the current instance $(t,B_t^0,B_t^1,B_t^2,M,\hat{A},\hat{Z})$.
    Assume $S=(A_1,A_2\dots,A_b)$ where $A_i=(a_1^i,\dots, a_{l_i}^i)$ corresponds to the order given by the simple paths of agents in $S$. 
    Since in $G_I[t]$ no edge can be incident to $v$, the degree of $v$ in $F$ must be zero, and agent $v$ can not appear in any of the sequences $A_i$. 
    If $v\notin B_t^0$, no feasible solution can exist and we have $D[t,B_t^0,B_t^1,B_t^2,M,\hat{A},\hat{Z}]=\{\}$.
    Otherwise, if $v\in B_t^0$, since agent $v$ can not appear in any sequence $A_i$, it does not contribute in the cost of $c(\hat{A}(S))$.
    Therefore, the optimal solution for the previous node $t'$ with the same matching requirement $M$, psp agent requirement $\hat{A}$ and npip agent requirement $\hat{Z}$ must also be optimal for the current instance. 
    Since all entries for the child node $t'$ are assumed to be correct, by setting $D[t,B_t^0,B_t^1,B_t^2,M,\hat{A},\hat{Z}] = D[t',B_t^0\setminus\{v\}, B_t^1,B_t^2,M,\hat{A},\hat{Z}]$ we get the optimal solution for $(t,B_t^0,B_t^1,B_t^2,M,\hat{A},\hat{Z})$.
\end{proof}
\lemmaForget*
\begin{proof}
    Let $F\subseteq E_t$ be an optimal forest for the given instance, and let $S=(A_1,\dots,A_b)$ denote the ordered agent sequences given by $F$ and $M$.
    First, assume that $v$ is not contained in the optimal solution. 
    Then $v$ has degree $0$ in $F$, and the solution $F$ is also optimal for the instance $(t',B_t^0\cup \{v\}, B_t^1,B_t^2,M,\hat{A},\hat{Z})$, since otherwise the corresponding solution would also be an improvement to the solution given by $F$.

    Now assume that $v$ is contained in the solution $F$. 
    Then $v$ must have degree $2$, so it appears as an intermediate agent in some sequence $A_i=(a_1^i,\dots a_j^i=v,\dots,a_{l_i}^i)$.
    Because $v$ has degree $2$, we must have $1<j<l_i$.
    Since $t$ is a forget node, the edge set $E_{t'}$ must contain the edges $\{a_{j-1}^i,a_j^i\},\{a_j^i,a_{j+1}^i\}$.
    Therefore, agent $v$ has degree $2$ in the child instance.
    This means that the solution must also be contained in the entry $D[t', B_t^0,B_t^1,B_t^2\cup \{v\},M,\hat{A},\hat{Z}]$.

    By considering both solutions and taking the one with smallest cost, we must obtain an optimal solution for the current instance. 
\end{proof}
\lemmaIntroE*
\begin{proof}
    Let $F\subseteq E_t$ be the optimal solution for the given instance and let $S=(A_1,\dots,A_b)$ be the corresponding agent sequences given by the paths in $F$.
    If $e\notin F$, then the solution $F$ must also be optimal for the instance $(t',B_t^0,B_t^1,B_t^2,M,\hat{A},\hat{Z})$, so we can choose the solution saved in $D[t,B_t^0,B_t^1,B_t^2,M,\hat{A},\hat{Z}]$.

    Now assume $e\in F$.
    This implies $\hat{A}(u)\neq v$ and $\hat{A}(v)\neq u$, since otherwise some agent would be visited multiple times.
    Let $F^{-e}:=F\setminus \{e\}$ be the solution without $e$.
    Assume first $u$ and $v$ have degree two in $F$, i.e., $u,v\in B_t^2$.
    This means there exists some sequence $A_i$ such that $A_i=(a_1^i,\dots, a_j^i=u,a_{j+1}^i=v,\dots, a_{l_i}^i)$ with $1<j<j+1<l_i$.
    This implies the matching $(a_1^i,a_{l_i}^i)\in M$, the npip agents $\hat{Z}(a_1^i)=a_2^i$, $\hat{A}(a_{l_i}^i)=a_{l_i-1}^i$, and we have for both endpoints $a_1^{i},a_{l_i}^i$ their respective psp agent given by $\hat{A}$.
    Removing edge $e$ corresponds to splitting this sequence in two new sequences $A_i^1,A_i^2$ with $A_i^1=(a_1^i,\dots,a_j^i=u)$ and $A_i^2=(a_{j+1}=v, \dots, a_{l_i}^i)$. 
    This gives us $b+1$ new sequences $S^{-e}=(A_1'=A_1,\dots, A_i'=A_i^1,A_{i+1}'=A_i^2,\dots, A_{b+1}'=A_b)$.
    We consider the new matching $M'$ which would correctly match the endpoints of every sequence in $S^{-e}$, i.e., $M':=(M\setminus \{(a_1^i,a_{l_i}^i)\}) \cup \{(a_1^i,u), (v,a_{l_i}^i)\}$.
    Additionally we extend the predecessor/successor matching to $\hat{A}'=\hat{A} \cup \{\{a_j^i, a_{j+1}^i\},\{a_{j+1}^i, a_j^i\}\}$
    and the matching given by $\hat{Z}$ to $\hat{Z}'=\hat{Z}\cup \{\{a_j^i,a_{j-1}^i\},\{a_{j+1}^i,a_{j+2}^i\}\}$.
    By these adaptations we can see that the solution $F^{-e}$ is a feasible solution for the instance $G_I[t']$ with adapted matching $M'$, $\hat{Z}'$ and $\hat{A}'$, i.e., it is feasible for the instance $(t',B_t^0, B_t^1\cup \{u,v\}, B_t^2\setminus \{u,v\}, M', \hat{A}',\hat{Z}')$.
    Assume $F^{-e}$ were not not optimal for this instance.
    Then there exists another solution $F_D$ with corresponding sequences of agents $S_D=(^DA_1,\dots, ^DA_{b+1})$ where $^DA_i=(^Da_1^i,\dots, ^Da_{m_i}^i)$ and $^Da_1^i={a_1^i}$ and $^Da_{m_i}^i = {a_{l_i}^i}$ such that $\sum_{i=1}^{b+1} c(\hat{A}(^DA_i)) < \sum_{i=1}^{b+1} c(\hat{A}(A_i'))$.
    Adding edge $e$ to solution $F_D$ results in a solution $F_D^{+e}:=F_D\cup \{e\}$. 
    This solution is feasible for $(t,B_t^0,B_t^1,B_t^2,M,\hat{A},\hat{Z})$, since the degrees of every vertex in $X_t\setminus \{u,v\}$ did not change and agents $u,v$ were previously required to have degree $1$ before edge $e$ was added, so their degree is $2$ in $F_D^{+e}$ as required.
    Additionally, the solution $S_D$ contains $b+1$ disjoint paths that fulfill the matching $M'$.
    Every matching except for $(a_1^i,a_{l_i}^i)$ are also contained in $M$, since these matchings were not modified for $M'$.
    The matching $(a_1^i, a_{l_i}^i)$ is fulfilled by the two sequences $^DA_{j_1},^DA_{j_2}$ where $^Da_1^{j_1}=a_1^i$ and $^Da_{m_{j_2}}^{j_2}=a_{l_i}^i$ are the two sequences which, by $M'$ need to have $^Da_{m_{j_1}}^{j_1}=u$ and $^Da_1^{j_2=v}$ as the respective last and first agent in the sequence, so together they form a valid sequence (by adding edge $e$ into the respective two forest paths) to fulfill matching $(a_1^i, a_{l_i}^i)$.
    The desired npip agent requirement $\hat{Z}$ is also fulfilled, since these are also part of $\hat{Z}'$ in any case.
        
    At last, introducing the edge $e$ into the solution does not change cost, since every agent still traverses the same path (we only merge two sequences together, for which each agent already paid the corresponding travel time).
    Therefore, we get the same cost $\sum_{i=1}^{b+1} c(\hat{A}'(^DA_i)) < \sum_{i=1}^{b+1} c(\hat{A}'(A_i')) = \sum_{i=1}^b c(\hat{A}(A_i))=c(\opt)$ which is a contradiction in $\opt$ having minimal cost for the current instance. 

    Now assume both $u$ and $v$ have degree $1$ in $F$, i.e, $u,v\in B_t^1$.
    We consider the case $(u,v)\in M$, the other case follows analogously.
    Since both are contained in $B_t^1$, the condition $e\in F$ implies that this is the only edge that is incident to $u$ and $v$ and therefore a path in $E_t$ consisting of a single edge.
    This means they define a unique ordered sequence $\hat{A}(A_i) = (\hat{A}(u), u, v, \hat{A}(v))$ with cost $c(\hat{A}(A_i)) = d_{u}(D(\hat{A}(u),u)), D(u, v)) + d_v(D(u,v),D(v, \hat{A}(v)))$. 
    Additionally, the npip agent requirement must correspond to $\hat{Z}(u)=v$ and $\hat{Z}(v)=u$ in this case, since otherwise $e$ can not be the unique path between $u$ and $v$.
    Therefore, the solution $F^{-e}$ is feasible for the instance $(t', B_t^0\cup \{u,v\}, B_t^1\setminus \{u,v\}, B_t^2, M\setminus \{(u,v)\}, \hat{A}\setminus \{(u,\hat{A}(u)), (v,\hat{A}(v))\},\hat{Z}\setminus \{(u,v),(v,u)\})$ and has cost $c(\hat{A}(S)) - c(\hat{A}(A_i))$.
    If a solution for this instance would have less cost than $F^{-e}$, then this solution together with edge $e$ would correspond to a feasible solution for $(t,B_t^0,B_t^1,B_t^2,M,\hat{A},\hat{Z})$ and has less cost since only the cost $c(\hat{A}(A_i))$ gets added to this solution.
    
    At last, assume that $u$ has degree 1 and $v$ has degree 2, the other case follows analogously.
    Similar to before, in $F$ there now must exist a sequence $A_i$ such that $A_i=(a_1^i=u,a_2^i=v,\dots,a_{l_i}^i)$.
    This implies that $M$ contains the matching $(u,a_{l_i}^i)$ and $\hat{Z}$ contains $(u,v)$ as well as $(a_{l_i}^i,a_{l_i-1}^i)$.
    Removing edge $e$ yields a solution $F^{-e}$ which now instead contains the sequence $A_i^1=(v,\dots,a_{l_i}^i)$. 
    We adapt $M'=(M\setminus \{(u,a_{l_i}^i)\}) \cup \{(v,a_{l_i}^i)\}$ and $\hat{A}'=(\hat{A}\setminus \{(u,\hat{A}(u))\})\cup \{(v, u)\}$ and $\hat{Z}'=(\hat{Z}\setminus \{(u,v)\})\cup \{(v,a_{3}^i)\}$.
    We can see that for the child instance $(t',B_t^0\cup \{u\}, (B_t^1\setminus \{u\}) \cup \{v\}, B_t^2\setminus \{v\}, M',\hat{A}', \hat{Z}')$, the solution $F^{-e}$ is feasible and has cost $c(\hat{A}(S))-d_u(D(\hat{A}(u),u), D(u, v))$, since the partial tour $A_i$ does not include agent $u$ anymore, and therefore $u$ does not contribute any cost.
    Now, if the solution $F^{-e}$ would not be optimal for this instance, we can again show that the solution which strictly improves it must also strictly improve solution $F$ by adding $e$ to it.
    Let $F_D$ be the solution with strictly less cost for this instance and let $S_D=(^DA_1,\dots, {^DA}_b)$ be the corresponding sequence of agents. 
    For any matching $(x,y)$ with $x,y\neq u,v$, we can see that there exists a corresponding tour, given by $F_D$, that traverses from agent $x$ to $y$ and the given psp agents as well as the npip agents are still identical, so these tours are still feasible and correctly define the cost from the same intersection vertex $D(\hat{A}(x),x)$ to $D(y,\hat{A}(y))$.
    At last, when adding edge $e$ to the solution $F_D$, since it defines some agent sequence $A_i'=({a_2^i}'=v,\dots,{a_{l_i}^i}'=a_{l_i}^i)$ from $D({\hat{A}}'(v)=u,v)$ to $D(a_{l_i}^i,\hat{A}(a_{l_i}^i))$, this gives us a new feasible sequence of agents $({a_1^i}'=u,{a_2^i}'=v,\dots, {a_{l_i}^i}'=a_{l_i}^i)$.
    By adding the cost $d_u(D(\hat{A}(u),u), D(u, v))$ to $c(\hat{A}'(S_i))$, this in turn corresponds to a feasible solution that first brings the package from $D(\hat{A}(u),u)$ to $D(u, v)$ using agent $u$, then to agent $\hat{Z}(u)$, then continues with the remaining agents from sequence $A_i'$ and brings it first to $\hat{Z}(a_{l_i}^i)$ and finally to $a_{l_i}^i$.
    Therefore, the solution $F_D \cup \{e\}$ is feasible for the original instance.  
    Since $c(S_D)-d_u(D(\hat{A}(u),u), D(u, v))<c(S)-d_u(D(\hat{A}(u),u), D(u, v))$, it also must correspond to a solution with strictly less cost than $\opt$, which is a contradiction.

    As for the runtime, by our case distinction we can see that in the worst case both $u$ and $v$ have degree 2.
    In this case we need to guess which tour $A_i$ gets split up, because $e$ was removed.
    There are at most $\vert B_t^2\vert \leq \tw{}+1$ many vertices of degree $2$, so at most $(\tw{}+1)/2$ many matchings that we can potentially consider for splitting up. 
    For any such choice $(a,a')\in M$, we have two choices for which agent appears first in the resulting sequence, either $u$ or $v$.
    Assume we include the matchings $(a,u)$ and $(v,a')$.
    For $u$ we guess its adjacent agent in the respective neighborhood in $G_I[t]$ such that the corresponding agent is either $a$ or some agent in $N_u[t]$ which is not contained in $B_t^1$, i.e., we choose some agent $a''\in N_u' = (N_u[t]\setminus B_t^1)\cup \{a\}$.
    We do the same analogously for agent $v$ and its neighborhood $N_v'=(N_v[t]\setminus B_t^1)\cup \{a'\}$ in this case.
    We can upper bound the size of $N_u'$ and $N_v'$ by $\dm$.
    Together, this gives a runtime of $\mathcal{O}((\tw{}+1)\cdot \dm^2)$.
\end{proof}

For \cref{lem:twjoin} we will assume that no edge $\{u,v\}$ with $u,v\in X_t$ can exist, as this would imply that this edge would have been introduced at least twice, once in the subtree rooted at $t_1$ and once in the subtree rooted at $t_2$. 

\lemmaJoin*

\begin{proof}
    Let $F\subseteq E_t$ be the optimal solution for the given instance and let $S=(A_1,\dots,A_b)$ with $A_i=(a_1,\dots,a_{b_i})$ be the corresponding ordered agent sequences given by $F$ and $M$.
    If $F\subseteq E_{t_1}$,
    then this solution will also be feasible for the instance $(t_1,B_t^0,B_t^1,B_t^2,M,\hat{A},\hat{Z})$ and since the solution in $D[t_1,B_t^0,B_t^1,B_t^2,M,\hat{A},\hat{Z}]$ is correct, its cost is upper bounded by the cost of $F$, so we compute an optimal solution.
    The same holds analogously if $F\subseteq E_{t_2}$. 

    Now assume that neither of those two cases did apply, so the optimal forest contains edges from $E_{t_1}$ and $E_{t_2}$.
    We now argue that there has to exist a solution pair, given by some $D[t_1,{B_t^0}', {B_t^1}', {B_t^2}', M', \hat{A}',\hat{Z}']$ and $D[t_2,{B_t^0}'',{B_t^1}'', {B_t^2}'',M'',\hat{A}'',\hat{Z}'']$, that, if combined, yields a feasible solution for the current instance and its cost is upper bounded by the one of the optimal solution.

    Since no edge $\{u,v\}$ with $u,v\in X_t$ can exist in $E_t$, as otherwise it would need to be introduced at least twice, any edge will always have been introduced in some (child)bag $t_x$ for $x\in \{1,2\}$, so we can subdivide each subsequence $A_i$ into disjoint sequences, depending if the corresponding sequence is resulting from edges in $E_{t_1}$ or $E_{t_2}$.
    Let $F^1=F\cap E_{t_1}$ and $F^2=F\cap E_{t_2}$ be the set of edges contained in their respective subgraphs with $F^1\cup F^2=F$ and $F^1\cap F^2=\emptyset$ and let $A^1,A^2$ be their respective induced ordered agent sequences.

    An original sequence $A_i=(a_1^i,\dots,a_{l_i}^i)$ can now be subdivided into multiple sequences $(A_1^{i,1},A_1^{i,2},A_2^{i,1},\dots,A_{b_{i,j}}^{i,2})$ such that each sequence $A_f^{i,j}=(a_1^{i,j,f},\dots,a_{l_{i,j,f}}^{i,j,f})$ for $f\in [b_{i,j}]$ is formed by the respective edges in $F^j$ that originally made up the whole sequence $A_i$.
    In this case the sequences are not disjoint, because for every adjacent sequence pair $A_f^{i,j},A_{f'}^{i',j'}$ we have $a_{l_{i,j,f}}^{i,j,f} = a_1^{i',j',f'}$.
    Therefore, we later subtract the cost of any agent of degree 1 in the induced instance, since otherwise this cost gets counted exactly twice.
    
    Assume in the following that each sequence $A_i$ is first formed by some edges in $E_{t_1}$ which make up the sequence $A_1^1$ and ends with an agent sequence formed by some edges in $E_{t_2}$.
    The arguments that some sequence starts with some set of edges in $E_{t_2}$ or ends with some edges in $E_{t_1}$ follows analogously.
    By the properties of our instance we know that this implies $(a_1^{i,1,1},a_{l_{i,2,b_{i,2}}}^{i,2,b_{i,2}})\in M$. 
    
    \noindent
    Let $(F_1^{i,j},F_2^{i,j}\dots,F_{b_{i,j}}^{i,j})$ be the disjoint paths that make up the sequences $(A_1^{i,j},A_2^{i,j},\dots,A_b^{i,j})$ for $j\in \{1,2\}$.
    We can see that, in case of node $t_1$ and edge set $E_{t_1}$, the set of disjoint paths $(F_1^{1,1},F_2^{1,1},\dots,F_{b_{1,1}}^{1,1},F_1^{2,1},\dots,F_{b_{b,1}}^{b,1})$
    makes up the set of sequences $(A_1^{1,1},A_2^{1,1},\dots,A_{b_{1,1}}^{1,1},A_1^{2,1},$
    $\dots,A_{b_{b,1}}
    ^{b,1})$.
    
    These disjoint paths are all contained in $E_{t_1}$ and make up $F^1$.
    Setting $\hat{A}'(a_1^{i,1,1})=\hat{A}(a_1^{i,1,1})$ and $\hat{Z}'(a_1^{i,1,1})=\hat{Z}'(a_1^{i,1,1})$
    for every $i\in [b]$ sets the same predecessor and successor for the original sequence $A_i$. 
    Additionally, we set for every sequence $A_m^{i,1}$ the psp agents of the first and last agent to the last or first agent in the adjacent sequence and the npip agent to the neighboring agent in the own sequence, as these connections are not part of $E_{t_1}$.
    This means we set $\hat{A}'(a_1^{i,1,m})=a_{l_{i,2,m-1}}^{i,2,m-1}$
    and $\hat{A}'(a_{l_{i,1,m}}^{i,1,m})=a_1^{i,2,m}$ and $\hat{Z}'(a_1^{i,1,m})=a_2^{i,1,m}$ and $\hat{Z}'(a_{l_{i,1,m}}^{i,1,m})= a_{l_{i,1,m}-1}^{i,1,m}$ for any possible $m\in [b_{i,1}]$.
    Additionally, we set the corresponding matching $M'(a_1^{i,1,m})=a_{l_{i,1,m}}^{i,1,m}$ for every possible $i,j,m$.
    We do the same analogously for all sequences $A_m^{i,2}$.

    We can see that the solution $F^1$ is feasible for the instance $(t_1,{B_t^0}',{B_t^1}',{B_t^2}',M',\hat{A}',\hat{Z}')$ By our construction, where ${B_t^0}'=B_t^0$, ${B_t^1}'$ is the set of degree one vertices in $X_t$ given by $F^1$ which is the corresponding first and last agent in any sequence $A_f^{i,1}$ and ${B_t^2}'$ is the set of degree two vertices in $X_t$ given by $F^1$.
    Similarly, the solution $F^2$ is feasible for the instance $(t_2,{B_t^0}'',{B_t^1}'',{B_t^2}'',M'',\hat{A}'',\hat{Z}'')$, where the respective sets and matchings are derived in an identical manner.

    Now assume that the solution given by $F^1$ is not optimal for its corresponding instance, the other case follows analogously. 
    Let $S_{D_1}^i=({^{D_1}}A_1^{i,1},{^{D_1}}A_2^{i,1},\dots,{^{D_1}}A_{b_{i,1}}^{i,1})$
    be the corresponding sequence given by a different set of disjoint paths $F_{D_1} = \{{^{D_1}}F_1^{i,1},{^{D_1}}F_2^{i,1},\dots,{^{D_1}}F_{b_{i,1}}^{i,1}\}$ that would make up the original edges in $E_{t_1}$ for sequence $A_i$.
    Let $^{D_1}A_m^{i,1}=(^{D_1}a_1^{i,1,m},\dots,{^{D_1}a}_{b_{i,1,m}}^{i,1,m})$.
    By definition of our instance $(t_1,{B_t^0}',{B_t^1}',{B_t^2}',M',\hat{A}',\hat{Z}')$, the endpoints in any sequence $^{D_1}A_m^{i,1}$ will match the one given by their respective partner sequence $A_m^{i,1}$,
    i.e., $^{D_1}a_1^{i,1,m}=a_1^{i,1,m}$ and $^{D_1}a_{b_{i,1,m}}^{i,1,m}=a_{b_{i,1,m}}^{i,1,m}$.
    Additionally, since the psp agents and npip agents given by $\hat{A}'$ and $\hat{Z}'$ match as well, we know that for any original triplet sequence $A_{f-1}^{i,2},A_f^{i,1},A_{f}^{i,2}$, the sequence $A_f^{i,1}$ will take the package starting with agent $a_1^{i,1,f}= {^{D_1}a_1^{i,1,f}}$ at vertex $D(\hat{A}'(a_1^{i,1,f}),a_1^{i,1,f})=D(\hat{A}'(^{D_1}a_1^{i,1,f}),{^{D_1}}a_1^{i,1,f})$ and delivers it to agent $a_{l_{i,1,f}}^{i,1,f}={^{D_1}}a_{b_{i,1,f}}^{i,1,f}$ at vertex $D(\hat{A}'(a_{l_{i,1,f}}^{i,1,f}),a_{l_{i,1,f}}^{i,1,f})=D(\hat{A}'(^{D_1}a_{b_{i,1,f}}^{i,1,f}),{^{D_1}}a_{b_{i,1,f}}^{i,1,f})$.
    This also holds true for the psp agent of the first agent in first sequence $A_1^{i,1}$.
    Since we are guaranteed that replacing every sequence $A_f^{i,1}$ by ${^{D_1}A}_f^{i,1}$ does still correspond to a valid sequence of agents where none is taken multiple times and the delivery of the package happens at the same intersection points, we know that the union of the solutions $A'_i=({^{D_1}A}_1^{i,1},A_1^{i,2},{^{D_1}A}_2^{i,1},\dots,A_{b_{i,2}}^{i,2})$ for every $i\in [b]$ will still deliver the package like $A_i$ from the intersection point $D(\hat{A}(a_1),a_1)$ to $D(a_{b_i},\hat{A}(a_{b_i}))$ and starts the sequence with agent $a_1$ and ends with agent $a_{b_i}$.
    Therefore, the union of the corresponding set of edges $F'_i=({^{D_1}}F_1^{i,1},F_1^{i,2},{^{D_1}}F_2^{i,1},\dots,F_{b_{i,2}}^{i,2})$ will be a feasible solution for the original instance $(t,B_t^0,B_t^1,B_t^2,M,\hat{A})$.

    Consider an arbitrary sequence $A_i$ that was divided into sequences $(A_1^{i,1},A_1^{i,2},A_2^{i,1},\dots,$
    $A_{b_{i,1}}^{i,1},A_{b_{i,2}}^{i,2})$.
    We can see that every agent $a\in A_i$ appears only twice in this sequence, if it is the first agent in some sequence $A_m^{i,2}$, since it will also be the last agent in the previous sequence $A_m^{i,1}$, i.e., we have $a=a_{l_{i,1,m}}^{i,1,m}=a_1^{i,2,m}$. 
    Analogously, each agent $a'=a_1^{i,1,m}=a_{i,2,l_{i,2,m-1}}$ will be counted twice as well.
    Both types of agents will always be contained in ${B_t^1}'$, since they are adjacent to some agent from the other sequence (except if they define the very first or last agent in the sequence $A_i$.
    
    Adding the respective costs of both sequences together would result in the cost of these agents being counted twice.
    Therefore, we need to subtract the cost of the agents which were counted twice from the final cost.
    Since the agents that are counted twice are only from ${B_t^1}'$, their respective psp agents and npip agent are defined, so their influence in the cost is always fixed and independent of the actual sequences $A_m^{i,1}$ and $A_m^{i,2}$.
    This means the cost of any sequence $c(\hat{A}(A_i))$ can be expressed as the sum of costs of both sequences minus the degree 1 agents that appear by the division into $(A_1^{i,1},A_1^{i,2},A_2^{i,1},\dots,A_{b_{i,j}}^{i,2})$ which correspond to the vertices in ${B_t^1}'\cap A_i$, i.e., 
    \[\sum_{f=1}^{b_{i,1}}c(\hat{A}'(A_f^{i,1})) + \sum_{f=1}^{b_{i,2}}c(\hat{A}''(A_f^{i,2}))  - \sum_{a\in {B_t^1}'\cap A_i}c\left(\hat{Z}'(a), a, \hat{A}'(a)\right).\]
    
    We can thus write the final cost $c(\hat{S})$ as the sum of the respective parts of edges from $E_{t_1}$ and $E_{t_2}$ minus the cost of the agents appearing twice, i.e., 
    \[c(\hat{A}(S))=\sum_{i=1}^b \sum_{f=1}^{b_{i,1}}c(\hat{A}'(A_f^{i,1})) + \sum_{i=1}^b \sum_{f=1}^{b_{i,2}}c(\hat{A}''(A_f^{i,2})) - \sum_{a\in {B_t^1}'}c\left(\hat{Z}'(a), a, \hat{A}'(a)\right).\]
    Since we assumed that $F^1$ was not optimal for its corresponding instance, we have 
    
    $\left(\sum_{f=1}^{b_{i,1}}c(\hat{A}'(A_f^{i,1}))\right) > \left(\sum_{f=1}^{b_{i,1}}c(\hat{A}'(^{D_1}A_f^{i,1}))\right)$, so we also have 
    \begin{align*}
        c(\hat{A}(S)) > &\sum_{f=1}^{b_{i,1}}c(\hat{A}'(^{D_1}A_f^{i,1})) + \sum_{i=1}^b \sum_{f=1}^`{b_{i,2}}c(\hat{A}''(A_f^{i,2}))- \sum_{a\in {B_t^1}'}c\left(\hat{Z}'(a), a, \hat{A}'(a)\right)\\ = & \;c(\hat{A}(S'))
    \end{align*}
    which is a contradiction in the optimality of $S$.
    The other cases follow analogously to the previous one.

    As for the runtime, we need to consider every possibility in which a solution for $(t,B_t^0,B_t^1,B_t^2,M,\hat{A},\hat{Z})$ could be partitioned into respective subpaths of edges.
    As a procedure to check if some instance pair $(t_1,{B_t^0}',{B_t^1}',{B_t^2}',M',\hat{A}',\hat{Z}')$, $(t_2,{B_t^0}'',{B_t^1}'', {B_t^2}'',M'',\hat{A}'',\hat{Z}'')$ makes up a feasible solution, we first check if for any agent $u\in {B_t^1}'\cap {B_t^1}''$ the corresponding psp agent in one instance is the npip agent in the other instance and the other way around, i.e., $\hat{A}'(u)=\hat{Z}''(u)$ and $\hat{Z}'(u)=\hat{A}''(u)$.
    If this is the case we know that both solutions will have the same cost assumption for these agents.
    We then create an auxiliary (multi)graph $G'=(X_t,E_1'\cup E_2')$, where $E_1'$ consists of an arbitrary directed set of edges such that the degree requirements given by the corresponding sets ${B_t^0}',{B_t^1}',{B_t^2}'$ for all vertices in $X_t$ are fulfilled and for any matching pair $(u,v)\in M'$ there exists a unique simple path in $(X_t,E_1')$ from $u$ to $v$. 
    We define the set $E_2'$ in the same manner, using the sets ${B_t^0}'',{B_t^1}'',{B_t^2}''$ and matching $M''$.
    If the resulting graph $G'$ is simple and the degree conditions given by $B_t^0,B_t^1,B_t^2$ are fulfilled, we check if for every matching pair $(u,v)\in M$ there exists a simple path in $G'$ that connects both endpoints.
    At last, we check for every vertex $u\in B_t^1$, if $u\in {B_t^1}'$ or $u\in {B_t^1}''$ and for the corresponding set if the psp and npip agent specifications given by $\hat{A}$ and $\hat{Z}$ match as well, assuming all pairwise shortest path distances between the agents and their unique intersection points are already computed.

\noindent
    Given two configurations $(t_1,{B_t^0}',{B_t^1}',{B_t^2}',M',\hat{A}',\hat{Z}')$ and $(t_2,{B_t^0}'',{B_t^1}'',{B_t^2}'',M'',\hat{A}'',\hat{Z}'')$, one can check in time $\text{poly}(\tw{})$, if the two configurations are a combination candidate for the original instance $(t,B_t^0,B_t^1,B_t^2,M,\hat{A}, \hat{Z})$ and in the same runtime evaluate the cost of the combined solution as well.
    
    There are at most $(3^{\tw{}+1})$ many possibilities for choosing configurations $({B_t^0}',{B_t^1}',{B_t^2}')$,
    at most $(\tw{}+1)^{\tw{}+1}$ many possibilities for choosing the matching $M'$ and at most $(\tw{}+1)^{\dm}$ many possibilities to choose the psp agent specification $\hat{A}'$ and npip agent $\hat{Z}'$.
    This gives us a total of $\OO((3^{\tw{}+1})\cdot (\tw{}+1)^{2\dm+\tw{}+1})$ many specifications. 
    Therefore there are $\OO((3^{2\tw{}+2})\cdot (\tw{}+1)^{4\dm+2\tw{}+2})=\OO((\tw{}+2)^{\OO(\dm{}+\tw{})})$ many combined solutions which we can consider, where each check takes polynomial time and we choose the solution with smallest cost.
\end{proof}

\end{document}